%% file: main.tex
\documentclass[11pt]{article}
\usepackage{times}
\usepackage{hyperref}

\hypersetup{colorlinks,citecolor=blue}
\usepackage[dvips, paper=letterpaper, top=1in, bottom=.75in, left=0.95in, right=0.95in, nohead, includefoot, footskip=.25in]{geometry}

\usepackage{color}
\usepackage{amsmath,amssymb,amsthm,amsfonts,latexsym,bbm,xspace,graphicx,float,mathtools,dsfont,bm}
\usepackage{mathrsfs}
\usepackage{mathtools}

\usepackage{comment}
\usepackage{tikz}
\usepackage{algorithm}
\usepackage{algorithmic}
\usepackage{verbatim}
\usepackage{float}
\usepackage[export]{adjustbox}
\usepackage{caption}
\usepackage{subcaption}

\usepackage{cleveref}
\usepackage{enumitem}
\usepackage{todonotes}

\newtheorem{theorem}{Theorem}
\newtheorem{corollary}{Corollary}
\newtheorem{lemma}{Lemma}

\newtheorem{claim}{Claim}

\newtheorem{definition}{Definition}

\newtheorem{observation}{Observation}
\newtheorem{example}{Example}
\newtheorem{remark}{Remark}

\newcommand{\opt}{\textsc{OPT}}
\newcommand{\optstar}{\textsc{OPT}^*}
\newcommand{\optfi}{\textsc{FRev}}
\newcommand{\optsi}{\textsc{SIRev}}
\newcommand{\rev}{\textsc{Rev}}

\newcommand{\yangnote}[1]{{\color{blue}{#1}}}
\newcommand{\mingfeinote}[1]{{\color{magenta}{#1}}}
\newcommand{\notshow}[1]{{}}

\newcommand{\argmax}{\mathop{\arg\max}}
\newcommand{\argmin}{\mathop{\arg\min}}
\DeclareMathOperator{\conv}{conv}
\DeclareMathOperator{\divg}{div}
\DeclareMathOperator{\eps}{\varepsilon}
\newcommand{\cuxd}{\succeq_{cux}}

\def\vol{\textsc{Vol}}

\DeclareMathOperator*{\E}{\mathbb{E}}

\newenvironment{prevproof}[2]{\noindent {\em {Proof of {#1}~\ref{#2}:}}}{$\Box$\vskip \belowdisplayskip}

\def \MM  {\mathcal{M}}
\def \cP  {\mathcal{P}}

\def \cU  {\mathcal{U}}

\def \cT  {\mathcal{T}}
\def \cH  {\mathcal{H}}
\def \cC  {\mathcal{C}}
\def \cE  {\mathcal{E}}
\def \cS  {\mathcal{S}}

\def\Bq{\ensuremath{\textbf{q}}}
\def\Bp{\ensuremath{\textbf{p}}}
\def\Bt{\ensuremath{\textbf{t}}}
\def\Bi{\ensuremath{\textbf{i}}}
\def\supp{\textsc{Supp}}

\newcommand{\ind}{\mathds{1}}

\def \eps  {\varepsilon}

\definecolor{MyGray}{rgb}{0.8,0.8,0.8}

\def \MenuTwoAction {\cC_2}
\def \MenuWithQ {\cC_2'}

\def \MonoConstr {(1)}
\def \MarginConstr {(2)}
\def \IntegralConstr {(3)}
\def \IRConstr {(4)}
\def \MarginConstrDis {(1)}
\def \istarConstrLeft {(3)}
\def \istarConstrRight {(4)}
\def \IntegralConstrDis {(2)}
\def \ikConstrLeft {(5)}
\def \ikConstrRight {(6)}
\def \IRConstrDis {(7)}

\def \responsiveIC {{responsive-IC}}
\def \ResponsiveIC {{Responsive-IC}}
\def \sigmaIC {{$\sigma$-IC}}

\def \semiinformative {{semi-informative}}

\def \optionsize {{3m-1}}

\def \ratio {\textsc{Ratio}}
\def \gap {\text{gap}}

\def \pcont {{\mathcal{P}_{cont}}}

\begin{document}

\title{Is Selling Complete Information (Approximately) Optimal?}

\author{Dirk Bergemann~\footnote{Supported by NSF Award SES-1948692.} \\Yale University \\dirk.bergemann@yale.edu
\and Yang Cai~\footnote{Supported by a Sloan Foundation Research Fellowship and the NSF Award CCF-1942583 (CAREER).} 
\\Yale University\\yang.cai@yale.edu
\and Grigoris Velegkas~\footnote{Supported by a Sloan Foundation Research Fellowship, the NSF Award CCF-1942583 (CAREER), a PhD Scholarship from the Onassis Foundation and a PhD Scholarship from the Bodossaki Foundation.}
\\Yale University \\grigoris.velegkas@yale.edu 
\and Mingfei Zhao~\footnote{Work done in part while the author was studying at Yale University, supported by the NSF Award CCF-1942583 (CAREER).}\\ Google Research\\mingfei@google.com}


\maketitle

\begin{abstract}
    \input{abstract}
\end{abstract}

\thispagestyle{empty}
\addtocounter{page}{-1}
\newpage


\input{intro}   
\input{prelim}

\input{lower_bound_full_info}

\input{upper_bound_full_info}

\input{characterization}

\input{general_setting}

\input{multi_state_uniform}
\input{conclusion}


\newpage
\bibliographystyle{alpha}
\bibliography{Yang.bib}

\newpage
\appendix
\input{additional_prelim}
\input{proof_lower_bound}

\input{proof_upper_bound}

\input{proof_option_size}

\input{proof_characterization}
\input{proof_multi_state}

\end{document}

%% file: abstract.tex
We study the problem of selling information to a data-buyer who faces a decision problem under uncertainty. We consider the classic Bayesian decision-theoretic model pioneered by Blackwell~\cite{blackwell1951comparison,blackwell1953equivalent}. Initially, the data buyer has only partial information about the payoff-relevant state of the world. A data seller offers additional information about the state of the world. The information is revealed through signaling schemes, also referred to as \emph{experiments}. In the single-agent setting, any mechanism can be represented as a menu of experiments. A recent paper by Bergemann et al.~\cite{BergemannBS2018design} present a complete characterization of the revenue-optimal mechanism 
in a binary state and binary action environment. By contrast, no characterization is known for the case with more actions. 
In this paper, we consider more general environments and study arguably the simplest mechanism, which only sells the \emph{fully informative experiment}. 
In the environment with binary state and $m\geq 3$ actions, we provide an $O(m)$-approximation to the optimal revenue by selling only the fully informative experiment and show that the approximation ratio is tight up to an absolute constant factor. An important corollary of our lower bound is that the size of the optimal menu must grow at least linearly in the number of available actions, so no universal upper bound exists for the size of the optimal menu in the general single-dimensional setting. 
We also provide a sufficient condition under which selling only the fully informative experiment achieves the optimal revenue. 

For multi-dimensional environments, we prove that even in arguably the simplest matching utility environment with 3 states and 3 actions, the ratio between the optimal revenue and the revenue by selling only the fully informative experiment can grow immediately to a polynomial of the number of agent types. Nonetheless, if the distribution is uniform, we show that selling only the fully informative experiment is indeed the optimal mechanism.

%% file: intro.tex
\section{Introduction}\label{sec:intro}
As large amounts of data become available and can be communicated more
easily and processed more effectively, information has come to play a
central role for economic activity and welfare in our age. In turn, markets
for information have become more prominent and significant in terms of
trading volume. Bergemann and Bonatti~\cite{bergemann2015selling} provide a recent introduction to markets for
information with a particular emphasis on data markets and data
intermediation in e-commerce.

Information is a valuable commodity for any decision-makers under
uncertainty. By acquiring more information, the decision-maker can refine
his initial estimates about the true state of the world and consequently
improve the expected utility of his decision. In many economically
significant situations, a decision-maker can acquire additional information
from a seller who either already has the relevant information or can
generate the relevant information at little or no cost. Beyond digital and
financial markets, the value of information is particularly important for
health economics where it appears as the expected value of sample
information, see \cite{ades2004expected} or more recently, \cite{ely2021optimal} and \cite{drakopoulos2020perfect} with applications to COVID-19 testing. The value of
additional information is also central in the literature on A/B testing, see 
\cite{azevedo2020b}.

For the seller of information, the question is therefore how much information
to sell and at what price to sell it. A natural starting point
is to make all the information available and sell the information at a price
that maximizes the revenue of the seller. The problem of selling information
at an optimal price is therefore closely related to the classic problem of
the monopolist selling a single unit of an object at an optimal price. Just as in the classic optimal monopoly problem, the seller faces a trade-off
in the choice of the price. A high price generates a high revenue for every
additional sale, but there might be few buyers who value the object higher
than the price asked. A lower price generates a larger volume of sales but
with low marginal revenue. The optimal policy finds the optimal balance
between the marginal revenue and the inframarginal revenue considerations. The sale of complete information faces similar trade-offs. A high price for
the complete information will be acceptable for decision-makers with diffuse
prior information, thus those who value the additional information most, but
 may not be acceptable for those buyers who already have some information.

We analyze these issues in the classic Bayesian decision-theoretic model
pioneered by Blackwell~\cite{blackwell1951comparison,blackwell1953equivalent}. 
Here we interpret the
decision-theoretic model as one where a data buyer faces a decision problem
under uncertainty. A data seller owns a database containing information
about a \textquotedblleft state\textquotedblright\ variable that is relevant
to the buyer's decision. Initially, the data buyer has only partial
information about the state. This information is private to the data buyer
and unknown to the data seller. The precision of the buyer's private
information determines his willingness to pay for any supplemental
information. Thus, from the perspective of the data seller, there are many
possible types of data buyer.

A recent contribution by Bergemann et al.~\cite{BergemannBS2018design} analyzes the optimal selling policy
with the tools of mechanism design. Their analysis is mostly focused on the
canonical decision-theoretic setting with a binary state and a binary action
space. In this setting, the type space is naturally one-dimensional and
given by the prior probability of one state, the probability of the other
state being simply the complementary probability. 

As in \cite{BergemannBS2018design}, we investigate the revenue-maximizing information
policy, i.e., how much information the data seller should provide and how
she should price access to the data. In order to screen the heterogeneous
data buyer types, the seller offers a menu of information products. In the
present context, these products are statistical experiments--- signals that
reveal information about the payoff-relevant state. Only the information
product itself is assumed to be contractible. By contrast, payments cannot
be made contingent on either the buyer's action or the realized state and
signal. Consequently, the value of an experiment to a buyer is determined by
the buyer's private belief and can be computed independently of the price of the experiment. The seller's problem is then to design and price different
versions of experiments, that is, different information products from the
same underlying database.

Indeed, Bergemann et al.~\cite{BergemannBS2018design} find that in the binary action and state setting it is
often optimal to simply sell the complete information to the buyers and
identify the optimal monopoly price. This is for example the case when the
distribution of the prior beliefs, the types of the buyers, is symmetric
around the uniform prior. Yet, they find that sometimes it can be
beneficial to offer two information goods for sale, one that indeed offers
complete information, but also an additional one that offers only partial information.
Thus, a menu of information goods sometimes dominates the sale of a single
object, namely the complete information.

Given the proximity of the problem of selling information to the problem of
selling a divisible good with a unit capacity, the superiority of a menu
over a single choice may appear to be surprising. After all, in the single
good problem, Riley and Zeckhauser~\cite{riley1983optimal} have famously established that selling the
entire unit at an optimal price is always an optimal policy. When we
consider the problem of selling information, we find that the utility of the
buyer is also linear in the posterior probability, just as it is linear in
the quantity in the aforementioned problem. Thus the emergence of a menu
rather than a single item appears puzzling. The difference is however that
the utility of the buyer is only piecewise linear. In particular, in the
binary state binary action setting, it displays exactly one kink in the interior of the
prior probability. The kink emerges at an interior point of the prior belief
where the decision-maker is indifferent between the two actions that are at
his disposal. This particular prior indeed identifies the type of the decision-maker who
has the highest willingness-to-pay for complete information. Moving away from this
point, the utility is then linear in the prior probability. The kink is
foremost an expression of the value of information for the decision-maker.
But following the utility descent away from the kink, the value of
information is decreasing linearly, and thus the seller faces not one, but
two endpoints at which the participation constraints of the buyer have to be
satisfied.

Naturally, with complete information about binary states, there will be two
actions, one for each state which will lead to the highest utility of the
decision-maker. Thus, we might expect that the cardinality of the optimal
menu remains at most binary when we allow the agents to have a larger choice
or action set but stay with binary state space that represents the
uncertainty. In this paper we pursue this question and find that \emph{the
cardinality of the optimal menu increases at least linearly with the number
of available actions}. Thus, selling information forces us to consider larger
menus even when the space of uncertainty remains binary, and thus the
utility as a function of the one-dimensional prior remains piecewise linear
everywhere. Moreover, the cardinality of a set of optimal actions will
always remain at two. Thus, the nature of selling information is
high-dimensional even when the underlying state and ex-post optimal action
space remains~binary~and~thus~small.

The main idea behind the revenue-maximizing mechanism for the information
seller is akin to offering \textquotedblleft damaged
goods\textquotedblright\ to low-value buyers. However, when selling
information goods (see \cite{shapiro1998information}), product versioning allows for richer
and more profitable distortions than with physical goods. This is due to a
peculiar property of information products: Because buyers value different
dimensions (i.e., information about specific state realizations), the buyers
with the lowest willingness to pay also have very specific preferences. For
example, in the context of credit markets, very aggressive lenders are
interested in very negative information only, and are willing to grant a
loan otherwise. The seller can thus leverage the key insight that
information is only valuable if it changes optimal actions---to screen the
buyer's private information.

\subsection{Our Results and Techniques}\label{subsec:result}

We first study the environment when the state is binary, investigating the ratio between the optimal revenue (denoted as $\opt$) and the revenue that can be attained with the sale of complete information (denoted as $\optfi$).
{
\cite{BergemannBS2018design} showed that in the environment with a binary state and a binary action space, $\frac{\opt}{\optfi}\leq 2$. What happens when there are more actions? 
Since there are only two states, the agent has a single dimensional preference. Conventional wisdom from the monopoly pricing problem suggests that the ratio should be no more than a fixed constant. Nonetheless, in the paper we show that the ratio $\frac{\opt}{\optfi}$ is $\Theta(m)$, which is neither a fixed constant nor a function that scales with the number of agent types. Here $m$ is the number of actions.       
}

Our first main result, \Cref{thm:lower_bound_full_info}, 
shows that the revenue obtainable with the sale of complete information is only a fraction $1/\Omega (m)$ of the optimal revenue. 
Using \Cref{thm:lower_bound_full_info} as a springboard, our second main result, \Cref{thm:lower_bound_menu_complexity}, shows that the cardinality of the optimal menu is at least $\Omega(m)$. The reason is that any menu with $\ell$ experiments can obtain no more than $\ell$ times the revenue attained with the sale of
complete information (\Cref{lem:1/k-approx for menus of size k}). To prove \Cref{thm:lower_bound_full_info}, we explicitly construct an environment with $m$ actions, where selling the complete information yields low revenue. However, characterizing the optimal menu appears to be challenging for this environment, and even providing any incentive compatible mechanism with high revenue seems to be non-trivial. We take an indirect approach by first constructing an approximately incentive compatible mechanism that has a revenue at least $\Omega(m)$ times larger than the revenue attained with the sale of complete information, then converting the mechanism to a menu with negligible revenue loss.


In \Cref{thm:1/k approximation of full information for $k$ actions }, 
we show that selling the complete information can always obtain at least a fraction $1/O(m)$ of the optimal revenue, matching the lower bound in~\Cref{thm:lower_bound_full_info} up to a constant factor. This result is established by considering a relaxed problem for the optimal revenue
maximization, where we only keep a subset of the original incentive constraints. We show that the cardinality of the optimal menu in the relaxed problem is always $O(m)$, thus \Cref{thm:1/k approximation of full information for $k$ actions } follows from~\Cref{lem:1/k-approx for menus of size k}.  Note that we only show that the relaxed problem under which the optimal revenue
can be obtained with $O(m)$ experiments. It remains an open question whether the
cardinality of the optimal menu in the original problem grows linearly in the number of available
actions, or perhaps exceeds linear growth in the feasible actions.

These results show that the optimal mechanism to sell information or data
are likely high-dimensional even when the underlying decision problem is
low-dimensional. By tailoring the information to different
decision problems the seller can extract substantially more revenue from the
decision-maker as if he were to rely on a simple mechanism. More specifically, for any constant $c>0$ and any fixed finite menu, there exists a binary-state environment such that the revenue attained by the menu is no more than a $c$-fraction of the optimal revenue (\Cref{cor:fix-menu-size}). Nonetheless, we provide in~\Cref{cor:special_distributions} conditions under which it is indeed optimal to
only sell the largest possible amount of information, akin to the single unit
monopoly problem. 
In Examples~\ref{ex:uniform},~\ref{ex:exponential}, and~\ref{ex:normal} we apply this condition to various parametrized classes of information problems.

Our results above provide a complete understanding of $\frac{\opt}{\optfi}$ in the binary-state environments. What happens if there are more than two states? In fact, we show that the ratio becomes substantially larger when the agent has a multi-dimensional preference. We consider arguably the simplest multi-dimensional environment, where there are 3 states and 3 actions, and the buyer receives payoff $1$ if they match their action $j$ with the state $i$ ($j=i$) and receives payoff $0$ otherwise. We refer to this as the three state \emph{matching utility} environment. We show in \Cref{thm:general-main} that in this multi-dimensional environment, the ratio $\frac{\opt}{\optfi}$ can scale polynomially with the number of agent types $N$. In particular, $\frac{\opt}{\optfi}=\Omega(N^{1/7})$. With \Cref{lem:1/k-approx for menus of size k}, the result also implies that the optimal menu contains at least $\Omega(N^{1/7})$ experiments. The proof is adapted from the approach in \cite{hart2019selling}, which was originally used in the multi-item auction problem. Their approach does not directly apply to our problem due to both the non-linearity of the buyer's value function and the more demanding incentive compatibility constraints in our setting. In the proof, we construct a sequence of types placed on a sequence of concentric thin circular sectors. \cite{hart2019selling} placed all types in a sequence of complete circles. We place the types in these carefully picked circular sectors, so that the constructed mechanism satisfies the more demanding incentive compatibility constraints. Then we construct a discrete distribution on those types, so that the ratio $\frac{\opt}{\optfi}$ is large.

An alternative interpretation of the result is that: There does not exist a universal finite upper bound on either the cardinality of the optimal menu or the ratio $\frac{\opt}{\optfi}$ that holds for all possible type distributions. To complement this result, in \Cref{Thm:multi-state-uniform}, we show that in the same environment -- matching utility with 3 states/actions, selling the complete information is indeed optimal if the distribution is uniform. To prove the result, we propose another relaxation of the problem optimizing over the agent's utility functions. We then construct a dual problem in the form of optimal transportation~\cite{villani2009optimal, daskalakis2017strong, rochet1998ironing}. To verify the optimality of selling complete information, we construct a feasible dual that satisfies the complementary slackness conditions. Although \Cref{Thm:multi-state-uniform} focuses on a special matching utility environment with 3 states, our relaxation provides a general approach to certify the optimality of any specific menu. 

\subsection{Additional Related Work}
Unlike selling a single unit of item, the monopolist pricing problem becomes much more involved in  multi-dimensional settings. Complete characterizations are known only in several special cases~\cite{GiannakopoulosK14,GiannakopoulosK15,haghpanah2021pure,daskalakis2017strong, fiat2016fedex, devanur2017optimal, devanur2020optimal}. A recent line of work provides simple and approximately-optimal mechanisms~\cite{ChawlaHK07,ChawlaHMS10,HartN12,BabaioffILW14,LiY13,Yao15,CaiDW16,CaiZ17,DuttingKL20}.                  Our results (\Cref{thm:lower_bound_full_info} and \Cref{thm:1/k approximation of full information for $k$ actions }) provide matching upper and lower bounds for the performance of selling the complete information, arguably the simplest mechanism for selling information, in single-dimensional settings. In sharp contrast to the monopolist pricing problem, our \Cref{thm:lower_bound_menu_complexity} indicates that the optimal solution for selling information is unlikely to be simple even in single-dimensional settings. 

Our analysis considers the direct sale of information. Here
contracting takes place at the ex ante stage: In this case, the buyer
purchases an information structure (i.e., a Blackwell experiment), as
opposed to paying for specific realizations of the seller's informative
signals.
By contrast, Babaioff et al.~\cite{babaioff2012optimal} and Chen et al.~\cite{chen2020selling} study a model of data lists (i.e.,
pricing conditional on signal realizations) when buyers are heterogeneous
and privately informed. Similarly, Bergemann and Bonatti~\cite{bergemann2015selling} consider the trade of
information bits (\textquotedblleft cookies\textquotedblright ) that are an
input to a decision problem. In these models, the price paid by the buyer
depends on the realization of the seller's information, thus it prices the
information ex interim. 

There are some recent works studying the revenue-optimal mechanism in different model of selling information. Liu et al.~\cite{liu2021optimal} characterizes the revenue-optimal mechanism when the buyer has a linear value function on the scalar-valued state and action. A recent paper by Li~\cite{li2021selling} studies the case where the agent has endogenous information, meaning that she can perform her own experiment at a certain cost after receiving the signal. Finally, Cai and Velegkas~\cite{cai2020sell} study the same problem as ours but focus on efficient algorithms to compute the optimal menu for discrete type distributions. We consider general type distributions and study the cardinality of the optimal menu as well as the performance of selling the complete information.


\notshow{
As large amounts of data become available and can be communicated more easily 
and processed more effectively, information has come to play a central role 
for economic activity and welfare in our age. In turn, markets for 
information have become more prominent and significant in terms of trading volume. 
Information is a valuable commodity for the decision-makers under
uncertainty. By requiring more information, the decision-maker can refine
his initial estimates about the true state of the world and consequently
improve the expected utility of his decision. In many economically
significant situations, a decision-maker can acquire additional information
from a seller who either already has the relevant information or can
generate the relevant information at little or no cost.

For the seller of information the question is therefore how much information
to sell and at what price to sell the information. A natural starting point
is to make all the information available and sell the information at a price
that maximizes the revenue of the seller. The problem of selling information
at an optimal price is therefore closely related to the classic problem of
the monopolist selling a single unit of an object at an optimal price. And
just as in the classic optimal monopoly problem the seller faces a trade-off
in the choice of the price. A high price generates a high revenue for every
additional sale, but there might be few buyers who value the object higher
than the price asked. A lower price generates a larger volume of sales but
with low marginal revenue. The optimal policy finds the optimal balance
between the marginal revenue and the inframarginal revenue considerations.

The sale of complete information faces similar consideration. A high price
for the complete information will be acceptable for decision-makers with
diffused prior information, thus those who value the additional information
most. But a high price, even for complete information may not acceptable for
those buyers who already have some information.

In a recent contribution \cite{BergemannBS2018design} analyze the optimal selling policy
with the tools of mechanism design. Their analysis is mostly focused on the
canonical setting with a binary state and a binary action space. In this
setting the type space is naturally one-dimensional and given by the prior
probability of one state, the probability of the other state is simple the
complementary probability. Indeed, they find that it is often optimal to
simply sell the complete information to the buyers and identify the optimal
monopoly price. This is for example the case when the distribution of the
prior beliefs, the types of the buyers, is symmetric around the uniform
prior. \ Yet, they find that sometimes it can be beneficial to offer two
information goods for sale, one that indeed offers complete information, but
then add one that offers only partial information. Thus, a menu of
information goods sometimes dominates the sale of a single object, namely
the complete information.

Given the proximity of the problem of selling information to the problem of
selling a divisible good with a unit capacity, the superiority of a menu
over a single choice may appear to be surprising. After all, in the single
good problem, \cite{riley1983optimal} have famously established that selling the
entire unit at an optimal price is always an optimal policy. When we
consider the problem of selling information, we find that the utility of the
buyer is also linear in the posterior probability, just as it is linear in
the quantity in the aforementioned problem. Thus the emergence of a menu
rather than a single item appears puzzling. The difference is however that
the utility of the buyer is only piecewise linear. In particular, in the
binary state setting, it displays exactly one kink in the interior of the
prior probability. The kink emerges at an interior point of the prior belief
where the decision-maker is indifferent between the two actions that are at
his disposal. This particular prior indeed identifies the decision-maker who
has the highest willingness-to-pay for information. Moving away from this
point, the utility is then linear in the prior probability. The kink is
foremost an expression of the value of information for the decision-maker.
But following the utility descent away from the kink, the value of
information is decreasing linearly, and thus the seller faces not one, but
two endpoints at which the participation constraints of the buyer has to be
satisfied.

Naturally, with complete information about binary states, there will be two
actions, one for each state which will lead to the highest utility of the
decision maker. Thus, we might expect that the cardinality of the optimal
menu remains at most binary when we allow the agents to have a larger choice
or action set but stay with binary state space that represents the
uncertainty. In this paper we pursue this question and find that the
cardinality of the optimal menu increases at least linearly with the number
of available actions. Thus, selling information forces us to consider larger
menus even when the space of uncertainty remains binary, and thus the
utility as a function of the one-dimensional prior remains piecewise linear
everywhere. Moreover, the cardinality of a set of optimal actions will
always remain at two. Thus, the nature of selling information is
high-dimensional even when the underlying state and ex-post optimal action
space remains binary and thus small. 

Our first main result, Theorem 1, shows that in environments with two states
and $m$ possible actions, the revenue that can be attained with the sale of
complete information is only a fraction $1/\Omega (m)$ of the optimal
revenue. We obtain this bound constructively with a payoff matrix where a
menu with $m$ items obtains a revenue that is $m$ times higher than the
revenue that \grigorisnote{can} be attained from selling complete information. Theorem 1 can
therefore be interpreted as establishing an upper bound on the revenue from
selling complete information. Theorem 2 then establishes a lower bound of $%
\Omega (m)$ on the number of items that are necessary to create an optimal
menu. In Theorem 3, we return to the revenue that can be obtained with the
simple mechanism of selling the complete information. We now provide a lower
bound relative to the optimal revenue. This third result is actually
established by considering a relaxed problem for the optimal revenue
maximization. Here we consider only a subset of incentive constraints. In
particular, as Theorem 3 considers a relaxed problem under which the revenue
can be obtained with $m$ items, it remains an open question whether the
cardinality of the optimal menu grows linearly in the number of available
actions, or perhaps exceeds linear growth in the feasible actions. 

These results show that the optimal mechanism to sell information or data
are\grigorisnote{is?} likely to \grigorisnote{be} high-dimensional even when the underlying decision problem is
low dimensional. By tailoring the information to the different
decision-problems the seller can extract substantially more revenue from the
decision-makers as if he were to rely on simple mechanism. Nonetheless we
provide in Theorem 4 and 5 conditions under which it is indeed optimal to
only sell the largest possible amount of information, akin to single unit
monopoly problem. These conditions are stated in abstract terms, and it
would be interesting to find parametrized class of information problems
under which these conditions could be verified.  
}

%% file: prelim.tex
\section{Preliminaries}\label{sec:prelim}

\paragraph{Model and Notation.} A data buyer (also referred to as the agent) faces a decision problem under uncertainty. The state of the world $\omega$ is drawn from a state space $\Omega=\{\omega_1,\ldots,\omega_n\}$. 
For each $i\in [n]$, we refer to state $\omega_i$ as state $i$ for simplicity.\footnote{Throughout the paper, we denote $[k]=\{1,2,...,k\}$ for any integer $k\geq 1$.} The buyer chooses an action $a$ from a finite action space $A$. We use $m$ to denote the size of $A$ and let $A=[m]$. For every $i\in [n], j\in [m]$, the buyer's payoff for choosing action $j$ under state $i$ is defined to be $u_{ij}$. Denote $\mathcal{U}$ the $n\times m$ payoff matrix that contains all $u_{ij}$'s. The buyer has \emph{matching utility} payoff if $n=m$ and $\mathcal{U}$ is an identity matrix ($u_{ij}=1$ if $i=j$ and 0 otherwise). 

The buyer has some prior information about the state of the world, which is captured by a distribution that represents the probability that the buyer assigns to each of the states. 
We call this piece of prior information the \textit{type} of the buyer, denoted by $\theta=(\theta_1,\ldots,\theta_{n-1})$, where $\theta_i$ represents the probability that the buyer assigns to the state of the world $\omega_i$ for each $i\in [n-1]$ (and a probability of $1-\sum_{i=1}^{n-1}\theta_i$ for state $\omega_n$). 
Denote $\Theta=\{\theta=(\theta_1,\ldots,\theta_{n-1})\in [0,1]^{n-1}\mid \sum_{i=1}^{n-1}\theta_i\leq 1\}$ the space of $\theta$. When the state space is binary, $\theta$ is a scalar in $[0,1]$. For ease of notation, denote $\theta_n=1-\sum_{i=1}^{n-1}\theta_i$ throughout the paper.

{The type $\theta$ is distributed according to some known probability distribution function $F$. We use $\supp(F)$ to denote the support of distribution $F$. When $F$ is a continuous probability distribution (or discrete probability distribution), we use $f$ to denote its probability density function (or its probability mass function). }
Apart from the buyer, there is also a \textit{seller} who observes the state of the world and is willing to sell supplemental information to the buyer.~\footnote{It is not crucial to assume that the seller knows the state of the world, and it suffices to assume that the seller can send signals that are correlated with the state of the world.}
We refer to the buyer as he and to the seller as she.

\paragraph{Experiment.} The seller provides supplemental information to the buyer via a \textit{signaling scheme} which we call \textit{experiment}. A signaling scheme is a commitment to $n$ probability distributions over a finite set of signals $S$, such that when the state of the world is realized, the seller draws a signal from the corresponding distribution and sends it to the buyer. According to \cite{BergemannBS2018design}, we can without loss of generality restrict our attention to the experiments whose signal set is the same as the action space $[m]$ (see \Cref{lem:responsive menu} in \Cref{sec:more_prelim}). One can think of every signal $j$ as the seller recommending the buyer to choose action $j$. We denote such an experiment $E$ by a $n\times m$ matrix. For every $i\in [n], j\in [m]$, the $(i,j)$-th entry, denoted as $\pi_{ij}(E)$ or $\pi_{ij}$ if the experiment is clear from context, is the probability that experiment $E$ sends signal $j$ when the state of the world is $\omega_i$. It satisfies that $\sum_{j\in [m]}\pi_{ij}(E)=1,\forall i\in [n]$. 
\notshow{
\begin{table}[h]
\centering
\begin{tabular}{ c| c c c} \label{tab:experiment}
$E$ & $1$ & $\cdots$ & $m$ \\ 
  \hline
$\omega_1$ & $\pi_{11}$ & $\cdots$ & $0$ \\ 
 $\omega_2$ & $\pi_{21}$ &$\cdots$& $\pi_{2m}$ \\ 
\end{tabular}
\caption{Experiment}
\end{table}
}

We call an experiment $E$ \emph{fully informative} if the seller completely reveals the state of the world by sending signals according to $E$. This means that for every state of the world the experiment recommends the action that yields the highest payoff under this state, i.e. for every $i\in [n]$, $\pi_{ij}(E)=1$ if $j=\argmax_{k}u_{ik}$~\footnote{We break ties lexicographically.} and 0 otherwise.


\paragraph{The Value of an Experiment.} To understand the behavior of the buyer, we first explain how the buyer evaluates an experiment. 
Without receiving any additional information from the seller, the buyer's best action is 
$a(\theta)\in \argmax_{j\in [m]}\{\sum_{i\in [n]}\theta_iu_{ij}\}$, and his maximum expected payoff is $u(\theta) = \max_{j\in [m]}\{\sum_{i\in [n]}\theta_iu_{ij}\}$. 
If he receives extra information from the seller, he updates his beliefs and may choose a new action that induces higher expected payoff. After receiving signal $k\in [m]$ from experiment $E$ his posterior belief about the state of the world is: $\Pr[\omega_i| k, \theta] = \frac{\theta_i\pi_{ik}}{\sum_{\ell\in [n]}\theta_{\ell}\pi_{\ell k}}, \forall i\in [n]$.
The best action is
$a(k|\theta)= \argmax_{j\in [m]}\left\{\frac{\sum_{i\in [n]}\theta_i\pi_{ik}u_{ij}}{\sum_{i\in [n]}\theta_i\pi_{ik}}\right\}$
, which yields maximum expected payoff
$u(k|\theta)=\max_{j\in [m]}\left\{\frac{\sum_{i\in [n]}\theta_i\pi_{ik}u_{ij}}{\sum_{i\in [n]}\theta_i\pi_{ik}}\right\}$.
Taking the expectation over the signal the buyer will receive, we define the value of the experiment $E$ for type $\theta$ to be 
$$V_{\theta}(E) = \sum_{k\in [m]}\max_{j\in [m]}\left\{\sum_{i\in [n]}\theta_i\pi_{ik}u_{ij}\right\}$$

\paragraph{Mechanism.}
Any mechanism $\MM$ can be described as $\{(E(\theta),t(\theta))\}_{\theta\in \Theta}$, where $E(\theta)$ is the experiment type $\theta$ purchases and $t(\theta)$ is the payment. The interaction between the seller and the buyer in any mechanism works as follows:
\begin{enumerate}
    \item The seller commits to a mechanism $\MM=\{E(\theta),t(\theta)\}_{\theta\in\Theta}$.
    \item The state of the world $\omega_i$ ($i\in [n]$) and the type of the buyer $\theta$ are realized.
    \item The buyer reports his type $\theta$ to the mechanism.
    \item The seller sends the buyer a signal $k\in [m]$ with probability $\pi_{ik}(E(\theta))$.
    \item The buyer chooses an action $j\in [m]$, based on his type $\theta$ and the signal $k$. He receives payoff $u_{ij}$, and pays $t(\theta)$ to the seller.
\end{enumerate} 

\noindent In subsequent sections, sometimes we abuse the notation and denote the experiment $E(\theta)$ as $\MM(\theta)$. We assume that the buyer is quasilinear, i.e. he wants to maximize his utility -- the maximum expected payoff minus the payment. 
\paragraph{Incentive Compatibility} A mechanism is \textit{Incentive Compatible} (IC) if reporting his type $\theta$ truthfully maximizes his expected utility: $V_\theta(E(\theta)) - t(\theta) \geq V_{\theta}(E(\theta')) -t(\theta'),\forall \theta, \theta' \in \Theta$. For any $\eps>0$, a mechanism is $\eps$-IC if the inequality is violated by at most $\eps$. Given any mapping $\sigma:[m]\to [m]$, denote $V_\theta^{(\sigma)}(E)$ the value of the experiment $E$, if the buyer chooses action $\sigma(j)$ whenever he receives signal $j$. Formally, 
$$V_{\theta}^{(\sigma)}(E) = \sum_{i\in [n]}\sum_{j\in [m]}\theta_i\pi_{ij}(E)\cdot u_{i,\sigma(j)}.$$ 
\noindent Hence we have $V_{\theta}(E)=\max_{\sigma}V_{\theta}^{(\sigma)}(E)$. Denote $V_\theta^{*}(E)$ the value of the experiment $E$, if the buyer follows the recommendation of the seller, i.e., $\sigma$ is an identity mapping. According to \cite{BergemannBS2018design}, we can without loss of generality restrict our attention to the mechanisms such that if the buyer reports truthfully, following the recommendation from the seller maximizes her expected payoff (see \Cref{lem:responsive menu}). In those mechanisms, $V_{\theta}(E(\theta))=V_{\theta}^*(E(\theta))$.  
Then IC constraints are equivalent to:
\begin{equation}\label{equ:ic}
V_\theta^{*}(E(\theta)) - t(\theta) \geq V_{\theta}^{(\sigma)}(E(\theta')) -t(\theta'),\forall \theta, \theta' \in \Theta, \sigma:[m]\to [m]    
\end{equation}

\paragraph{Individually Rationality.}A mechanism is \textit{Individual Rational} (IR) if reporting his type $\theta$ truthfully induces expected utility at least $u(\theta)$: $V_\theta(E(\theta)) - t(\theta) \geq u(\theta),\forall \theta \in \Theta$. We remark that the agent has expected payoff $u(\theta)$ before receiving any additional information. Thus IR constraints guarantee that the agent has a non-negative utility surplus by participating in the mechanism. Any IC and IR mechanism can be described as a menu $\MM=\{(E,t(E))\}_{E\in \cE}$. The buyer with any type $\theta$ chooses the experiment $E$ that maximizes $V_\theta(E)-t(E)$. The pair $(E,t(E))$ consisting of an experiment $E$ and its price $t(E)$ is called an \emph{option}. The option enables the data-buyer to improves his information and consequently improve his decision. A menu is called fully informative if it only contains the fully informative experiment. We also refer to it as ``selling complete information''. 

\paragraph{Revenue.} An \emph{environment} is a particular choice of parameters of the model (i.e. the payoff matrix and type distribution). Fix an environment, we denote by $\rev(\MM)$ the revenue that mechanism $\MM$ generates and by $\opt$ the optimal revenue among all IC, IR mechanisms. 
We denote by $\optfi$ the maximum revenue achievable by any fully informative menu.


\paragraph{Payoff Matrix in Binary States}
When the state space is binary, we can without loss of generality make some assumptions on the payoff matrix. For any action $j$, we say the action is \emph{redundant}, if there exists a set of actions $S\subseteq [m]\backslash \{j\}$ and a distribution $d=\{d_k\}_{k\in S}$ such that: $u_{ij}\leq \sum_{k\in S}d_ku_{ik},\forall i\in \{1,2\}$. In this scenario, the buyer will never choose action $j$ as it is dominated by choosing an action according to distribution $d$ regardless of the underlying state. Without loss of generality, we assume that none of the actions is redundant, which clearly implies that there does not exist $k\not=j$ such that $u_{1k}\geq u_{1j},u_{2k}\geq u_{2j}$. Thus we can assume that $$u_{11}>u_{12}>...>u_{1m}=0 \quad\text{and}\quad 0=u_{21}<u_{22}<...<u_{2m},$$ otherwise we can change the action indexes, or modify every $u_{1j}$ (or $u_{2j}$) by the same amount. We further assume $u_{11}\geq u_{2m}=1$, otherwise we can scale all payoffs or swap the states. Under the assumption above, the fully informative experiment $E$ satisfies that $\pi_{11}(E)=\pi_{2m}(E)=1$.

\begin{table}[h]
\centering
\subfloat[Payoff Matrix]{\begin{tabular}{ c| c c c c c} \label{tab:payoff matrix}
$u$ & $1$ & $\cdots$ & $j$ & $\cdots$ & $m$ \\ 
  \hline
$\omega_1$ & {$u_{11}$} & $\cdots$ & $u_{1j}$ & $\cdots$ &$0$ \\ 
 $\omega_2$ & $0$ &$\cdots$& $u_{2j}$ & $\cdots$ & {$u_{2m}=1$} \\ 
\end{tabular}}
\qquad\quad
\subfloat[Experiment]{\begin{tabular}{ c| c c c c c} \label{tab:experiment}
$E$ & $1$ & $\cdots$ &  $j$ & $\cdots$ & $m$ \\ 
  \hline
$\omega_1$ & $\pi_{11}$ & $\cdots$ & $\pi_{1j}$ & $\cdots$ & $\pi_{1m}$ \\ 
 $\omega_2$ & $\pi_{21}$ &$\cdots$& $\pi_{2j}$ & $\cdots$ & $\pi_{2m}$ \\ 
\end{tabular}}
\end{table}












%% file: lower_bound_full_info.tex
\section{Binary State: A Lower Bound for Selling Complete Information}\label{sec:lower_bound_full_info}
We focus on the binary-state case ($n=2$) in \Cref{sec:lower_bound_full_info} and \Cref{sec:upper_bound_full_info}. In this section, we show that for any $m$, there exists some environment with 2 states and $m$ actions such that selling the complete information is only an $\frac{1}{\Omega(m)}$-fraction of the optimal revenue $\opt$. 
Proofs in this section are postponed to \Cref{sec:proof_lower_bound}. 

In fact, we prove a stronger statement that in this environment, $\optfi$ is an $\frac{1}{\Omega(m)}$-fraction of the maximum revenue among a special class of menus called \emph{semi-informative menu}. 
In particular, an experiment $E$ is semi-informative if it satisfies: (i) It only recommends the fully informative actions (Action 1 and $m$), i.e., $\pi_{ij}(E)=0$, for all $i\in \{1,2\},2\leq j\leq m-1$; (ii) Either $\pi_{11}(E)=1$ or $\pi_{2m}(E)=1$. Since $\sum_{j\in [m]}\pi_{1j}(E)=\sum_{j\in [m]}\pi_{2j}(E)=1$, any experiment $E$ that satisfies both of the above properties has one of the following patterns: 

\vspace{-.1in}
\begin{table}[ht]
\centering
\subfloat[Pattern 1: $\pi_{2m}(E)=1$]{\begin{tabular}{ c| c c c c c}
$E$ & $1$ & $\cdots$ & $j$ & $\cdots$ & $m$ \\ 
  \hline
$\omega_1$ & {$\pi_{11}$} & $\cdots$ & $0$ & $\cdots$ &$1-\pi_{11}$ \\ 
 $\omega_2$ & $0$ &$\cdots$& $0$ & $\cdots$ & $1$ \\ 
\end{tabular}}
\qquad\quad
\subfloat[Pattern 2: $\pi_{11}(E)=1$]{\begin{tabular}{ c| c c c c c} 
$E$ & $1$ & $\cdots$ &  $j$ & $\cdots$ & $m$ \\ 
  \hline
$\omega_1$ & $1$ & $\cdots$ & $0$ & $\cdots$ & $0$ \\ 
 $\omega_2$ & $1-\pi_{2m}$ &$\cdots$& $0$ & $\cdots$ & $\pi_{2m}$ \\  
\end{tabular}}
\caption{Two Specific Patterns of the Semi-informative Experiment}\label{table:two-patterns}
\end{table}

\vspace{-.2in}
A mechanism (or menu) is called semi-informative if every experiment in the mechanism (or menu) is semi-informative. Denote $\optsi$ the optimal revenue achieved by any semi-informative menu. Clearly $\opt\geq\optsi\geq \optfi$ since the fully informative experiment is also semi-informative. The class of semi-informative menus is also useful in \Cref{sec:upper_bound_full_info}. 

\begin{theorem}\label{thm:lower_bound_full_info}
For every $m$, there exists a payoff matrix with 2 states and $m$ actions, together with a type distribution $F$ of the agent, such that $\optsi=\Omega(m)\cdot \optfi$, which implies that $\opt=\Omega(m)\cdot \optfi$. 
\end{theorem}
 
To prove the theorem, we first introduce the concept of IR curve in the binary-state case.

\paragraph{IR Curve.} 
In binary-state case, $u(\cdot)$ can be viewed as a function in $[0,1]$, which we refer to as the \emph{IR curve}. By definition, $u(\cdot)$ is a maximum over $m$ linear functions. Thus, it is a continuous, convex and piecewise linear function. Since none of the actions is redundant, the IR curve $u(\cdot)$ contains exactly $m$ pieces. We show in \Cref{table:example} an example with $4$ actions and in \Cref{fig:IRcurve-example} the corresponding IR curve.

\begin{minipage}[b]{0.4\textwidth}
\centering
\begin{tabular}[b]{ c| c c c c}
$u$ & $1$ & $2$ & $3$ & $4$ \\ 
\hline
$\omega_1$ & $1$ & $0.8$ & $0.6$ & $0$\\
$\omega_2$ & $0$ & $0.5$ & $0.8$ & $1$
\end{tabular}
\captionof{table}{The payoff matrix of an example with $2$ states and $4$ actions}
\label{table:example}
\end{minipage}
\hfill
\begin{minipage}[b]{0.5\textwidth}
\centering
\includegraphics[width=0.55\textwidth]{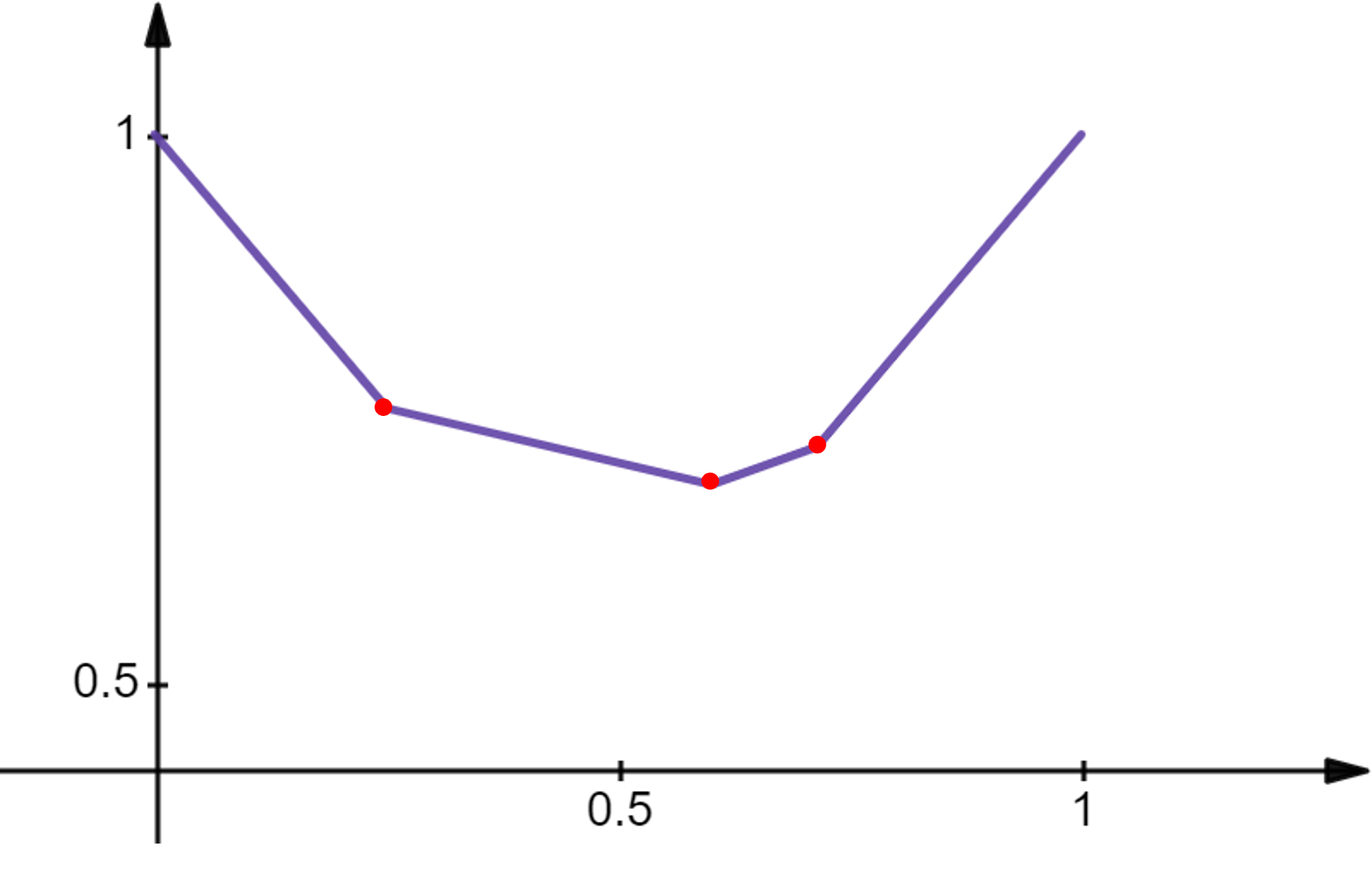}
\captionof{figure}{The IR curve $u(\cdot)$. The slopes of the $4$ pieces are $-1, -0.2, 0.3$ and $1$ respectively.}
\label{fig:IRcurve-example}
\end{minipage}

\begin{lemma}\label{lem:ER-dist}
Given any $a,b\in (0,1)$ such that $a<b$, let $h:[a,b]\to (0,1)$ be any piecewise linear function that is {strictly} decreasing and convex. 
Consider any environment with IR curve $h$.~\footnote{
When the type distribution has support $[a,b]$, we can restrict our attention to the IR curve at range $[a,b]$. In our construction, we construct a payoff matrix that implements the IR curve.} Then there exists a continuous distribution $F$ over support $[a,b]$, such that $\optfi=1-h(a)$, and

$$\frac{\int_a^bf(\theta)\cdot (1-h(\theta))d\theta}{\optfi}= \log\left(\frac{1-h(b)}{1-h(a)}\right),$$
where $f$ is the pdf of $F$. {When $u_{11}=u_{2m}=1$, the highest expected payoff an agent can achieve is $1$, so $\int_a^bf(\theta)\cdot (1-h(\theta))d\theta$ is also the expected full surplus under distribution $F$.}
\end{lemma}

We provide a sketch of the proof of \Cref{thm:lower_bound_full_info} here. In the first step, we construct an environment with 2 states and $m$ actions by constructing an IR curve then apply \Cref{lem:ER-dist} to create an distribution $F$ over the agent types to show that the ratio between the full expected surplus and $\optfi$ is $\Omega(m)$. In the second step, we indirectly construct a semi-informative menu whose revenue is a constant fraction of the total expected surplus, which implies that $\optsi=\Omega(m)\cdot \optfi$. In our construction, the payoffs satisfy $u_{11}=u_{2m}=1$.

\paragraph{Construction of the IR curve:}

\notshow{
\begin{table}[h!]
\centering
\subfloat[Payoff Matrix]{\begin{tabular}{ c| c c c c c} 
$u$ & $1$ & $\cdots$ & $j$ & $\cdots$ & $m$ \\ 
  \hline
$\omega_1$ & {$u_{11}$} & $\cdots$ & $u_{1j}$ & $\cdots$ &$0$ \\ 
 $\omega_2$ & $0$ &$\cdots$& $u_{2j}$ & $\cdots$ & {$u_{2m}=1$} \\ 
\end{tabular}}
\qquad\quad
\subfloat[Re-indexed Payoffs in Our Construction]{\begin{tabular}{ c| c c c c c} 
 & $a_{m-1}$ & $\cdots$ &  $a_{m-j}$ & $\cdots$ & $a_0$ \\ 
  \hline
$\omega_1$ & {$u_{m-1}=1$} & $\cdots$ & $u_{m-j}$ & $\cdots$ &$0$ \\ 
 $\omega_2$ & $0$ &$\cdots$& $v_{m-j}$ & $\cdots$ & {$v_0=1$} \\ 
\end{tabular}}
\caption{Re-indexing the Actions for the Lower Bound Construction}
\label{table-lower-bound-payoff}
\vspace{-.1in}
\end{table}
}

Given $m-1$ types $0<\theta_1<\theta_2<...<\theta_{m-1}<1$ which will be determined later, we consider the following IR curve $u(\cdot)$ on $[\theta_1,\theta_{m-1}]$. The IR curve is a piecewise linear function with $m-2$ pieces. For simplicity, for every $1\leq i\leq m-2$, we will refer to piece $L_i$ as the piece for the corresponding action $m-i$.\footnote{We choose the subscript this way so that the IR curve goes from a piece with smaller index to a piece with larger index as $\theta$ increases.} Fix any $\eps$ less than $\frac{1}{2^m}$.

Let $d_0:=\eps^{2^m}$. For every $1\leq i\leq m-2$, let $d_i:=2^i\cdot \eps^{2^m-2^i}$. For every $1\leq i\leq m-1$, let $\theta_i:=\sum_{j=0}^{i-1}d_j$. Then when $\eps<\frac{1}{2^m}$,
    \begin{equation}\label{equ:theta_m-1}
        \theta_i=\eps^{2^m}+\sum_{j=1}^{i-1}2^j\cdot \eps^{2^m-2^j}<\eps^{2^{m}-2^{i-1}-1}
    \end{equation}
For every $1\leq i\leq m-2$, piece $L_i$ is in the range $\theta\in [\theta_i,\theta_{i+1}]$. The slope of piece $L_i$ is $l_i:=-\eps^{2^i}$.
Thus the height of piece $L_i$ is $h_i:=u(\theta_{i+1})-u(\theta_i)=-l_id_i=2^i\cdot\eps^{2^m}$. Let $u(\theta_1):=1-\theta_1$ and $h_0:=\theta_1=d_0=\eps^{2^m}$. Then for every $1\leq i\leq m-1$, we have $u(\theta_i)=1-\sum_{j=0}^{i-1}h_j$. See \Cref{fig:IRcurve-construction} for an illustration of notations.

Now we compute the payoff matrix according to the constructed IR curve $u(\cdot)$. $u_{11}=u_{2m}=1$, $u_{1m}=u_{21}=0$. For every $1\leq i\leq m-2$, before receiving any additional information, choosing action $m-i$ induces expected payoff $\theta\cdot u_{1, m-i}+(1-\theta)\cdot u_{2, m-i}$. By the definition of the IR curve, it must go through the two points $(\theta_{i},u(\theta_i))$ and $(\theta_{i+1},u(\theta_{i+1}))$. Thus we have $\theta_i\cdot u_{1,m-i}+(1-\theta_i)\cdot u_{2,m-i}=1-\sum_{j=0}^{i-1}h_j$ and $\theta_{i+1}\cdot u_{1,m-i}+(1-\theta_{i+1})\cdot u_{2,m-i}=1-\sum_{j=0}^{i}h_j$. 
Thus $u_{2,m-i}=1-\sum_{j=0}^{i-1}h_j-\theta_il_i$, $u_{1,m-i}=1-\sum_{j=0}^{i-1}h_j-\theta_il_i+l_i$. 

By \Cref{lem:ER-dist}, there exists a distribution $F$ under support $[\theta_1,\theta_{m-1}]$, such that $\optfi=1-u(\theta_1)=\theta_1=\eps^{2^m}$, and
\begin{equation}\label{equ:lower bound integral}
\begin{aligned}
    &\int_{\theta_1}^{\theta_{m-1}}f(\theta)\cdot (1-u(\theta))d\theta\geq \log\left(\frac{1-u(\theta_{m-1})}{1-u(\theta_1)}\right)\cdot \optfi
    =\optfi\cdot \log\left(\frac{\sum_{i=0}^{m-2}h_i}{h_0}\right)\\
    =&\optfi\cdot \log \left(\frac{\sum_{i=0}^{m-2}2^i\cdot \eps^{2^m}}{\eps^{2^m}}\right)
    =\optfi\cdot \log \left(2^{m-1}-1\right)
    \geq \optfi\cdot\log(2)\cdot (m-2) 
\end{aligned}
\end{equation}

Next, we present a semi-informative mechanism $\MM$ whose revenue is comparable to the integral $\int_{\theta_1}^{\theta_{m-1}}f(\theta)\cdot (1-u(\theta))d\theta$.

\paragraph{Construction of the mechanism:} Consider the following $m-2$ semi-informative experiments. 
For every $1\leq i\leq m-2$, we design experiment $E_i$ and its price $p_i$ so that the buyer's utility after purchasing experiment $E_i$ coincides with the piece $L_i$ of $u(\cdot)$ for $\theta\in [\theta_i,\theta_{i+1}]$, when the buyer follows the recommendation. Formally, experiment $E_i$ is as follows:

\begin{minipage}[b]{0.5\textwidth}
\centering
\includegraphics[width=0.45\textwidth]{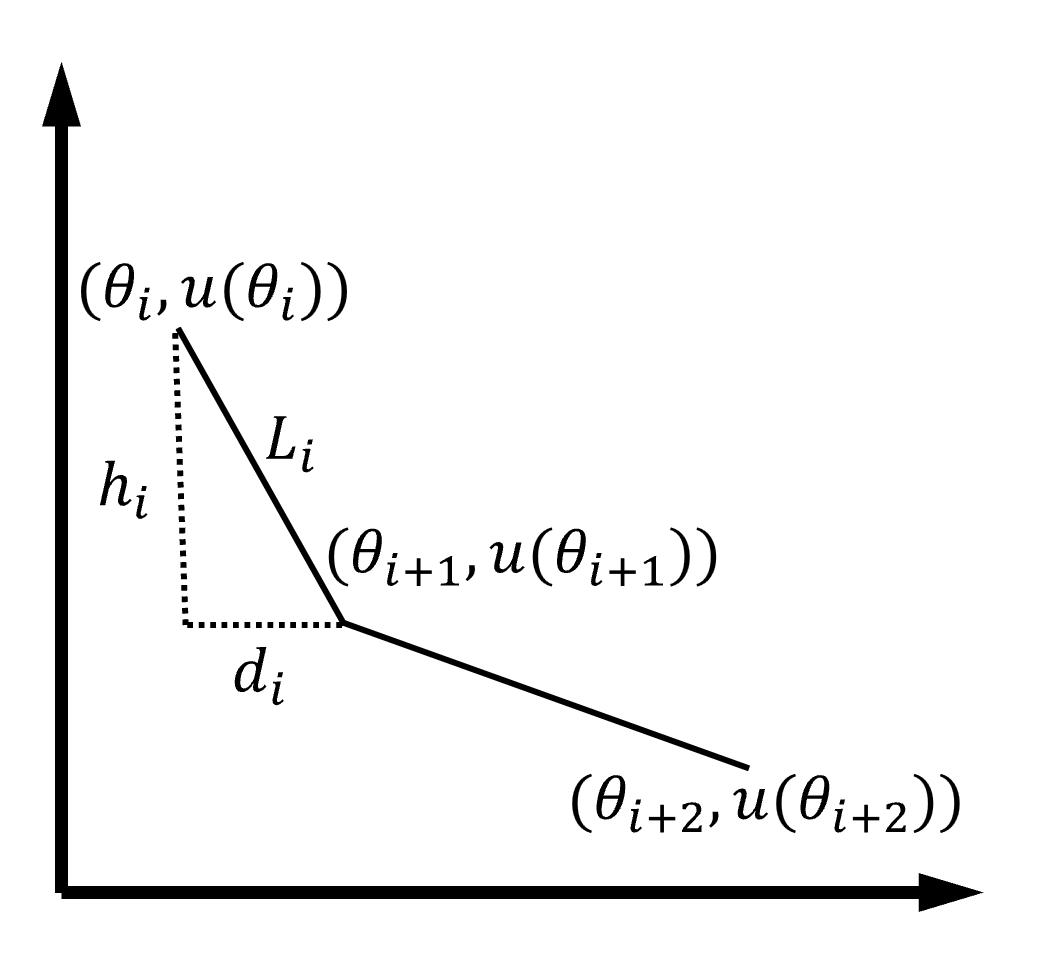}
\captionof{figure}{Notations in the construction of IR curve}
\label{fig:IRcurve-construction}
\end{minipage}
\hfill
\begin{minipage}[b]{0.45\textwidth}
\centering
\begin{tabular}{ c| c c c} \label{tab:experiment in construction}
$E_i$ & $1$ & $\cdots$ & $m$ \\ 
  \hline
$\omega_1$ & $1+l_i$ & $\mathbf{0}$ & $-l_i$ \\ 
 $\omega_2$ & $0$ &$\mathbf{0}$& $1$
\end{tabular}
\captionof{table}{Experiment $E_i$}
\label{table:construction}
\end{minipage}

\notshow{$\begin{pmatrix}
1+l_i & -l_i\\
0 & 1
\end{pmatrix}$} 
\noindent The price for experiment $E_i$ is $p_i=\sum_{j=0}^{i-1}h_j+\theta_il_i$. It is not hard to verify that the buyer's utility for  buying experiment $E_i$ is $u(\theta_i)$ at $\theta_i$ and $u(\theta_{i+1})$ at $\theta_{i+1}$.

Consider the following mechanism $\MM$. For every $\theta\in [\theta_1,\theta_{m-1}]$, let $i$ be the unique number such that $\theta_i<\theta\leq \theta_{i+1}$.\footnote{Choose $i=1$ when $\theta=\theta_1$. It's a point with 0 measure, so it won't affect the revenue of the mechanism.} The outcome and payment of $\MM$ under input $\theta$ are defined as follows: It computes the buyer's utilities for $(E_i,p_i)$, $(E_{i-1},p_{i-1})$ (when $i\geq 2$), and $(E_{i-2}, p_{i-2})$ (when $i\geq 3$) under type $\theta$, and chooses the option with the highest utility as the experiment and payment. 

We prove in \Cref{lem:mech-m} that $\MM$ is IR and $\delta$-IC for some $\delta=o(\eps^{2^m})$. Moreover, $\rev(\MM)$ is comparable to the integral $\int_{\theta_1}^{\theta_{m-1}}f(\theta)\cdot (1-u(\theta))d\theta$. $\MM$ can then be converted to a semi-informative menu by losing no more than a constant fraction of the total revenue (see \Cref{lem:eps-ic} in \Cref{sec:proof_lower_bound}). Since $\optfi=\eps^{2^m}$, \Cref{thm:lower_bound_full_info} then follows from Inequality~\eqref{equ:lower bound integral}. The proofs of \Cref{lem:mech-m} and \Cref{thm:lower_bound_full_info} are postponed to \Cref{sec:proof_lower_bound}.

\begin{lemma}\label{lem:mech-m}
{For any $\eps\in(0,2^{-m})$,} $\MM$ is IR and $\delta$-IC where $\delta=7\cdot\eps^{2^m+1}$. Moreover, $\rev(\MM)\geq \frac{1}{9}\cdot \int_{\theta_1}^{\theta_{m-1}}f(\theta)\cdot (1-u(\theta))d\theta$.
\end{lemma}


An important application of~\Cref{thm:lower_bound_full_info} is that, there is an $\Omega(m)$ lower bound of the size of the optimal menu. The main takeaway of this theorem is that, in our model where there is a single agent whose type is single-dimensional, the optimal menu, however, can be complex.   


\begin{theorem}\label{thm:lower_bound_menu_complexity}
For every $m$, there exists a payoff matrix with 2 states and $m$ actions, together with a type distribution $F$, such that any optimal menu $\MM^*$ consists of at least $\Omega(m)$ different experiments. 
\end{theorem}

To prove \Cref{thm:lower_bound_menu_complexity}, it suffices to show that: Given any menu $\MM$ with $\ell$ experiments, we can sell the complete information at an appropriate price to achieve revenue at least $\rev(\MM)/\ell$. In \Cref{lem:1/k-approx for menus of size k}, we prove a more general result that applies to any IR (not necessarily IC) mechanism, which is useful in \Cref{sec:upper_bound_full_info}. 

\begin{definition}\label{def:size_of_mechanism}
For any finite integer $\ell>0$, we say that a mechanism $\MM=\{(E(\theta),t(\theta))\}_{\theta\in[0,1]}$ (not necessarily IC nor IR) has \emph{option size} $\ell$, if there exists $\ell$ different options $\{(E_j,t_j)\}_{j\in [\ell]}$ such that: For every $\theta\in [0,1]$, there exists some $j\in [\ell]$ that satisfies $E(\theta)=E_j$ and $t(\theta)=t_j$. 
\end{definition}

\begin{lemma}\label{lem:1/k-approx for menus of size k}
For any positive integer $\ell$, let $\MM$ be any IR mechanism {of option size $\ell$} 
that generates revenue $\rev{(\MM)}$. 
Then there exists a menu $\MM'$ that contains only the fully informative experiment and generates revenue at least $\frac{\rev(\MM)}{\ell}$.
\end{lemma}

Another important implication of \Cref{thm:lower_bound_full_info} is shown in \Cref{cor:fix-menu-size}. It states that there is no menu with finite cardinality that achieves any finite approximation to the optimal revenue for all single-dimensional environments. The proof directly follows from \Cref{thm:lower_bound_full_info} and \Cref{lem:1/k-approx for menus of size k}.

\begin{corollary}\label{cor:fix-menu-size}
For any finite $\ell>0$, any menu $\MM$ with size at most $\ell$, and any $c>0$, there exists an finite integer $m$, a payoff matrix with 2 states and $m$ actions, together with a distribution $F$, such that $\rev(\MM)$ is at most a $c$-fraction of the optimal revenue in this environment.
\end{corollary}

%% file: upper_bound_full_info.tex
\section{Binary State: An Upper Bound for Selling Complete Information}\label{sec:upper_bound_full_info}

In this section, we provide upper bounds for the gap between the optimal revenue $\opt$ and the revenue by selling complete information $\optfi$. In \Cref{subsec:upper_bound}, we prove an $O(m)$-approximation to the optimal menu using only the fully informative experiment, where $m$ is the number of actions. In \Cref{sec:characterization}, we provide conditions under which selling complete information is the optimal menu. Proofs in this section are postponed to \Cref{sec:proof_upper_bound} and \Cref{sec:proof_characterization}.

\subsection{An Upper Bound for Selling Complete Information}\label{subsec:upper_bound}

The main result of this section is stated in \Cref{thm:1/k approximation of full information for $k$ actions }. Together with \Cref{thm:lower_bound_full_info}, we know that the $O(m)$ approximation ratio is tight up to an absolute constant factor. 

\begin{theorem}\label{thm:1/k approximation of full information for $k$ actions }
For any environment with 2 states and $m$ actions, there is a menu that contains only the fully informative experiment, whose revenue is at least $\Omega(\frac{\opt}{m})$. In other words, $\opt=O(m)\cdot \optfi$.
\end{theorem}

To prove \Cref{thm:1/k approximation of full information for $k$ actions }, a natural idea is to first show that the optimal menu contains $O(m)$ experiments, then use  \Cref{lem:1/k-approx for menus of size k} to argue that we can sell the full information at an appropriate price to achieve revenue at least $\Omega\left(\frac{\rev(\MM)}{m}\right)$. Unfortunately, we are not aware of such an upper bound on the size of the optimal menu, and it is not clear if the size of the optimal menu is indeed linear in $m$. Instead, we drop some of the IC constraints and consider the maximum revenue of a relaxed problem. 

We first introduce the concept of {\responsiveIC} and {\sigmaIC} constraints.
Recall the IC constraint:
$V_\theta^{*}(E(\theta)) - t(\theta) \geq V_{\theta}^{(\sigma)}(E(\theta')) -t(\theta'), \forall \theta, \theta' \in \Theta, \sigma:[m]\to [m]$
where $V_\theta^{*}(E(\theta))$ is the agent's utility for receiving experiment $E(\theta)$ and following the recommendation. We distinguish the IC constraints by the mapping $\sigma$. When $\sigma$ is the identity mapping, we refer to the constraint as the \emph{\responsiveIC} constraint, and when $\sigma$ is any non-identity mapping, we refer to the constraint as the \emph{\sigmaIC} constraint. A mechanism is \emph{\responsiveIC} if it satisfies all {\responsiveIC} constraints.  

As \Cref{lem:1/k-approx for menus of size k} applies to any IR mechanism, we drop the {\sigmaIC} constraints and bound the number of experiments offered by the optimal {\responsiveIC} and IR mechanism by $O(m)$, which suffices to prove \Cref{thm:1/k approximation of full information for $k$ actions }. An important component of our proof is \Cref{lem:describe the menu with q}. Denote $\MenuWithQ$ the set of all semi-informative, {\responsiveIC} and IR mechanisms.
\notshow{
\begin{table}[ht]
\centering
\subfloat[Pattern 1: $\pi_{2m}(E)=1$]{\begin{tabular}{ c| c c c c c}
$E$ & $1$ & $\cdots$ & $j$ & $\cdots$ & $m$ \\ 
  \hline
$\omega_1$ & {$\pi_{11}$} & $\cdots$ & $0$ & $\cdots$ &$1-\pi_{11}$ \\ 
 $\omega_2$ & $0$ &$\cdots$& $0$ & $\cdots$ & $1$ \\ 
\end{tabular}}
\qquad\quad
\subfloat[Pattern 2: $\pi_{11}(E)=1$]{\begin{tabular}{ c| c c c c c} 
$E$ & $1$ & $\cdots$ &  $j$ & $\cdots$ & $m$ \\ 
  \hline
$\omega_1$ & $1$ & $\cdots$ & $0$ & $\cdots$ & $0$ \\ 
 $\omega_2$ & $1-\pi_{2m}$ &$\cdots$& $0$ & $\cdots$ & $\pi_{2m}$ \\  
\end{tabular}}
\caption{Two Specific Patterns of the Experiment}\label{table:two-patterns}
\vspace{-.15 in}
\end{table}
}
We prove in~\Cref{lem:describe the menu with q} that the maximum revenue achievable by any {\responsiveIC} and IR mechanism (denoted as $\optstar$) can be achieved by a semi-informative mechanism. 

\begin{lemma}\label{lem:describe the menu with q}
There exists $\MM^*\in \MenuWithQ$ such that $\rev(\MM^*)=\optstar$.

\end{lemma}

 
By \cite{BergemannBS2018design}, any {\semiinformative} experiment is determined by a single-dimensional variable $q(E)=\pi_{11}\cdot u_{11}-\pi_{2m}\cdot u_{2m}$ (see also \Cref{obs:q to exp}). 
Given a mechanism $\MM$ where every experiment is {\semiinformative}, for any $\theta\in [0,1]$, we slightly abuse the notation and let $q(\theta):=q(\MM(\theta))$. $\MM$ can also be described as the tuple $(\Bq=\{q(\theta)\}_{\theta\in [0,1]},\Bt=\{t(\theta)\}_{\theta\in [0,1]})$. In \Cref{lem:ic in q space}, we present a characterization of all $\{q(\theta)\}_{\theta\in [0,1]}$ that can be implemented with a {\responsiveIC} and IR mechanism. The proof uses the payment identity and it is similar to 
Lemma~1 in \cite{BergemannBS2018design}, where there are only two actions (Action 1 and $m$) in their setting. The main difference is that in our lemma, the IR constraints are not implied by the monotonicity of $q(\cdot)$. 

\begin{lemma}\label{lem:ic in q space}
Given $\Bq=\{q(\theta)\}_{\theta\in [0,1]}$, there exists non-negative payment rule $\Bt$ such that $\MM=(\Bq,\Bt)$ is a {\responsiveIC} and IR mechanism if and only if
\begin{enumerate}
\item $q(\theta)\in [-u_{2m},u_{11}]$ is non-decreasing in $\theta$.
\item $\int_{0}^1q(x)dx=u_{11}-u_{2m}$.
\item For every $\theta\in [0,1]$, $u_{2m}+\int_0^{\theta}q(x)dx\geq u(\theta)$. Recall that $u(\theta)$ is the value of the agent with type $\theta$ without receiving any experiment. 
\end{enumerate}
Moreover, the payment rule $\Bt$ must satisfy that for every $\theta$, 
\vspace{-.05in}
\begin{equation}\label{equ:payment}
\textstyle t(\theta) = \theta\cdot q(\theta) + \min\{u_{11} - u_{2m} - q(\theta), 0\} - \int_0^\theta q(x) dx
\end{equation}
\end{lemma}

By \Cref{lem:ic in q space}, for any $\Bq$ that can be implemented with some {\responsiveIC} and IR mechanism $\MM$, the revenue of $\MM$ can be written as an integral of $q(\cdot)$ using \Cref{equ:payment} (see the objective of the program in \Cref{fig:continous-program-modified}). We notice that for continuous distribution $F$, Property 3 of \Cref{lem:ic in q space} corresponds to uncountably many inequality constraints. To bound the size of the optimal {\responsiveIC} and IR mechanism, another important component is to show that all (IR) constraints in Property 3 can be captured by $O(m)$ constraints (Constraint {\IRConstr} in \Cref{fig:continous-program-modified}). By \Cref{lem:ic in q space}, the optimal semi-informative, {\responsiveIC} and IR mechanism 
is captured by the optimization problem in \Cref{fig:continous-program-modified}. A formal argument is shown in \Cref{lem:modified-program-equivalent}.



\begin{figure}[H]
\colorbox{MyGray}{
\begin{minipage}{.98\textwidth}
$$\quad\textbf{{sup}  } \int_0^1 [(\theta f(\theta) + F(\theta))q(\theta) + \min\{(u_{11}-u_{2m} - q(\theta))f(\theta), 0\}]d\theta$$
  \begin{align*}
\text{s.t.} \qquad 
 &{\MonoConstr}\quad q(\theta) \text{ is non-decreasing in }\theta\in [0,1]\\
 & {\MarginConstr}\quad q(0)\geq-u_{2m},\quad q(1)\leq u_{11}\\
 &{\IntegralConstr}\quad \int_{0}^1 q(\theta) d\theta = u_{11}-u_{2m} \\
 & {\IRConstr}\quad \int_0^{1}(q(x)-\ell_k)\cdot \ind[q(x)\leq \ell_k]dx\geq u_{2,m+1-k}-u_{2m}, &\forall k\in \{2,3,...,m-1\}
\end{align*}
\end{minipage}}
\caption{Maximizing Revenue over {\ResponsiveIC} and IR Mechanisms}~\label{fig:continous-program-modified}
\end{figure}

\begin{lemma}\label{lem:modified-program-equivalent}
For any optimal solution $\Bq^*$ to the program in \Cref{fig:continous-program-modified}, the mechanism $\MM^*$ that implements $\Bq^*$ (\Cref{lem:ic in q space}) achieves the maximum revenue among all {\responsiveIC} and IR mechanisms.~\footnote{There is a feasible solution that achieves the supremum of the program. See \Cref{sup_equal_max}.} 
\end{lemma}

In \Cref{lem:opt without sigma-IC has k pieces}, we show that there exists an optimal {\responsiveIC} and IR mechanism whose option size is $O(m)$. By \Cref{lem:modified-program-equivalent}, it's equivalent to prove that there is an optimal solution $\Bq^*$ that takes only $O(m)$ different values. When $F$ is a discrete distribution, we turn the program into a collection of LPs, so that the highest optimum among the collection of LPs correspond to the optimum of the program in \Cref{fig:continous-program-modified}. Each LP has variables that represent the difference of the $q$-value between two adjacent $\theta$s (as the types are discrete). The LP has $O(m)$ constraints, that is independent of $\supp(F)$, as constraint {\MonoConstr} corresponds to the variables being non-negative. Thus each LP has an optimal solution where at most $O(m)$ variables are strictly positive, which corresponds to a $\Bq$ that takes only $O(m)$ different values. For continuous distribution $F$, we prove the claim by approximating the continuous program with an infinite sequence of discrete programs. 
\Cref{thm:1/k approximation of full information for $k$ actions } then follows from Lemmas \ref{lem:1/k-approx for menus of size k}, \ref{lem:describe the menu with q} and \ref{lem:opt without sigma-IC has k pieces}.


\begin{lemma}\label{lem:opt without sigma-IC has k pieces}
There exists a semi-informative, {\responsiveIC} and IR mechanism 
that 
has option size at most $\optionsize$ (\Cref{def:size_of_mechanism}) and obtains revenue $\rev(\MM^*)=\optstar$. 
\end{lemma}


%% file: characterization.tex
\subsection{When is Selling Complete Information Optimal?}\label{sec:characterization}

We have proved a tight approximation ratio of $\Theta(m)$ for selling complete information that applies to all binary-state environments. A natural follow-up question is whether the approximation ratio becomes significantly smaller for special environments. In this section, we provide a sufficient condition for the environment, under which selling complete information achieves revenue equals to $\opt^*$, the maximum revenue achievable by any {\responsiveIC} and IR mechanisms.\footnote{In the paper, we indeed prove a necessary and sufficient condition under which selling complete information is optimal among all {\responsiveIC} and IR mechanisms (\Cref{thm:characterization-full-info}). However, the conditions are in abstract terms and requires further definitions. For the purpose of presentation, we state here the sufficient condition that is easy to verify.} Note that it immediately implies that selling complete information is the optimal menu in this environment.

Throughout this section, we consider continuous distributions and assume that the pdf $f(\cdot)$ of the agent's type distribution is strictly positive on $[0,1]$ and differentiable on $(0,1)$. 




\begin{theorem}\label{cor:special_distributions}
For every $\theta\in [0,1]$, let $\varphi^-(\theta)=\theta f(\theta)+F(\theta)$ and $\varphi^+(\theta)=(\theta-1)f(\theta)+F(\theta)$. Suppose $\varphi^-(\cdot)$ and $\varphi^+(\cdot)$ are both monotonically non-decreasing. Suppose the payoff matrix satisfies that $\varphi^-(\frac{\widehat{p}}{u_{11}})\geq \varphi^+(1-\frac{\widehat{p}}{u_{2m}})$, where $\widehat{p}=\min\left\{\frac{(u_{11}-u_{12})u_{2m}}{u_{11}-u_{12}+u_{22}},\frac{(u_{2m}-u_{2,m-1})u_{11}}{u_{2m}+u_{1,m-1}-u_{2,m-1}}\right\}$.
Then 
selling complete information at any price $p$ such that $\varphi^-(\frac{p}{u_{11}})= \varphi^+(1-\frac{p}{u_{2m}})$ achieves the maximum revenue achievable by any {\responsiveIC} and IR mechanisms. 
\end{theorem}

Here is an interpretation of \Cref{cor:special_distributions}:
Both $\varphi^-(\cdot)$ and $\varphi^+(\cdot)$ are also considered in~\cite{BergemannBS2018design}. Intuitively, they can be viewed as the agent's ``virtual value'' when $q(\theta)\leq u_{11}-u_{2m}$ and when $q(\theta)>u_{11}-u_{2m}$ respectively, in a semi-informative mechanism. Both virtual value functions being non-decreasing is a standard regularity condition on the type distribution. The theorem applies to stardard distributions such as uniform, exponential and Gaussian distributions. See \Cref{subsec:characterization-full-info} for several examples.

    To understand the condition on the payoff matrix, we point out that if selling complete information at some price $p$ is the optimal {\responsiveIC} and IR mechanism, then the buyer's utility function must intersect with the IR curve at the first and last piece $(*)$. Otherwise we can add some extra experiments to strictly increase the revenue, while maintaining the same utility function and the {\responsiveIC} and IR property, contradicting with the optimality (see \Cref{lem:endpoints of the full info} in \Cref{subsec:characterization-full-info} for details). By some simple calculation, $\widehat{p}$ is exactly the largest price that satisfies $(*)$. The condition guarantees the existence of such a price $p\leq \widehat{p}$.
    
 Note that although the theorem applies to arbitrary number of actions, the condition itself only depends on the payoffs of the the first two and the last two actions. Thus if the condition is satisfied, selling complete information is always optimal for any choice of the payoffs for other actions.


To prove the theorem, we provide an exact characterization of the optimal semi-informative, {\responsiveIC} and IR mechanism, i.e., the optimal solution $\Bq^*=\{q^*(\theta)\}_{\theta\in [0,1]}$ of the program in \Cref{fig:continous-program-modified}, by Lagrangian duality (\Cref{thm:characterization}). It is a generalization of the characterization by Bergemann et al.~\cite{BergemannBS2018design} to $m\geq 3$ actions. 


%% file: general_setting.tex
\section{Selling Complete Information in Multiple ($n\geq 3$) States}\label{sec:multi_state}

In this section, we consider the environment with $n\geq 3$ states. In \Cref{sec:general_setting}, we prove that even in arguably the simplest environment -- matching utility environment with 3 states and 3 actions, the optimal revenue $\opt$ and the revenue by selling complete information $\optfi$ can have an arbitrarily large gap, i.e. there is no universal finite upper bound of the gap that holds for all type distributions.
Nonetheless, we prove in \Cref{sec:matching_utilities} that in the same environment, if the distribution is uniform, selling complete information is in fact the optimal mechanism. Proofs in this section are postponed to \Cref{sec:proof_multi_state}.

\subsection{Lower Bound Example for Matching Utilities}\label{sec:general_setting}

To avoid ambiguity, throughout this section we denote $\rev(\MM,D)$ the revenue of $\MM$ with respect to $D$, for any (not necessarily IC or IR) mechanism $\MM$, and any distribution $D$. Denote $\opt(D)$ and $\optfi(D)$ the optimal revenue and the maximum revenue by selling complete information, respectively. The main result of this section shows that in the matching utility environment with $n=3$ states and $m=3$ actions, the ratio $\frac{\opt(D)}{\optfi(D)}$ arbitrarily large for some distribution $D$. Recall that in a matching utility environment, $n=m$ and the payoffs satisfy $u_{ij}=\ind[i=j],\forall i,j$.

\begin{theorem}\label{thm:general-main}
Consider the matching utility environment with 3 states and 3 actions. For any integer $N$, there exists a distribution $D$ with support size $N$ such that
$$\frac{\opt(D)}{\optfi(D)}=\Omega(N^{1/7})$$
\end{theorem}

The following corollary directly follows from \Cref{thm:general-main} and \Cref{lem:1/k-approx for menus of size k}.

\begin{corollary}
Consider the matching utility environment with 3 states and 3 actions. For any integer $N$, there exists a distribution $D$ with support size $N$ such that any optimal menu has option size $\Omega(N^{1/7})$.
\end{corollary}

The proof of \Cref{thm:general-main} is adapted from the approach in \cite{hart2019selling}, which was originally used in problem of multi-item auctions. Their approach does not directly apply to our problem due to the non-linearity of the buyer's value function and the existence of the extra {\sigmaIC} and more demanding IR constraints in our problem.   

For any mechanism $\MM$, denote $\ratio(\MM)$ the largest ratio between $\rev(\MM,D)$ and $\optfi(D)$ among all distributions, i.e.,
$$\ratio(\MM)=\sup_D \frac{\rev(\MM,D)}{\optfi(D)}$$

For every $\theta\in \Theta$, denote $U(\theta)$ the gain in value of a buyer with type $\theta$, after receiving the fully informative experiment. Formally,
$U(\theta)=\sum_{i=1}^n\theta_i\cdot \max_ju_{ij}-\max_j\left\{\sum_i\theta_iu_{ij}\right\}$. Here $\theta_n=1-\sum_{i=1}^{n-1}\theta_i$.
In \Cref{thm:ratio_characterization}, we prove that $\ratio(\MM)$ can be written as a concrete term using $U(\cdot)$ and the payment $t(\cdot)$ of $\MM$.

\begin{theorem}\label{thm:ratio_characterization}
(Adapted from Theorem 5.1 of \cite{hart2019selling}) For any mechanism $\MM=\{E(\theta),t(\theta)\}_{\theta\in \Theta}$ (not necessarily IC or IR), we have
$\ratio(\MM)=\int_0^\infty \frac{1}{r(x)}dx$, where $r(x)=\inf\{U(\theta):\theta\in \Theta \wedge t(\theta)\geq x\}, \forall x\geq 0$.
\end{theorem}

\notshow{
\begin{lemma}\label{lem:ratio-opt and fi}
(Adapted from \cite{hart2019selling}) Given any integer $N$, a sequence of types $\{y_k\}_{k=1}^N$ in $\Delta_n$ and a sequence of experiments $\{\Pi^{(k)}\}_{k=1}^N$ such that

\begin{enumerate}
    \item 
    $\displaystyle\gap_k=\min_{0\leq q<k}\left\{V_{y_k}(\Pi^{(k)})-V_{y_k}(\Pi^{(q)})\right\}>0$. Here $\Pi^{(0)}$ satisfies $\Pi_{ij}^{(0)}=0,\forall i,j$.
    \item 
    $V_{y_k}(\Pi^{(k)})\geq \max_j\left\{\sum_iy_{k,i}u_{ij}\right\}+\gap_k$, where $y_{k,i}$ is the $i$-th entry of vector $y_k$.
\end{enumerate}

Then for any $\eps>0$, there exists a distribution $D$ with support size $N$, together with an IC and IR mechanism $\MM$ such that
$$\frac{\rev(\MM, D)}{\optfi(D)}\geq (1-\eps)\cdot \sum_{k=1}^N\frac{\gap_k}{U(y_k)}$$

Moreover, if both of the above properties hold by replacing $V_{y_k}(\Pi^{(q)})$ with $V_{y_k}^*(\Pi^{(q)})=\sum_{i,j}\Pi_{ij}^{(q)}\cdot y_{k,i}u_{ij}$ for every $q, k$, then the same inequality hold for some {\responsiveIC} and IR mechanism $\MM$.
\end{lemma}


\begin{prevproof}{Lemma}{lem:ratio-opt and fi}

Let $\{t_k\}_{k=1}^N$ be a sequence of positive numbers that increases fast enough so that: (i) $(t_k/\gap_k)\cdot U(y_k)$ is increasing, and (ii) $t_{k+1}/t_k\geq 1/\eps$ for all $1\leq k<N$. Such a sequence must exist since after $t_1,\ldots,t_{k-1}$ are decided, we can choose a large enough $t_k$ to satisfy both properties. Let $\delta>0$ be any number such that $1/\delta>\max_{k\in [N]}\{t_k/\gap_k\}$. For every $k\in [N]$, let $x_k=(\delta t_k/\gap_k)\cdot y_k$. Let $\xi_k=U(x_k)=(\delta t_k/\gap_k)\cdot U(y_k)$, which is increasing due to (i).  


For every $k\in [N]$, $\delta t_k/\gap_k<1$, thus $x_k\in \Delta_n^*$ as $y_k\in \Delta_n^*$. Consider the distribution $D$ with support $\{x_1,\ldots, x_N\}$.~\footnote{For every $k\in [N]$, the $k$-th type has probability $x_{k,i}$ for every state $\omega_i$, and probability $1-\sum_ix_{k,i}$ for the dummy state $\omega_0$.} $\Pr[\theta=x_k]=\xi_1\cdot(\frac{1}{\xi_k}-\frac{1}{\xi_{k+1}})$ for all $k\in [N]$, where $\xi_{N+1}=\infty$. 

Consider the following mechanism $\MM$. For every buyer's type $\theta$ in the support, the buyer chooses the experiment $\Pi^{(k^*)}$ where 
$k^*=\argmax_k\left\{ V_{\theta}(\Pi^{(k)})-\delta\cdot t_k\right\}$, and pays $\delta\cdot t_{k^*}$. In other words, he chooses the experiment that obtains the highest utility, when he follows the recommendation. Then clearly $\MM$ is {\responsiveIC}. Moreover, for every buyer's type $\theta=x_q$, we have   
$$V_{x_q}(\Pi^{(k^*)})-\delta\cdot t_{k^*}\geq V_{x_q}(\Pi^{(q)})-\delta\cdot t_q\geq \max_j\left\{\sum_ix_{q,i}u_{ij}\right\}$$
where the second inequality follows from property 2 of the statement (by dividing both sides with $\delta t_q/\gap_q$).  

Thus $\MM$ is IR. Now we compute $\rev(\MM,D)$. Consider any buyer's type $x_k$. For any $1\leq q<k$, by the definition of $\gap_k$, we have
$$V_{x_k}(\Pi^{(k)})-V_{x_k}(\Pi^{(q)})=(V_{y_k}(\Pi^{(k)})-V_{y_k}(\Pi^{(q)}))\cdot \frac{\delta t_k}{\gap_k}\geq \delta t_k>\delta t_k-\delta t_q$$

We notice that the utility of the buyer (with type $x_k$) after purchasing $\Pi^{(q)}$ is $V_{x_k}(\Pi^{(q)})-\delta\cdot t_q$.
Thus the above inequality implies that in $\MM$, the buyer must purchase an experiment $\Pi^{(q)}$ where $q\geq k$. Since $\{t_k\}_{k=1}^N$ is an increasing sequence, the buyer's payment is at least $\delta\cdot t_k$. By \Cref{thm:ratio_characterization},






\begin{align*}
\rev(\MM,D)&\geq \sum_{k=1}^N\delta\cdot t_k\cdot \xi_1\cdot(\frac{1}{\xi_k}-\frac{1}{\xi_{k+1}})=\delta\xi_1\cdot \sum_{k=1}^N \frac{t_k-t_{k-1}}{\xi_k}\\
&\geq (1-\eps)\xi_1\cdot \sum_{k=1}^N \frac{\delta t_k}{\xi_k}=(1-\eps)\cdot \sum_{k=1}^N\frac{\gap_k}{U(y_k)}\cdot \optfi(D)
\end{align*}

Here the second inequality follows from $t_{k+1}/t_k\geq 1/\eps$. The third inequality follows from the definition of $\xi_k$ and the fact that $\optfi(D)=\xi_1$, which is proved in \Cref{thm:ratio_characterization}. We include the proof here for completeness. Recall that a buyer with type $x$ will purchase the fully informative experiment at price $p$ if and only if  $U(x)\geq p$. Thus under distribution $D$, it's sufficient to consider the price $p=U(x_{k})=\xi_k$ for some $k\in [N]$. Since $\{U(x_{k})\}_{k\in [N]}$ is a strictly increasing sequence, then the buyer will purchase at price $p=\xi_k$ whenever $x\in \{x_{k},\ldots,x_N\}$, which happens with probability $\sum_{j=k}^N\xi_1\cdot(\frac{1}{\xi_k}-\frac{1}{\xi_{k+1}})=\frac{\xi_1}{\xi_k}$. Hence, $\optfi(D)=\xi_1$. The second part of the statement follows from a similar argument by noticing that the value of the buyer (with type $\theta$) after purchasing any experiment $\Pi$ is $V_{\theta}^*(\Pi)$, if she follows the recommendation.
\end{prevproof}


\begin{corollary}\label{cor:special environment E}
Consider the environment $\cE$. Given any integer $N$ and a sequence of types $\{y_k\}_{k=1}^N$ in $\Delta_2^*$ such that
\begin{enumerate}
    \item $\displaystyle\gap_k=1.9\cdot \min_{0\leq j<k} (y_k-y_j)\cdot y_k\in (0,0.09)$. Here $y_0=(0,0)$.
    \item $\frac{y_{k,1}}{y_{k,2}}\in [\frac{9}{10},\frac{10}{9}]$.
    \item $||y_k||_2=\sqrt{y_{k,1}^2+y_{k,2}^2}\geq 0.5$.
\end{enumerate}

Then for any $\eps>0$, there exists a distribution $D$ with support size $N$ such that
$$\frac{\opt(D)}{\optfi(D)}\geq (1-\eps)\cdot \sum_{k=1}^N\frac{\gap_k}{\min\{y_{k,1},y_{k,2}\}}\geq 2(1-\eps)\cdot \sum_{k=1}^N\gap_k$$
\end{corollary}

\begin{proof}
For every $k\in [N]$, let $g_k=1.9\cdot y_k$. By property 2, $g_{k,1}=1.9y_{k,1}\leq y_{k,1}+y_{k,2}\leq 1$. Similarly, $g_{k,2}\leq 1$. Consider the experiment $\Pi^{(k)}$ as follows:
\begin{table}[h]
\centering
\begin{tabular}{ c| c c} \label{tab:payoff matrix}
$\Pi^{(k)}$ & $1$ & $2$\\ 
  \hline
$\omega_1$ & $g_{k,1}$ & $1-g_{k,1}$ \\ 
 $\omega_2$ & $1-g_{k,2}$ &$g_{k,2}$
\end{tabular}
\end{table}

We will verify that under our choice of $\{y_k\}_{k\in[N]}$ and $\{\Pi^{(k)}\}_{k\in[N]}$, both properties in \Cref{lem:ratio-opt and fi} hold by replacing every $V_{y_k}(\Pi^{(q)})$ by $V_{y_k}^*(\Pi^{(q)})$. And thus we can apply the ``Moreover'' part of \Cref{lem:ratio-opt and fi}. Property 1 in~\Cref{lem:ratio-opt and fi} is equivalent to $$\gap_k=\min_{0\leq j<k}\{(g_{k,1}-g_{j,1})\cdot y_{k,1}+(g_{k,2}-g_{j,2})\cdot y_{k,2}\}=1.9\cdot \min_{0\leq j<k} (y_k-y_j)\cdot y_k>0.$$ 

Property 2 of \Cref{lem:ratio-opt and fi} can be simplified to the following inequality: $$y_{k}\cdot g_k\geq \max\{y_{k,1},y_{k,2}\}+\gap_k.$$ To prove the inequality, without loss of generality assume that $y_{k,1}\geq y_{k,2}$. Let $\delta=\frac{y_{k,2}}{y_{k,1}}\in [0.9,1]$. Then we have $||y_k||_2= \sqrt{1+\delta^2}\cdot y_{k,1}\geq 0.5$ which implies that $y_{k,1}\geq 0.5/\sqrt{1+\delta^2}\geq 0.5/\sqrt{2}>0.35$. Thus,
$$y_k\cdot g_k=1.9||y_k||_2^2\geq 0.95 ||y_k||_2= 0.95\cdot \sqrt{1+\delta^2}y_{k,1}\geq 1.27y_{k,1}\geq y_{k,1}+0.09>\max\{y_{k,1},y_{k,2}\}+\gap_k$$
where the first inequality is due to the fact that $||y_k||_2\geq 0.5$, the second inequality follows from $\delta\geq 0.9$, and the third inequality follows from $y_{k,1}>0.35$. As both properties in the statement of \Cref{lem:ratio-opt and fi} are satisfied, there exists a distribution $D$ with support size $N$, together with a {\responsiveIC} and IR mechanism $\MM$ such that
$$\frac{\rev(\MM, D)}{\optfi(D)}\geq (1-\eps)\cdot \sum_{k=1}^N\frac{\gap_k}{U(y_k)}.$$

Under environment $\cE$, $U(y_k)=y_{k,1}+y_{k,2}-\max\{y_{k,1},y_{k,2}\}=\min\{y_{k,1},y_{k,2}\}$, which is at most $\frac{1}{2}$ since $y_{k,1}+y_{k,2}\leq 1$. Thus to complete the proof, it remains to prove that $\MM$ is also IC, which implies $\rev(\MM,D)\leq \opt(D)$. In fact, given any true type $\theta$ and misreported type $\theta'\not=\theta$. Suppose the buyer purchases experiment $E$ (or $E'$) at price $t$ (or $t'$) at type $\theta$ (or $\theta'$). Since there are only 2 actions, for any {\sigmaIC} constraint, the corresponding mapping $\sigma$ can only be $(1,1)$, $(2,2)$ or $(2,1)$. When $\sigma=(1,1)$, the {\sigmaIC} constraint
$$\theta_1\pi_{11}(E)+\theta_2\pi_{22}(E)-t\geq \theta_1(\pi_{11}(E')+\pi_{12}(E'))u_{11}+\theta_2(\pi_{21}(E')+\pi_{22}(E'))u_{21}-t'=\theta_1-t'$$
is satisfied since the mechanism is IR and $t'\geq 0$. Similarly, the {\sigmaIC} constraint is satisfied when $\sigma=(2,2)$. When $\sigma=(2,1)$, the {\sigmaIC} constraint can be written as
\begin{align*}
\theta_1\pi_{11}(E)+\theta_2\pi_{22}(E)-t&\geq \theta_1(\pi_{11}(E')u_{12}+\pi_{12}(E')u_{11})+\theta_2(\pi_{21}(E')u_{22}+\pi_{22}(E')u_{21})-t'\\
&=\theta_1\pi_{12}(E')+\theta_2\pi_{21}(E')-t'
\end{align*}

According to the construction of $\MM$, the purchased experiment $E'$ is $\Pi^{(k)}$ for some $k\in [N]$. Again without loss of generality assume $y_{k,1}\geq y_{k,2}$. From the above analysis we have $g_{k,1}=1.9y_{k,1}\geq 1.9\times 0.35=0.665>0.5$. $g_{k,2}=1.9y_{k,2}\geq 1.9\times 0.9y_{k,1}>0.59>0.5$. Thus 
$$\theta_1\pi_{12}(E')+\theta_2\pi_{21}(E')=\theta_1(1-g_{k,1})+\theta_2(1-g_{k,2})<\theta_1g_{k,1}+\theta_2g_{k,2}=\theta_1\pi_{11}(E')+\theta_2\pi_{22}(E')$$
The {\sigmaIC} constraint is thus implied by the {\responsiveIC} constraint. Hence $\MM$ is IC. 
\end{proof}

}

Here is a proof sketch of \Cref{thm:general-main}. We first show in \Cref{lem:ratio-opt and fi-special} that: Given any sequence $\{y_k\}_{k=1}^N$ of vectors in $\Theta$ that satisfy all properties in the statement, we can construct a distribution $D$ with support size $N$, together with an IC, IR mechanism $\MM$ such that the ratio $\frac{\rev(\MM, D)}{\optfi(D)}$ has a lower bound that only depends on the sequence $\{y_k\}_{k=1}^N$. Properties 2 and 3 in the statement are carefully designed to guarantee that the constructed mechanism is IC and IR. Property 1 ensures the mechanism to obtain enough revenue. With \Cref{lem:ratio-opt and fi-special}, we then complete the proof of \Cref{thm:general-main} by constructing a valid sequence $\{y_k\}_{k=1}^N$ which induces a lower bound of $\Omega(N^{1/7})$.

\begin{lemma}\label{lem:ratio-opt and fi-special}
Consider the matching utility environment with 3 states and 3 actions.
Given any integer $N$ and a sequence of types $\{y_k=(y_{k,1},y_{k,2})\}_{k=1}^N$ in $\Theta$ such that
\begin{enumerate}
    \item $\displaystyle\gap_k=\min_{0\leq j<k} \left\{(y_{k,1}-y_{j,1})\cdot y_{k,1}+(y_{k,2}-y_{j,2})\cdot y_{k,2}\right\}\in (0,0.09)$. Here $y_0=(0,0,1)$.
    \item $\frac{y_{k,1}}{y_{k,2}}\in [\frac{9}{10},\frac{10}{9}]$.
    \item $||y_k||_2=\sqrt{y_{k,1}^2+y_{k,2}^2}\in [0.3, 0.4]$.
\end{enumerate}

Then for any $\eps>0$, there exists a distribution $D$ with support size $N$ such that
$$\frac{\opt(D)}{\optfi(D)}\geq (1-\eps)\cdot \sum_{k=1}^N\frac{\gap_k}{y_{k,1}+y_{k,2}}\geq \frac{3}{2}(1-\eps)\cdot \sum_{k=1}^N\gap_k$$
\end{lemma}

To complete the proof of \Cref{thm:general-main}, for any integer $N$, we construct a sequence of types $\{y_k\}_{k=1}^N$ that satisfies all three properties in the statement of \Cref{lem:ratio-opt and fi-special}, and $\gap_k=\Theta(k^{-6/7})$. Then by \Cref{lem:ratio-opt and fi-special}, there exists a distribution $D$ with support size $N$ such that (by choosing $\eps=\frac{1}{2}$) 
$$\frac{\opt(D)}{\optfi(D)}\geq \sum_{k=1}^N\gap_k=\Omega(\sum_{k=1}^Nk^{-6/7})=\Omega(N^{1/7})$$
The construction is adapted from Proposition 7.5 of \cite{hart2019selling}: All points $\{y_k\}_{k=1}^N$ are placed in a sequence of shells centered at $(0,0)$ with radius within the range of $[0.3, 0.4]$ (for Property 3). All points are placed in a thin circular sector close to the $45^\circ$ angle so that Property 2 is satisfied. See \Cref{sec:appx_multi_lower_bound} for the complete construction and proof.

%% file: multi_state_uniform.tex
\subsection{Matching Utility Environment with Uniform Distribution}\label{sec:matching_utilities}


As the main result of this section, we complement the result in \Cref{sec:general_setting}, by showing that in the matching utility environment with 3 states and 3 actions, selling complete information is indeed optimal if the type distribution is uniform.

\begin{theorem}\label{Thm:multi-state-uniform}
Consider the matching utility environment with $n=3$ states and $3$ actions, where the buyer has uniform type distribution. 
Then selling only the fully informative experiment at price 
$p=\frac{1}{3}$ achieves the maximum revenue among all IC, IR mechanisms.
\end{theorem}

The remaining of this section is dedicated to the proof of \Cref{Thm:multi-state-uniform}. We first provide a high-level plan of the proof. In \Cref{sec:upper_bound_full_info}, we considered a relaxation of our problem, which finds the optimal {\responsiveIC} and IR mechanism. Here we propose another relaxed problem of finding the optimal menu for any matching utility environment (\Cref{fig:transport-primal}). Then we construct a dual problem (\Cref{fig:transport-dual}) in the form of optimal transportation~\cite{villani2009optimal, daskalakis2017strong, rochet1998ironing}. This primal-dual framework provides a general approach to prove the optimality of any menu $\MM$: If we can construct a feasible dual that satisfies the complementary slackness conditions, then $\MM$ is the optimal primal solution to the relaxed problem, which implies that $\MM$ is the optimal menu. We apply this framework to the matching utility environment with 3 states and uniform type distribution, proving that selling complete information at $p=\frac{1}{3}$ is optimal.        


\subsubsection{Construction of the Primal and Dual Problem}

By \Cref{lem:responsive menu}, we will focus on responsive mechanisms $\MM=\{(E(\theta),t(\theta))\}_{\theta\in \Theta}$. For any $\theta\in \Theta$, denote $\pi_{ij}(\theta)$ the $(i,j)$-entry of experiment $E(\theta)$. We also use $\pi_i(\theta)$ to represent $\pi_{ii}(\theta)$. For every measurable set $S\subseteq \Theta$, denote $\vol(S)=\int_{S}1d\theta$ the volume of $S$. 
We first prove the following lemma that applies to any matching utility environment. 

\begin{lemma}\label{lem:one-column-has-1}
In any matching utility environment, there is an optimal responsive mechanism $\MM$ such that: for every $\theta\in \Theta$, there exists $i\in [n]$ such that $\pi_i(\theta)=1$. 
\end{lemma}

By \Cref{lem:one-column-has-1}, we focus on all mechanisms that are responsive and satisfy: for every $\theta\in \Theta$, there exists $i\in [n]$ such that $\pi_i(\theta)=1$. Recall that $V_{\theta}(E)$ (or $V_{\theta}^*(E)$) is the buyer's value of experiment $E$ (or the value if she follows the recommendation) at type $\theta$. For any $\theta$, denote $G(\theta)=V_{\theta}(E(\theta))-t(\theta)$ the buyer's utility at type $\theta$. Since the mechanism is responsive, we have $$G(\theta)=V_{\theta}^*(E(\theta))-t(\theta)=\sum_{i=1}^{n-1}\pi_i(\theta)\cdot \theta_i+\pi_n(\theta)\cdot \left(1-\sum_{i=1}^{n-1}\theta_i\right)-t(\theta)$$  

We prove the following observation using the IC constraints.

\begin{observation}\label{obs:partial G}
For any $i\in [n-1]$, $\frac{\partial G(\theta)}{\partial \theta_i}=\pi_i(\theta)-\pi_n(\theta)$.
\end{observation}

Let $\nabla G(\theta)=(\frac{\partial G(\theta)}{\partial \theta_1},\ldots, \frac{\partial G(\theta)}{\partial \theta_{n-1}})$. By \Cref{obs:partial G} and the fact that $\max_{i\in [n]}\{\pi_i(\theta)\}=1,\forall \theta\in \Theta$, we have
$$t(\theta)=-G(\theta)+\nabla G(\theta)\cdot \theta+\pi_n(\theta)=-G(\theta)+\nabla G(\theta)\cdot \theta+1-\max\left\{\frac{\partial G(\theta)}{\partial \theta_1},\ldots, \frac{\partial G(\theta)}{\partial \theta_{n-1}}, 0\right\}$$

In \Cref{obs:c-function}, we prove a necessary condition under which a function $G(\cdot)$ is derived by some IC, IR and responsive mechanism. 

\begin{definition}\label{def:transport-cost}
For every type $\theta, \theta'\in \Theta$, define $c(\theta, \theta')=\max_{E}\left\{V_{\theta}(E)-V_{\theta'}(E)\right\}$, where the maximum is taken over all possible experiments with $n$ states and $n$ actions.  
\end{definition}

\begin{lemma}\label{obs:c-function}
For any IC, IR and responsive mechanism $\MM$, let $G(\cdot)$ be the buyer's utility function. Then $G(\cdot)$ is convex and satisfies $G(\theta)-G(\theta')\leq c(\theta,\theta'), \forall \theta,\theta'\in \Theta$. Moreover, $G(\textbf{0})=1$.
\end{lemma}

We notice that given any convex function $G$ that satisfies $G(\theta)-G(\theta')\leq c(\theta,\theta'),\forall \theta,\theta'$. The function $\widehat{G}(\theta)=G(\theta)-G(\textbf{0})+1$ satisfies all properties in \Cref{obs:c-function}. Thus a relaxation of the problem of finding the revenue-optimal menu can be written as the optimization problem in \Cref{fig:transport-primal}. Here we replace $G(\theta)$ by $G(\theta)-G(\textbf{0})+1$ and remove the constraint $G(\textbf{0})=1$. 

\begin{figure}[ht!]
\colorbox{MyGray}{
\begin{minipage}{.98\textwidth}
$$\sup_{\substack{\text{$G$ is convex} \\ G(\theta)-G(\theta')\leq c(\theta,\theta'), \forall \theta, \theta'}}\int_\Theta \left[-G(\theta)+\nabla G(\theta)\cdot \theta-\max\left\{\frac{\partial G(\theta)}{\partial \theta_1},\ldots, \frac{\partial G(\theta)}{\partial \theta_{n-1}}, 0\right\}+G(\textbf{0})\right]f(\theta)d\theta $$
\end{minipage}}
\caption{The Relaxed Problem of the Optimal Menu }\label{fig:transport-primal}
\end{figure}

In the next step, we construct a dual problem in the form of optimal transportation.
We first introduce some useful notations in the measure theory.

\begin{itemize}
    \item $\Gamma(\Theta)$ and $\Gamma_+(\Theta)$ denote the signed and unsigned (Radon) measures on $\Theta$.
    \item Given any unsigned measure $\gamma\in \Gamma_+(\Theta\times \Theta)$, denote $\gamma_1, \gamma_2$ the two marginals of $\gamma$. Formally, $\gamma_1(A)=\gamma(A\times \Theta)$ and $\gamma_2(A)=\gamma(\Theta\times A)$ for any measurable set $A\subseteq \Theta$. \item Given any signed measure $\mu\in \Gamma(\Theta)$, denote $\mu_+, \mu_-$ the positive and negative parts of $\mu$ respectively, i.e. $\mu=\mu_+-\mu_-$.
    \item Given a measure $\mu\in \Gamma(\Theta)$. A \emph{mean-preserving spread} operation in set $A\subseteq \Theta$ is a sequence of the following operation: picking a positive point mass on $\theta\in A$, splitting it into several pieces, and sending these pieces to multiple points in $A$ while preserving the center of mass.
\end{itemize}

We define \emph{strongly convex dominance} similar to the notion of convex dominance in \cite{shaked2007stochastic, daskalakis2017strong}. The main difference here is that only the mean-preserving spread operation is allowed during the transformation between two measures. \footnote{In \cite{shaked2007stochastic, daskalakis2017strong}, the notion of convex dominance is less restricted: $\mu$ convex dominates $\mu'$ if $\int_{\Theta}gd\mu\geq \int_{\Theta}gd\mu'$ for any continuous, \textbf{non-decreasing}, convex function $g$ on $\Theta$. Under this notion, $\mu'$ can be transformed to $\mu$ by either performing the mean-preserving spread, or sending a positive mass to coordinatewise larger points. In our notion, the inequality holds for any convex function that is not necessarily non-decreasing. Thus the second operation may make $\int gd\mu'$ smaller and is not allowed during the transformation.}


\begin{definition}
Given any $\mu, \mu'\in \Gamma(\Theta)$, we say that $\mu$ \emph{strongly convex dominates} $\mu'$ (denoted as $\mu\cuxd \mu'$) if $\mu'$ can be transformed to $\mu$ by performing a mean-preserving spread. By Jensen's inequality, if $\mu\cuxd \mu'$, then $\int_{\Theta}gd\mu\geq \int_{\Theta}gd\mu'$ for any continuous, convex function $g$ on $\Theta$. 
\end{definition}

In the following definition, we define a measure $\mu^P$ for any partition $P=(\Theta_1,\ldots,\Theta_n)$ of $\Theta$ such that every component $\Theta_i$ is compact. Denote $\cP$ the set of all partitions. 

\begin{definition}\label{def:measure}
Given any partition $P=(\Theta_1,\ldots,\Theta_n)$ of $\Theta$ such that every $\Theta_i$ is compact.~\footnote{Here we assume each $\Theta_i$ to be compact such that a intergral on $\Theta_i$ is well-defined. Two $\Theta_i$s may overlap on their boundaries. We still refer to $(\Theta_i,\ldots,\Theta_n)$ a partition of $\Theta$ as the set of overlapping points has 0 measure with respect to the density $f$.} We define a measure $\mu^P\in \Gamma(\Theta)$ as follows:
\begin{align*}
\mu^P(A)=\ind_A(\textbf{0})&+\int_{\partial \Theta}\ind_A(\theta)f(\theta)(\theta\cdot \bm{n})d\theta-n\cdot\int_{\Theta}\ind_A(\theta)f(\theta)d\theta-\sum_{i=1}^{n-1}\int_{\partial \Theta_i}\ind_A(\theta)f(\theta)(\bm{e}_i\cdot \bm{n}_i)d\theta\\
&-\int_{\Theta}\ind_A(\theta)\cdot (\nabla f(\theta)\cdot \theta)d\theta+\sum_{i=1}^{n-1}\int_{\Theta_i}\ind_A(\theta)\cdot(\bm{e}_i\cdot \nabla f(\theta))d\theta  
\end{align*}
for any measurable set $A$. Here $\partial \Theta$ (or $\partial \Theta_i$) is the boundary of $\Theta$ (or $\Theta_i$). $\bm{n}$ (or $\bm{n}_i$) is the outward pointing unit normal at each point on the boundary $\partial \Theta$ (or $\partial \Theta_i$). 
$\bm{e}_i$ is the $i$-th unit vector of dimension $n-1$. $\ind_A(\theta)=\ind[\theta\in A]$. 
\end{definition}

We prove the following lemma by applying the divergence theorem. 
\begin{lemma}\label{lem:divergence}
For any differentiable function $G$, we have
$$\int_\Theta G(\theta)d\mu^P=\int_{\Theta}\left[-G(\theta)+\nabla G(\theta)\cdot \theta+G(\textbf{0})\right]f(\theta)d\theta-\sum_{i=1}^{n-1}\int_{\Theta_i} \frac{\partial G(\theta)}{\partial \theta_i}f(\theta)d\theta$$   
Choosing $G(\theta)=1,\forall \theta\in\Theta$ implies that $\mu^P(\Theta)=0$.
\end{lemma}

Now we are ready to define the dual problem (\Cref{fig:transport-dual}).

\begin{figure}[ht!]
\colorbox{MyGray}{
\begin{minipage}{.98\textwidth}
$$\inf_{\substack{P\in \cP, \gamma\in \Gamma_+(\Theta\times\Theta)\\ \gamma_1-\gamma_2\cuxd \mu^P}} \int_{\Theta\times \Theta} c(\theta,\theta')d\gamma(\theta,\theta')$$
\end{minipage}}
\caption{The Dual Problem in Form of Optimal Transportation}\label{fig:transport-dual}
\end{figure}

To have a better understanding of the dual problem, we prove the following lemma which shows that the weak duality holds. 

\begin{lemma}\label{lem:weak-duality}
For any feasible solution $G$ to the primal (\Cref{fig:transport-primal}) and any feasible solution $(P=(\Theta_1,\ldots, \Theta_n), \gamma)$ to the dual (\Cref{fig:transport-dual}),
\begin{align*}
\int_\Theta \left[-G(\theta)+\nabla G(\theta)\cdot \theta-\max\left\{\frac{\partial G(\theta)}{\partial \theta_1},\ldots, \frac{\partial G(\theta)}{\partial \theta_{n-1}}, 0\right\}+G(\textbf{0})\right]f(\theta)d\theta
\leq \int_{\Theta\times \Theta} c(\theta,\theta')d\gamma(\theta,\theta')
\end{align*}
The inequality achieves equality if and only if all of the following conditions are satisfied:
\begin{enumerate}
    \item $\frac{\partial G(\theta)}{\partial \theta_i}=\max\left\{\frac{\partial G(\theta)}{\partial \theta_1},\ldots, \frac{\partial G(\theta)}{\partial \theta_{n-1}}, 0\right\}$ for every $i\in [n-1], \theta\in \Theta_i$; $\max_{i\in [n-1]}\frac{\partial G(\theta)}{\partial \theta_i}\leq 0$ for every $\theta\in \Theta_n$. 
    \item $\gamma(\theta,\theta')>0\implies G(\theta)-G(\theta')=c(\theta,\theta')$.
    \item $\int_{\Theta}G(\theta)d(\gamma_1-\gamma_2)=\int_{\Theta}G(\theta)d\mu^P$. In other words, $\mu^P$ (in \Cref{def:measure}) can be transformed to $\gamma_1-\gamma_2$ by performing a mean-preserving spread operated in a region where $G$ is linear.
\end{enumerate}
\end{lemma}

An important application of \Cref{lem:weak-duality} is that: Given a menu $\MM$ in any matching utility environment, we can certify the optimality of $\MM$ by constructing a feasible dual $(P, \gamma)$ that satisfies all conditions in \Cref{lem:weak-duality} with respect to the utility function $G(\cdot)$ induced by $\MM$. This is because $G(\cdot)$ being an optimal solution to the primal problem in \Cref{fig:transport-primal} immediately implies that $\MM$ is the optimal menu since the primal problem is a relaxation of our problem. 

\begin{theorem}\label{thm:certify-optimality}
Given any matching utility environment and any menu $\MM$. Let $G(\cdot)$ be the buyer's utility function induced by $\MM$. If there exists a feasible dual solution
$(P, \gamma)$ to the problem in \Cref{fig:transport-dual} such that $G$ and $(P, \gamma)$ satisfies all conditions in the statement of \Cref{lem:weak-duality}, then $\MM$ is the optimal menu in this environment.
\end{theorem}

\subsubsection{Proof of \Cref{Thm:multi-state-uniform}}

Now we focus on the special case $n=3$ and give a proof of \Cref{Thm:multi-state-uniform} using \Cref{thm:certify-optimality}. The type space $\Theta=\{(\theta_1,\theta_2)\in [0,1]^2 \mid \theta_1+\theta_2\leq 1\}$ is a triangle.
Let $\MM^*=\{E^*(\theta),t^*(\theta)\}_{\theta\in\Theta}$ be the mechanism that only sells the full information at price 
$p=\frac{1}{3}$. Let $G^*(\theta)=V_{\theta}^*(E^*(\theta))-t^*(\theta)$ be the buyer's utility function in $\MM^*$. By \Cref{thm:certify-optimality}, it suffices to construct a feasible dual $(P^*,\gamma^*)$ that satisfies all conditions in \Cref{lem:weak-duality}.

Denote $\pi_i^*(\theta)$ the $(i,i)$-entry of $E^*(\theta)$. We consider the following partition $P^*=(\Theta_1^*,\Theta_2^*, \Theta_3^*)$ of $\Theta$: $\Theta_i^*=\{\theta\in\Theta \mid \theta_i\geq \theta_j,\forall j\in \{1, 2, 3\}\}$, here $\theta_3=1-\theta_1-\theta_2$. See \Cref{fig:uniform_1} for an illustration.
In uniform distribution, the density $f(\theta)=2$ is a constant for all $\theta\in \Theta$. Thus we can simplify the description of $\mu^{P^*}$ as in \Cref{obs:opt-mu-P}. 

\begin{observation}\label{obs:opt-mu-P}
By \Cref{def:measure}, the measure $\mu^{P^*}$ with respect to the partition $P^*$ is the sum of:
\begin{itemize}
    \item A point mass of $+1$ at $\textbf{0}$.
    \item For each $(i,j)=(1,2), (1,3), (2,3)$, a mass of $\frac{2}{3}$ uniformly distributed on the line segment $S_{ij}=\{\theta\in \Theta\mid \theta_i=\theta_j\geq \theta_k\}$. Here $k$ is the index other than $i,j$. $\theta_3=1-\theta_1-\theta_2$.
    \item A mass of $-3$ uniformly distributed through out $\Theta$.
\end{itemize}
\end{observation}

\begin{figure}[ht]
	\centering
	\begin{subfigure}[b]{0.48\textwidth}
	    \centering
	    \includegraphics[width=0.65\textwidth]{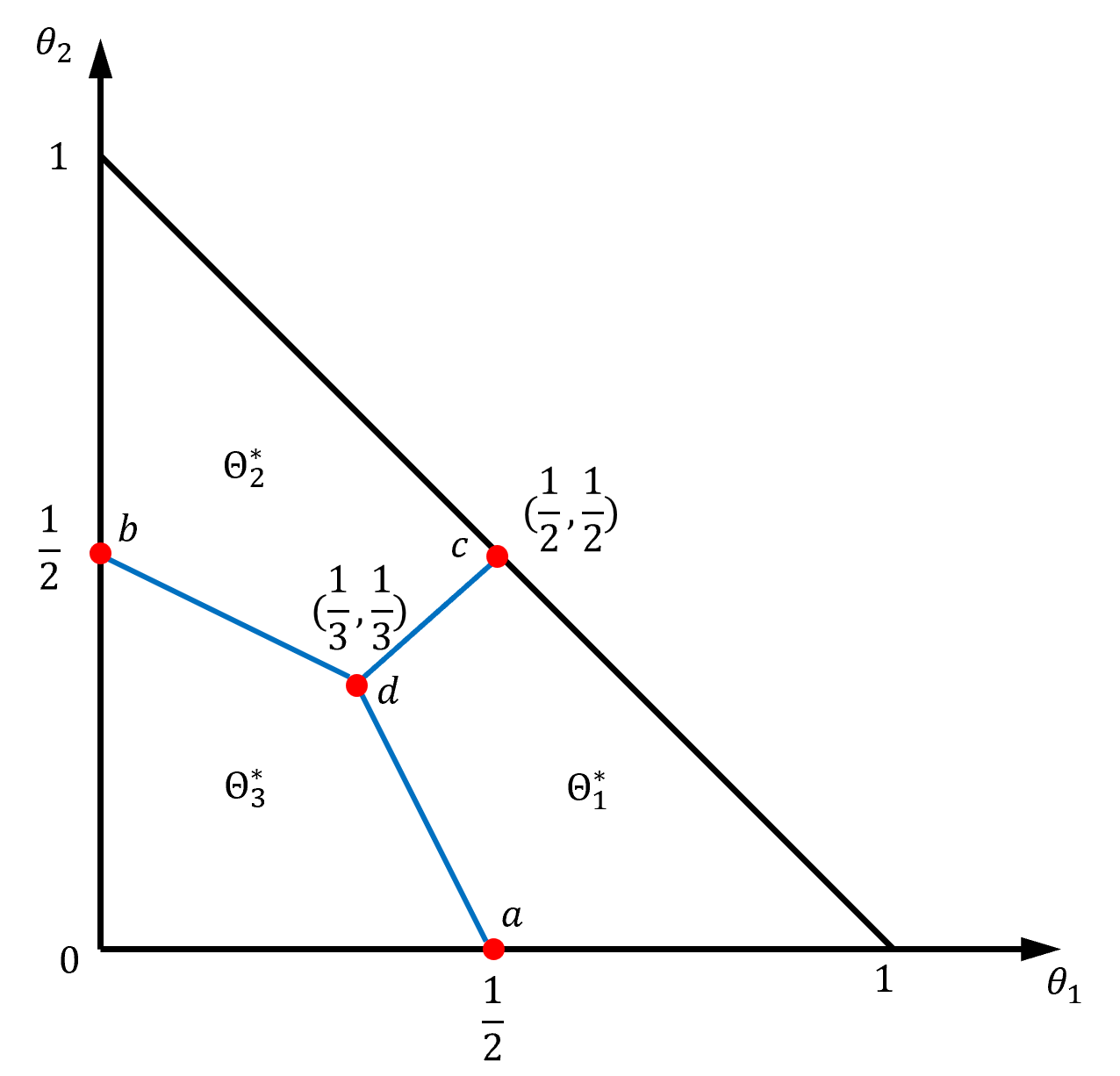}
	    \caption{Illustration of $P^*$. Line segments $ad$, $bd$, $cd$ correspond to $S_{13}$, $S_{23}$, $S_{12}$ respectively.} 
	    \label{fig:uniform_1}
	\end{subfigure}
	\hfill
	\begin{subfigure}[b]{0.48\textwidth}
	    \centering
	    \includegraphics[width=0.65\textwidth]{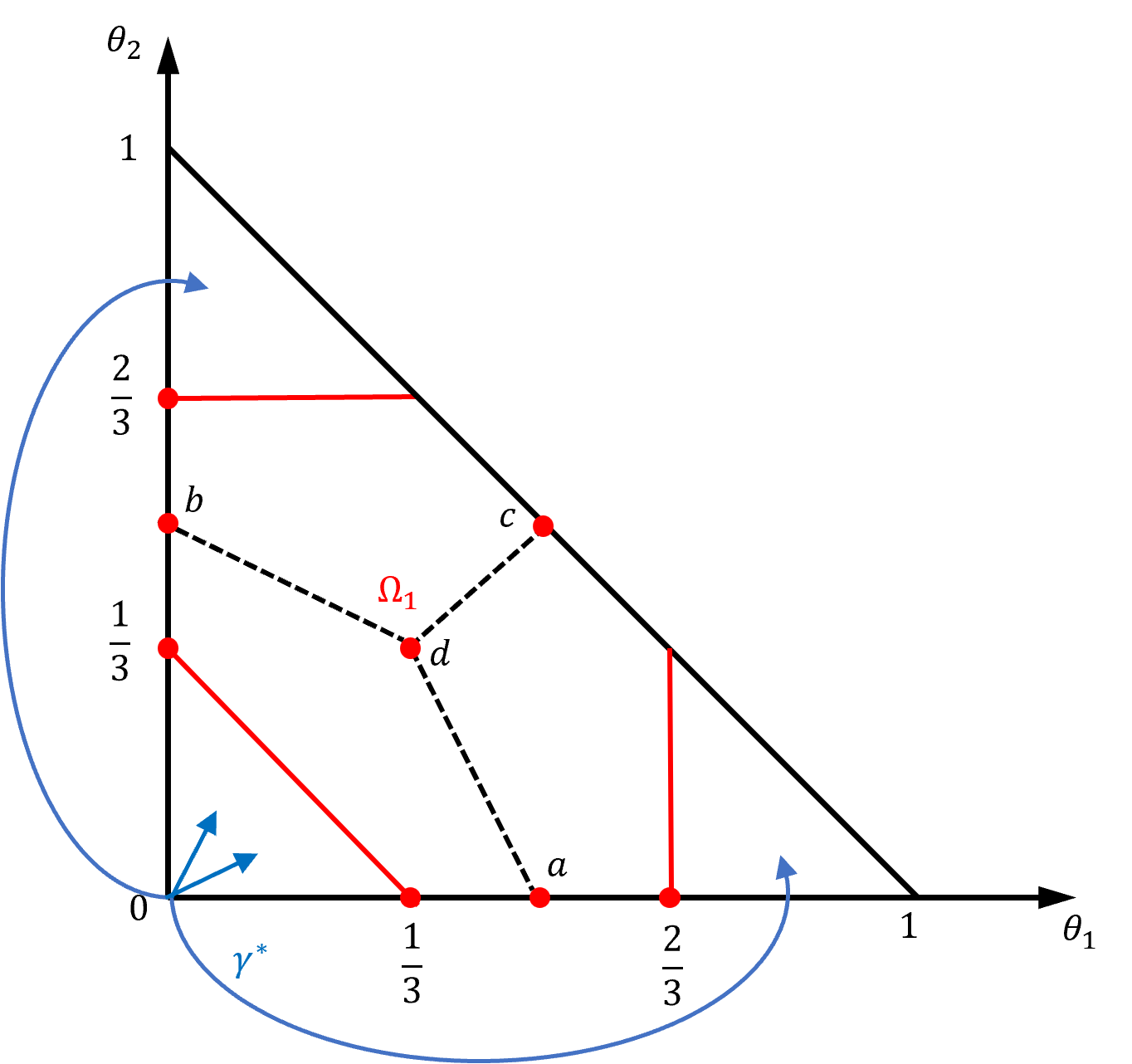}
	    \caption{$\Omega_1$ is the hexagon separated by the three red lines. $\gamma^*$ has a positive density between $\textbf{0}$ and each point in $\Omega_2$.}
	    \label{fig:uniform_2}
	\end{subfigure}
\caption{Illustrations of the optimal dual when $n=3$}\label{fig:uniform}
\end{figure}

Now we construct the measure $\gamma^*\in \Gamma_+(\Theta\times \Theta)$ as follows: Let $\Omega_1=\{\theta\in \Theta:1-\max_{i\in [n]}\theta_i\geq p\}$ be the set of types where the buyer purchases the full information in $\MM^*$. Let $\Omega_2=\Theta\backslash \Omega_1$. For any measurable set $A, B\subseteq \Theta$, define $\gamma(A\times B)=\ind_A(\textbf{0})\cdot \int_{\Omega_2}\ind_B(\theta)d\theta/\vol(\Omega_2)$. See \Cref{fig:uniform_2} for an illustration of the notations. To verify that $G^*$ and $(P^*, \gamma^*)$ are the optimal primal and dual solution respectively, it suffices to prove the following lemma.

\begin{lemma}\label{lem:verify-optimality}
$G^*$ and $(P^*, \gamma^*)$ are feasible solutions to the primal (\Cref{fig:transport-primal}) and dual problem (\Cref{fig:transport-dual}) respectively. Moreover, they satisfy all the conditions in the statement of \Cref{lem:weak-duality}.
\end{lemma}

The primal feasibility and the first two conditions are relatively easy to verify. To prove the dual feasibility and Condition 3, denote $\gamma_1^*, \gamma_2^*$ the marginals of $\gamma$. By definition of $\gamma$, $\mu^{P^*}-(\gamma_1^*-\gamma_2^*)$ contains only a mass of $-2$ uniformly distributed throughout out the hexagon $\Omega_1$, and a mass of $\frac{2}{3}$ uniformly distributed in each line segment $ad, bd, cd$. Due to the special geometric shape of the hexagon, we can transform the positive mass on the line segments to the whole region $\Omega_1$ via mean-preserving spread, to ``zero-out'' the negative mass. This proves both the dual feasibility and Condition 3 since $G^*(\theta)=\frac{2}{3}$ is constant throughout $\Omega_1$. 


%% file: conclusion.tex
\section{Conclusion}\label{sec:conclusion}

We considered the problem of selling information to a data-buyer with
private information. 
The (approximately) optimal mechanism
typically offers menu of options at different prices to the data-buyer. We
showed that this revenue maximization problem shares some features with the
problem of selling multiple items to a single buyer. Yet, the problem of
finding the optimal menu of information structure is richer in two important
aspects. First, every item on the menu is information structure, thus matrix
of signals given the true state, and second, the individual rationality
constraint varies with the private information of the agent. Thus, the
choice set as well as the set of constraints is richer than the standard
multi-item. 
Our analysis thus focused on establishing lower and upper bounds
for the optimality of the complete information structure, which is a natural information product and in a sense the equivalence of the grand bundle in the mutli-item problem. 

We established a lower bound on
the cardinality of the optimal menu in the binary-state environment. An interesting future direction is to show an upper bound of the cardinality of the optimal menu. Although we showed an $O(m)$ approximation ratio for selling complete information in this single-dimensional environment, no finite upper bound (in a function of $m$) of the cardinality that holds for all type distribution is known.

In general environments, we offered a primal-dual approach to prove the optimality of a given menu. We used the approach to show the optimality of selling complete information in a special environment. We believe that this approach should be productive to further characterize the optimal menu in other environments.

%% file: additional_prelim.tex
\section{Additional Preliminaries}\label{sec:more_prelim}


\begin{definition}[Responsive Experiment \cite{BergemannBS2018design}]\label{def:responsive experiment}
A buyer type $\theta$ is \emph{responsive} to an experiment $E$ if: There is a one-to-one mapping $g:S\to [m]$ from the signal set to the action space such that, for every $k\in [m]$, the buyer chooses the action $k$ when receiving signal $s$ iff $g(s)=k$.   
\end{definition}

\begin{lemma}[Adapted from Proposition 1 from \cite{BergemannBS2018design}]\label{lem:responsive menu}
   A mechanism $\MM=\{(E(\theta),t(\theta))\}_{\theta\in \Theta}$ is \emph{responsive} if every buyer type $\theta$ is responsive to $E(\theta)$. Then for any mechanism $\MM$, there exists a responsive mechanism $\MM'=\{E'(\theta),t(\theta)\}_{\theta\in \Theta}$ with the same payment rule, such that for every type $\theta$, the buyer has the same value for both experiments $E(\theta)$ and $E'(\theta)$.
\end{lemma}

\begin{lemma}\label{lem:full-info-contained}\cite{BergemannBS2018design}
Any optimal menu contains a fully informative experiment. 
\end{lemma}

%% file: proof_lower_bound.tex
\section{Missing Details from Section~\ref{sec:lower_bound_full_info}}\label{sec:proof_lower_bound}

\begin{prevproof}{Lemma}{lem:ER-dist}
Since $h$ is strictly decreasing in $[a,b]$, let $h^{-1}$ be its inverse function. 
Define $F$ as follows: $F(\theta)=1-\frac{c}{1-h(\theta)},\forall \theta\in [a,b)$; $F(b)=1$. Here $c=1-h(a)$ to ensure that $F(a)=0$. Consider the fully informative experiment with any price $p$. The utility for purchasing this experiment is $1-p$ regardless of the buyer's prior. Thus the buyer purchases the fully informative experiment iff $h(\theta)\leq 1-p$, which happens with probability $1-F(h^{-1}(1-p))=c/p$. Thus $\optfi=c$.
Moreover, $f(\theta)=-c\cdot \frac{h'(\theta)}{(1-h(\theta))^2}, \forall \theta\in [a,b)$.~\footnote{As $h(\cdot)$ is a convex function, it is also continuously differentiable on except countably many points. For those points, we can choose $h'(\theta)$ to be any subdifferential at $\theta$.} We have
$$\int_a^bf(\theta)\cdot (1-h(\theta))d\theta= -c\cdot \int_a^b\frac{h'(\theta)}{1-h(\theta)}d\theta=c\cdot \log(1-h(\theta))\Big|_a^b=c\cdot \log\left(\frac{1-h(b)}{1-h(a)}\right)$$
\end{prevproof}

\begin{prevproof}{Lemma}{lem:mech-m}
We first bound the revenue of $\MM$. Notice that for every $1\leq i\leq m-2$,
\begin{align}\label{equ:p_i upper bound}
p_i=\sum_{j=0}^{i-1}h_j+\theta_il_i=(2^{i}-1)\cdot \eps^{2^m}-\eps^{2^i}\cdot (\eps^{2^m}+\sum_{1\leq j<i}2^j\cdot \eps^{2^m-2^j})\leq (2^i-1)\cdot \eps^{2^m}
\end{align}
On the other hand, when $\eps<\frac{1}{2^m}$, we have
\begin{align}\label{equ:p_i lower bound}
p_i&=(2^i-1)\cdot \eps^{2^m}-\eps^{2^m+2^i}-\sum_{1\leq j<i}2^j\cdot \eps^{2^m+2^i-2^j}\\
&\geq(2^i-1)\cdot \eps^{2^m}-\eps^{2^m+2}\cdot(1+\sum_{1\leq j<i}2^j)\geq(2^i-1)\cdot \eps^{2^m}-\eps^{2^m+1}  \nonumber
\end{align}
Moreover,
$$1-u(\theta_{i+1})=\sum_{j=0}^{i}h_j=(2^{i+1}-1)\cdot \eps^{2^m}$$

Thus for every $\theta$ such that $\theta_i<\theta\leq \theta_{i+1}$, the payment of $\MM$ is at least $$\min\{p_i,p_{i-1},p_{i-2}\}=p_{i-2}\geq \frac{1}{9}(1-u(\theta_{i+1}))\geq \frac{1}{9}(1-u(\theta)).$$ 
Thus we have $\rev(\MM)\geq \frac{1}{9}\cdot \int_{\theta_1}^{\theta_{m-1}}f(\theta)\cdot (1-u(\theta))d\theta$.

Next we prove that $\MM$ is IR and $\delta$-IC. Firstly, if the buyer follows the recommendation of experiment $E_i$, his expected utility is $\theta(1+l_i)+1-\theta-p_i=1+\theta l_i-p_i=1-\sum_{j=0}^{i-1}h_j+(\theta-\theta_i)l_i$, {which is $u(\theta_i)$ and $u(\theta_{i+1})$ if we choose $\theta$ to be $\theta_i$ and $\theta_{i+1}$ respectively. Since $u(\cdot)$ is linear on $(\theta_i,\theta_{i+1}]$, the buyer' expected utility for receiving $(E_i,p_i)$ is exactly $u(\theta)$.} Thus her utility by reporting truthfully is at least $u(\theta)\geq 0$, and $\MM$ is IR.

Now fix any $\theta\in [\theta_1,\theta_{m-1}]$. Let $i$ be the unique number such that $\theta_i<\theta\leq \theta_{i+1}$. Suppose that the buyer misreports and receives experiment $E_j$ for $1\leq j\leq m-2$. We notice that when the experiment recommends action 1, the state $\omega_1$ is fully revealed and the buyer will follow the recommendation, choosing action 1. If the buyer chooses action $m-r$ when being recommended action $m$, for $0\leq r\leq m-1$, then his utility for buying experiment $j$ is $$U_{j,r}(\theta):=\theta(1+l_j)+\theta(-l_j)u_{1,m-r}+(1-\theta)u_{2,m-r}-p_j=u_{2,m-r}+\theta(1+l_j-l_ju_{1,m-r}-u_{2,m-r})-p_j.$$ When $1\leq r\leq m-2$, using the fact that $u_{2,m-r}=1-p_r$ and $u_{1,m-r}=1-p_r+l_r$, we can rewrite $U_{j,r}(\theta)=1+\theta(p_r+l_j(p_r-l_r))-p_r-p_j$.

We argue that for every $\theta\in (\theta_i,\theta_{i+1}]$, the following inequality holds.
\begin{equation}\label{equ:lower_bound}
\max_{j\in \{1,\ldots,m-2\},r\in \{0,\ldots,m-1\}}U_{j,r}(\theta)\leq \max_{j'\in \{i,i-1,i-2\},r'\in\{0,\ldots,m-1\}}U_{j',r'}(\theta)+7\cdot \eps^{2^m+1}. ~\footnote{In the RHS of the inequality, $j$ can only take value $i,i-1$ when $i\leq 2$, and can only take value $i$ when $i=1$.}
\end{equation}
Since a buyer with type $\theta$ receives utility $\max_{j'\in \{i,i-1,i-2\},r'}U_{j',r'}(\theta)$ in $\MM$ if he reports truthfully, Inequality~\eqref{equ:lower_bound} guarantees that $\MM$ is $\delta$-IC (recall that $\delta=7\cdot \eps^{2^m+1}$), as the buy can obtain utility at most $\max_{j\in \{1,\ldots,m-2\},r}U_{j,r}(\theta)$ by misreporting his type.

Now we prove that for every $j\in\{1,\ldots,m-2\}, r\in \{0,\ldots,m-1\}$, $U_{j,r}(\theta)$ is no more than the RHS of Inequality~\eqref{equ:lower_bound} by a case analysis.

\begin{enumerate}
    \item $r=0$. For every $j\in \{1,\ldots,m-2\}$, $U_{j,0}(\theta)=V_{\theta}^*(E_j)-p_j$ is the agent's utility for purchasing $E_j$, if he follows the recommendation. {By the definition of the IR curve $u(\cdot)$, the function $U_{\ell,0}(\cdot)$ coincides with $u(\cdot)$ on the interval $[\theta_\ell,\theta_{\ell+1}]$ for all $\ell\in \{1,\ldots, m-2\}$. Since $u(\cdot)$ is a convex function and $\theta\in(\theta_i,\theta_{i+1}]$, $U_{i,0}(\theta)\geq U_{j,0}(\theta)$ for all $j\in \{1,\ldots,m-2\}$.}
    \item $r=m-1$. For every $j\in \{1,\ldots,m-2\}$, since $\eps<\frac{1}{2^m}$,  $U_{j,m-1}(\theta)=\theta-p_j\leq \theta\leq \theta_{m-1}<\eps^{2^{m-1}-1}.$
On the other hand, $U_{i,0}(\theta)=1+\theta\ell_i-p_i\geq 1-2^m\cdot \eps^{2^m}>1-\eps^{2^m-1}>U_{j,m-1}(\theta)$. 
\item $1\leq r\leq m-2$. Recall $p_r-(2^r-1)\cdot \eps^{2^m}\in [-\eps^{2^m+1},0]$ for every $r\in \{1,...,m-2\}$ (Inequality~\eqref{equ:p_i lower bound}). For every $j,r\in \{1,...,m-2\}$, $U_{j,r}(\theta)$ can be further bounded as follows: 
\begin{align*}
    U_{j,r}(\theta)&\geq {1-p_r-p_j+\theta\cdot [(2^r-1)\eps^{2^m}-\eps^{2^m+1}-\eps^{2^j}((2^r-1)\eps^{2^m}+\eps^{2^r})]}\\
    &\geq 1-p_r-p_j-\theta\cdot \eps^{2^j+2^r}-2\cdot \eps^{2^m+1}\\
    &\geq 1-(2^j+2^r-2)\cdot \eps^{2^m}-\theta\cdot \eps^{2^j+2^r}-{2\cdot \eps^{2^m+1}}
\end{align*} 
{The first inequality follows from the upper bound (Inequality~\eqref{equ:p_i upper bound}) and the lower bound (Inequality~\eqref{equ:p_i lower bound}) of $p_i$.} The second inequality follows from $\theta<1$ and $\eps<\frac{1}{2^m}$. {The last inequality follows from Inequality~\eqref{equ:p_i upper bound}.} On the other hand,
\begin{align*}
    U_{j,r}(\theta)&\leq {1-p_r-p_j+\theta\cdot [(2^r-1)\eps^{2^m}-\eps^{2^j}((2^r-1)\eps^{2^m}-\eps^{2^m+1}+\eps^{2^r})]}\\
    &\leq 1-p_r-p_j-\theta\cdot \eps^{2^j+2^r}+2\cdot \eps^{2^m+1}\\
    &\leq 1-(2^j+2^r-2)\cdot \eps^{2^m}-\theta\cdot \eps^{2^j+2^r}+{4}\cdot \eps^{2^m+1}.
\end{align*}
{The first and last inequalities follow from Inequality~\eqref{equ:p_i upper bound} and~\eqref{equ:p_i lower bound}. The second inequality follows from $\theta\leq \theta_{m-1}<\eps^{2^{m-1}-1}$. There are two nice properties of the upper and lower bounds of $U_{j,r}(\theta)$: (i) they are within $O(\eps^{2^m+1})$ and (ii) they are both symmetric with respect to $j$ and $r$.}
\begin{enumerate}
    \item If $2^j+2^r>2^i$, then either $j$ or $r$ must be at least $i$, which implies that $2^j+2^r\geq 2^i+2$. 
By \Cref{equ:theta_m-1}, $\theta\leq \theta_{i+1}<\eps^{2^m-2^{i}-1}$. Thus $\theta\cdot \eps^{2^j+2^r}<\eps^{2^m+1}$. We notice that if either $j$ or $r$ is at least $i$, $(2^j+2^r-2)\cdot \eps^{2^m}$ is minimized when $(j,r)=(i,1)$ or $(1,i)$. Thus $U_{j,r}(\theta)\leq U_{i,1}(\theta)+7\cdot \eps^{2^m+1}$.
\item If $2^j+2^r<2^{i-1}$, then $i>1$ and thus $2^j+2^r\leq 2^{i-1}-2$ (as both are even numbers when $i>1$). Since $\theta\geq \theta_i>2^{i-1}\cdot \eps^{2^m-2^{i-1}}$, we have $\theta\cdot \eps^{2^j+2^r}\geq 2^{i-1}\cdot \eps^{2^m-2}$. As $\eps<\frac{1}{2^m}$, the term $\theta\cdot \eps^{2^j+2^r}$ dominates $(2^j+2^r-2)\cdot \eps^{2^m}$   {and $1-(2^j+2^r-2)\cdot \eps^{2^m}-\theta\cdot \eps^{2^j+2^r}$ is maximized when  
$2^j+2^r$ is maximized conditioned on $2^j+2^r<2^{i-1}$. In other words, it is maximized when $(j,r)=(i-2,i-3)$ or $(i-3,i-2)$.} Thus $U_{j,r}(\theta)\leq U_{i-2,i-3}(\theta)+6\cdot \eps^{2^m+1}$.
 
\item If $2^j+2^r\in [2^{i-1},2^i]$, then either $j$ or $r$ is in $\{i-2,i-1,i\}$. Since $|U_{j,r}(\theta)-U_{r,j}(\theta)|\leq 6\cdot\eps^{2^m+1}$, for all $j$ and $r$ in $\{1,\ldots, m-2\}$, we have $U_{j,r}(\theta)\leq \max_{j'\in \{i,i-1,i-2\},r'\in\{1,\ldots, m-2\}}U_{j',r'}(\theta)+6\cdot\eps^{2^m+1}$. 
\end{enumerate}
\end{enumerate}
\end{prevproof}

To complete the proof of \Cref{thm:lower_bound_full_info}, we need the following lemma adapted from~\cite{cai2020sell}.

\begin{lemma}[Adapted from Lemma 8 in \cite{cai2020sell}]\label{lem:eps-ic}
Suppose $\MM$ is a mechanism of option size $\ell$, where the IC and IR constraints are violated by at most $\delta$, for some $\delta>0$ and finite $\ell>0$. Then there exists an IC, IR mechanism $\MM'$ of option size $\ell$ such that $\rev(\MM')\geq (1-\eta)\cdot \rev(\MM)-\delta/\eta-\delta$, for any $\eta>0$. Moreover, if every experiment $E$ in $\MM$ is semi-informative, then $\MM'$ is a semi-informative menu. 
\end{lemma}

\begin{prevproof}{Lemma}{lem:eps-ic}
Let $\MM=(\{E(\theta),t(\theta)\})_{\theta\in [0,1]}$ and $\{(E_j,t_j)\}_{j\in [\ell]}$ be the $\ell$ options in $\MM$. For every $j\in [\ell]$, let $t_j' = (1-\eta) t_j - \delta$. Let $\MM'$ be the menu with options $\{(E_j,t_j')\}_{j\in [\ell]}$, i.e., the buyer with type $\theta$ purchase the experiment that maximizes $V_{\theta}(E_j)-t_j'$. Clearly $\MM'$ is IC and IR. 

For any type $\theta\in [0,1]$, let $j\in [m]$ be the unique number such that $E(\theta)=E_j$. Since the IC constraints in $\MM$ are violated by at most $\delta$, we know that
$$
    V_{\theta}(E_j) - t_j \geq  V_{\theta}(E_k) - t_k - \delta,  ~ \forall k\in [\ell]
$$
Now suppose that the buyer with type $\theta$ purchases $E_{k^*}$ in $\MM'$. Then
$$
    V_{\theta}(E_{k^*}) - (1-\eta) t_{k^*} \geq  V_{\theta}(E_j) - (1-\eta)t_j
$$
    
Choosing $k$ to be $k^*$ in the first inequality and  combining the two inequalities, we have that
$$
    V_{\theta}(E_{k^*}) - (1-\eta) t_{k^*} \geq V_{\theta}(E_{k^*}) - t_{k^*} - \varepsilon +\eta t_j \implies
        t_j - t_{k^*} \leq \frac{\delta}{\eta}
$$
    
    Hence, for the revenue we have
    $\rev(\MM') \geq (1-\eta)\cdot\rev(\MM) - \delta - \frac{\delta}{\eta}$. 
    The second part of the statement directly follows from the fact that $\MM'$ shares the same set of experiments with $\MM$.  
\end{prevproof}

\begin{prevproof}{Theorem}{thm:lower_bound_full_info}
We consider the construction present in \Cref{sec:lower_bound_full_info}. By \Cref{lem:ER-dist} and \Cref{lem:mech-m},
$$\rev(\MM)\geq \frac{1}{9}\cdot \log(2)\cdot (m-2)\cdot\optfi$$
We notice that $\optfi=\eps^{2^m}$.
According to the construction of the mechanism, $\MM$ is of option size at most $m-2$. Thus by applying \Cref{lem:eps-ic} with $\eta=\frac{1}{2}$ and $\delta=7\cdot \eps^{2^m+1}$, we have
$$\optsi\geq \frac{1}{2}\cdot\rev(\MM)-3\cdot 7\cdot \eps^{2^m+1}\geq \eps^{2^m}\cdot\left[\frac{\log(2)}{18}(m-2)-21\eps\right]=\Omega(m)\cdot \optfi$$
Since $\opt\geq \optsi$, we further have $\opt=\Omega(m)\cdot \optfi$.
\end{prevproof}

\begin{prevproof}{Lemma}{lem:1/k-approx for menus of size k}
Let $E_1, \ldots, E_{\ell}$ be the experiments contained in mechanism $\MM$, and $t(E_1), \ldots, t(E_{\ell})$ be the prices. Let $p_i$ be the revenue that is generated by experiment $E_i$, i.e., $p_i = \Pr_{\theta}[\MM(\theta)=E_i]\cdot t(E_i)$. Then, $\rev{(\MM)} = \sum_{i=1}^{\ell} p_i \implies p_j \geq \frac{\rev{(\MM})}{\ell}$, for some $j \in [\ell]$. Now, consider menu $\MM'$ that contains only
the fully informative experiment $E^*$ with price $t(E_j)$. We claim that for all types $\theta$ such that $\MM(\theta)=E_j$,
the agent with type $\theta$ is willing to purchase the experiment $E^*$ at price $t(E_j)$, as $V_\theta(E^*) \geq V_\theta(E_j), \forall \theta \in [0,1]$. Thus, for all $\theta$ such that $\MM(\theta)=E_j$, $V_{\theta}(E^*)-t(E^*)\geq V_{\theta}(E_j)-t(E_j)\geq u(\theta)$, where the second inequality follows from $\MM$ is IR. Hence, the revenue of $\MM'$ is at least $p_j\geq \frac{\rev(\MM)}{\ell}$.
\end{prevproof}

\begin{prevproof}{Theorem}{thm:lower_bound_menu_complexity}
By \Cref{thm:lower_bound_full_info}, there exists some environment such that $c:=\frac{\opt}{\optfi}=\Omega(m)$. Suppose there is an optimal menu $\MM^*$ that consists of $c'<c$ different experiments. Then by applying \Cref{lem:1/k-approx for menus of size k} to $\MM^*$, we have $\optfi\geq \frac{\opt}{c'}>\frac{\opt}{c}$. Contradiction. Thus any optimal menu consists of at least $c$ experiments.  
\end{prevproof}

%% file: proof_upper_bound.tex
\section{Missing Details from Section~\ref{sec:upper_bound_full_info}}\label{sec:proof_upper_bound}

\begin{observation}\label{obs:q to exp}\cite{BergemannBS2018design}
For any {\semiinformative} experiment $E$, $q(E)=\pi_{11}\cdot u_{11}-\pi_{2m}\cdot u_{2m}$ satisfies the following: When $q(E)\leq u_{11}-u_{2m}$, $E$ has Pattern 1 in \Cref{table:two-patterns},
where $\pi_{11}=(q(E)+u_{2m})/u_{11}$. When $q(E)>u_{11}-u_{2m}$, $E$ has Pattern 2 in \Cref{table:two-patterns}, 
where $\pi_{2m}=(u_{11}-q(E))/u_{2m}$.
\end{observation}

\begin{observation}\label{obs:value-function}\cite{BergemannBS2018design}
For any {\semiinformative} experiment $E$, and any $\theta\in [0,1]$, $V_{\theta}^*(E) = \theta q(E) + u_{2m} + \min\{ u_{11} - u_{2m} - q(E), 0\}$ (recall that $u_{11}\geq u_{2m}=1$).
\end{observation}

\subsection{Proof of \Cref{lem:describe the menu with q}}

\begin{lemma}\label{lem:opt without sigma-IC has mass on first two columns}
There exists a {\responsiveIC} and IR mechanism $\MM$ that satisfies both of the following properties:
\begin{enumerate}
    \item Every experiment of $\MM$ only recommends the fully informative actions. Formally, $\pi_{ij}(\MM(\theta))=0$, for all $\theta\in \Theta, i\in \{1,2\},2\leq j\leq m-1$. 
    \item $\rev(\MM)=\optstar$.
\end{enumerate}
\end{lemma}

\begin{proof}
Fix any {\responsiveIC} and IR mechanism $\MM$, any experiment $E$ that $\MM$ offers, any action $\ell$ such that $2\leq \ell\leq m-1$, and any state $i=1,2$. Suppose $\pi_{i\ell}(E)$ is not zero. We 
modify $E$ to get another experiment $E'$ by moving all the probability mass from $\pi_{i\ell}(E)$ to $\pi_{i1}(E)$ and $\pi_{im}(E)$ in a proper way, such that the modified mechanism generates the same revenue. 
Let $\MM^*$ be an optimal {\responsiveIC} and IR mechanism. We show how to modify its experiments $E$ to obtain a {\responsiveIC} and IR menu that satisfies both properties in the statement.  

For any $\theta\in [0,1]$, recall that $V_{\theta}^*(E)$ is the agent's value for experiment $E$ if she follows the recommendation:

$$V_\theta^*(E) = \theta \sum_{j\in [m]} \pi_{1j}(E) u_{1j} + (1-\theta)\sum_{j\in [m]} \pi_{2j}(E) u_{2j}$$

Let $U_{\theta}(E)=V_\theta^*(E)-t(E)$ be her utility for $E$ when following the recommendation. We only consider moving the probability mass from $\pi_{2\ell}(E)$ to $\pi_{21}(E)$ and $\pi_{2m}(E)$. The case where we move the probability mass of $\pi_{1\ell}(E)$ follows from a similar argument. 

Without loss of generality, assume $\eps=\pi_{2\ell}(E)>0$. We move $\eps_1$ mass from $\pi_{2\ell}(E)$ to $\pi_{21}(E)$ and $\eps_2$ mass from $\pi_{2\ell}(E)$ to $\pi_{2m}(E)$. Both $\eps_1$ and $\eps_2$, which satisfy $\eps_1+\eps_2=\eps$, will be determined later. Let $E'$ be the modified experiment, and we keep its price as $t(E)$. 
For every $\theta$, we show how to choose $\eps_1$ and $\eps_2$ appropriately so that $U_{\theta}(E')=U_{\theta}(E)$. 
\begin{align*}
    U_\theta(E')-U_{\theta}(E) &= (1-\theta) \sum_{j\in [m]} (\pi_{2j}(E')-\pi_{2j}(E)) u_{2j}\\
    &{=(1-\theta) \cdot \left( \eps_1(u_{21}-u_{2\ell} ) + \eps_2(u_{2m}-u_{2\ell})\right)}
\end{align*}

Here the first equality follows from $\pi_{1j}(E')=\pi_{1j}(E), \forall j\in [m]$. 
Now choose $\eps_1, \eps_2$ such that $\eps_1 + \eps_2 = \eps$ and $\eps_1(u_{2\ell} - u_{21}) =\eps_2(u_{2m} - u_{2\ell})$. In other words, $$\eps_1 = \frac{(u_{2m} - u_{2\ell})\eps}{(u_{2\ell} - u_{21})+(u_{2m} - u_{2\ell})},\quad\eps_2 = \frac{(u_{2\ell} - u_{21})\eps}{(u_{2\ell} - u_{21})+(u_{2m} - u_{2\ell})}$$
Notice that since $u_{21}<u_{2\ell}<u_{2m}$ we have that $\eps_1,\eps_2>0$. With the choices of $\eps_1,\eps_2$, we have $U_\theta(E) = U_\theta(E'), \forall \theta \in [0,1]$. Thus, the modified mechanism is still {\responsiveIC} and IR and has the same revenue. 
We modify  $\MM^*$ by applying the above procedure to every experiment $E$ offered by $\MM^*$, every $\ell\in \{2,...,m-1\}$ , and every state $i=1,2$. Let $\MM'$ be the mechanism after the modification. Then $\MM'$ is {\responsiveIC} and IR and satisfies both properties in the statement. 
\end{proof}

\notshow{
\begin{lemma}\label{lem:opt without sigma-IC has mass on first two columns}
We consider an environment where $\Omega = \{ \omega_1, \omega_2\}$ and $A = \{a_1, \ldots, a_k\}$. We consider the class of menus $\mathcal{C}_\sigma$ that satisfy the feasibility constraints, the  IC constraints, and the IR constraints but \underline{not} the $\sigma$-IC constraints. Let $\MM_\sigma \in \mathcal{C}_\sigma$ be a menu in this class. Let $a_1, a_2$ be the actions that generate the most payoff under states $\omega_1, \omega_2$, respectively. Then, there is a menu $\MM'_\sigma \in \mathcal{C}_\sigma$ in which every experiment $E \in \MM'_\sigma$ puts probability mass only in the columns that correspond to actions $a_1, a_2$ and has $\rev{(\MM'_\sigma)} = \rev{(\MM_\sigma)}$.
\end{lemma}

\begin{proof}
Consider an experiment $E \in \MM^*_\sigma$, an action $a_\ell, \ell > 2$, and let $\pi_{\omega_1,\ell}(E), \pi_{\omega_2, \ell}(E)$ be the probability that $E$ recommends action $\ell$ when the state of the world is $\omega_1, \omega_2$ respectively. Assume that $\pi_{\omega_1,\ell}(E) > \epsilon > 0$ or $\pi_{\omega_2, \ell}(E) > \epsilon > 0$. We will show how to modify this experiment $E$ to generate the same amount of revenue and having $\pi_{\omega_1,\ell}(E) = \pi_{\omega_2,\ell}(E) = 0$. In this class of menus, since we ignore the $\sigma$-IC constraints we consider a modified version of the value function 
\begin{align*}
 V_\theta^\sigma(E) = \theta \sum_{j\in [k]} \pi_{\omega_1, j}(E) u_{\omega_1, a_j} + (1-\theta)\sum_{j\in [k]} \pi_{\omega_2, j}(E) u_{\omega_2, a_j} \implies \\
 V_\theta^\sigma(E) = \theta \left(\sum_{j\in [k]} \pi_{\omega_1, j}(E) u_{\omega_1, a_j} -  \sum_{j\in [k]} \pi_{\omega_2, j}(E) u_{\omega_2, a_j}\right) + \sum_{j\in [k]} \pi_{\omega_2, j}(E) u_{\omega_2, a_j}
\end{align*}
Let $U_{\theta}(E)$ be the utility of $\theta$ if she picks experiment $E$. We have that
\begin{align*}
    U_{\theta}(E) = \theta \left(\sum_{j\in [k]} \pi_{\omega_1, j}(E) u_{\omega_1, a_j} -  \sum_{j\in [k]} \pi_{\omega_2, j}(E) u_{\omega_2, a_j}\right) + \sum_{j\in [k]} \pi_{\omega_2, j}(E) u_{\omega_2, a_j} - t(E)
\end{align*}
We first consider the case where $\pi_{\omega_2, \ell}(E) = \epsilon > 0$. The other case follows similarly. We move $\epsilon_1$ mass from $\pi_{\omega_2, \ell}(E)$ to $\pi_{\omega_2, 1}(E)$ and $\epsilon_2$ mass from $\pi_{\omega_2, \ell}(E)$ to $\pi_{\omega_2, 2}(E)$. These quantities will be determined later. We call this modified experiment $E'$ and the modified menu $\MM'_\sigma$. Now consider the utility of the type $\theta$ if she picks the experiment $E'$. We have

\begin{align*}
    U_\theta(E') = \theta \left(\sum_{j\in [k]} \pi_{\omega_1, j}(E') u_{\omega_1, a_j} -  \sum_{j\in [k]} \pi_{\omega_2, j}(E') u_{\omega_2, a_j}\right) + \sum_{j\in [k]} \pi_{\omega_2, j}(E') u_{\omega_2, a_j} - t(E') =\\
    \theta \left(\sum_{j\in [k]} \pi_{\omega_1, j}(E) u_{\omega_1, a_j}  -  \sum_{j\in [k]} \pi_{\omega_2, j}(E) u_{\omega_2, a_j} - \epsilon_1(u_{\omega_2, a_1} - u_{\omega_2, a_\ell}) - \epsilon_2(u_{\omega_2, a_2} - u_{\omega_2, a_\ell})\right) + \\
    \sum_{j\in [k]} \pi_{\omega_2, j}(E) u_{\omega_2, a_j} + \epsilon_1(u_{\omega_2, a_1} - u_{\omega_2, a_\ell}) + \epsilon_2(u_{\omega_2, a_2} - u_{\omega_2, a_\ell}) - t(E) 
\end{align*}

Now we pick $\epsilon_1, \epsilon_2$ in a way such that $\epsilon_1 + \epsilon_2 = \epsilon$ and $\epsilon_1(u_{\omega_2, a_1} - u_{\omega_2, a_\ell}) = -\epsilon_2(u_{\omega_2, a_2} - u_{\omega_2, a_\ell})$. Hence, we get that $$\epsilon_1 = -\frac{u_{\omega_2, a_2} - u_{\omega_2, a_\ell}}{(u_{\omega_2, a_1} - u_{\omega_2, a_\ell})-(u_{\omega_2, a_2} - u_{\omega_2, a_\ell})}\epsilon, \epsilon_2 = \frac{u_{\omega_2, a_1} - u_{\omega_2, a_\ell}}{(u_{\omega_2, a_1} - u_{\omega_2, a_\ell})-(u_{\omega_2, a_2} - u_{\omega_2, a_\ell})}\epsilon$$
Notice that since $u_{\omega_2, a_2} - u_{\omega_2, a_\ell} > 0, u_{\omega_2, a_1} - u_{\omega_2, a_\ell} < 0$ we have that $\epsilon_1 > 0, \epsilon_2 > 0$. With these choices, we have that $U_\theta(E) = U_\theta(E'), \forall \theta \in \Theta$. Hence, all the types who were buying experiment $E$ will now purchase experiment $E'$ and no type who was not picking this experiment will now deviate. Moreover, this menu will generate the same revenue as the original one since we have not modified the prices. The case where we have $\pi_{\omega_1, \ell}(E) > 0$ follows by a symmetric argument.
\end{proof}

}



Due to \Cref{lem:opt without sigma-IC has mass on first two columns}, we can focus on {\responsiveIC} and IR mechanisms that only recommend fully informative actions. Let $\MenuTwoAction$ be the set of all such mechanisms. For ease of notation, we denote an experiment offered by the mechanism $E=\begin{pmatrix}
\pi_{11} & 1-\pi_{11}\\
1-\pi_{2m} & \pi_{2m}
\end{pmatrix}$. Note that 
$V_\theta^*(E) = \theta u_{11}\cdot \pi_{11}(E) + (1-\theta)u_{2m}\cdot \pi_{2m}(E)$. 


\vspace{.1in}

\begin{prevproof}{Lemma}{lem:describe the menu with q}
Let $\MM$ be an arbitrary mechanism in $\MenuTwoAction$. Let $E$ be any experiment offered by $\MM$, such that $\pi_{11}<1$ and $\pi_{2m}<1$. Without loss of generality, assume that 
$(1-\pi_{11})u_{11}\leq (1-\pi_{2m})u_{2m}$. The other case is similar. Let $\eps_1 = 1-\pi_{11}$ and $\eps_2 = \frac{u_{11}}{u_{2m}}\cdot \eps_1$. By our assumption, $\eps_2\leq 1-\pi_{2m}$. 
We modify experiment $E$ and its price $t(E)$ in the following way: Move probability mass $\eps_1$ from $\pi_{1m}$ to $\pi_{11}$, and move probability mass $\eps_2$ from $\pi_{21}$ to $\pi_{2m}$. 
Denote $E'$ the modified experiment. Since $\pi_{11}+\eps_1=1$, we have 
$$E'=\begin{pmatrix}
1 & 0 \\
1-\pi_{2m}-\eps_2 & \pi_{2m}+\eps_2
\end{pmatrix}.$$ 
Consider the mechanism $\widehat{\MM}$ where we replace $E$ by $E'$ and let the price of experiment $E'$ be $t(E'):=t(E)+\eps_2\cdot u_{2m}$. Formally, for every $\theta$ such that $\MM(\theta)=E$, define $\widehat{\MM}(\theta)=E'$ and the payment $t'(\theta)=t(\theta)+\eps_2\cdot u_{2m}$. For any other types, $\widehat{\MM}$ offers the same experiment with the same price as offered by $\MM$.   

For any $\theta$, we have that 
\begin{align*}
    V^*_\theta(E') - V^*_\theta(E) &= \theta u_{11}\cdot (\pi_{11}(E')-\pi_{11}(E)) + (1-\theta) u_{2m}\cdot (\pi_{2m}(E')-\pi_{2m}(E))\\
    &= \theta \eps_1 \cdot u_{11} + (1-\theta) \eps_2\cdot u_{2m} = \eps_2\cdot u_{2m}.
\end{align*} Hence, we can immediately see that for every type $\theta$, the agent's utility for experiment $E'$ at price $t(E'))$ is the same as the agent's utility for experiment $E$ at price $t(E)$. Since $\MM$ is {\responsiveIC} and IR, $\widehat{\MM}$ is also {\responsiveIC} and IR. Moreover, we have that $\rev{(\widehat{\MM})} \geq \rev{(\MM)}$. By repeating this procedure 
for all the experiments in the optimal mechanism $\MM'\in \MenuTwoAction$, we can construct a mechanism $\MM^*\in \MenuWithQ$, such that $\rev(\MM^*)\geq \rev(\MM')$. By \Cref{lem:opt without sigma-IC has mass on first two columns}, $\rev(\MM')=\optstar$. Thus $\rev(\MM^*)=\optstar$. 
\end{prevproof}


\subsection{Proof of Lemma~\ref{lem:ic in q space}}

We begin with necessity. For every $\theta$, let $E(\theta)$ be the experiment that corresponds to $q(\theta)$ (\Cref{obs:q to exp}). For any two types $\theta_1,\theta_2\in [0,1]$, where $\theta_2>\theta_1$, 
we have
\begin{align*}
    V^*_{\theta_1}(E(\theta_1))-t(\theta_1)&\geq     V^*_{\theta_1}(E(\theta_2))-t(\theta_2)\\
    V^*_{\theta_2}(E(\theta_2))-t(\theta_2)&\geq     V^*_{\theta_2}(E(\theta_1))-t(\theta_1)    
\end{align*}

Adding up both inequalities obtains
\begin{equation}\label{equ1:lem:ic in q space}
V^*_{\theta_1}(E(\theta_1))+V^*_{\theta_2}(E(\theta_2))\geq V^*_{\theta_1}(E(\theta_2))+V^*_{\theta_2}(E(\theta_1))    
\end{equation}

By \Cref{obs:value-function}, for every $\theta,\theta'\in [0,1]$,
$$
V^*_{\theta}(E(\theta'))=
\begin{cases}
\theta\cdot q(\theta')+u_{2m},& \quad\text{if }q(\theta')\leq u_{11}-u_{2m}\\
(\theta-1)\cdot q(\theta')+u_{11},&\quad \text{if }q(\theta')>u_{11}-u_{2m}
\end{cases}
$$

If both $q(\theta_1)$ and $q(\theta_2)$ are at most $u_{11}-u_{2m}$, Inequality~\eqref{equ1:lem:ic in q space} $\iff (\theta_1-\theta_2)(q(\theta_1)-q(\theta_2))\geq 0\iff q(\theta_1)\leq q(\theta_2)$. Similarly, when both $q(\theta_1)$ and $q(\theta_2)$ are greater than $u_{11}-u_{2m}$, we also have $q(\theta_1)\leq q(\theta_2)$.  

Without loss of generality, the only remaining case is when 
$q(\theta_2)\leq u_{11}-u_{2m}<q(\theta_1)$. Inequality~\eqref{equ1:lem:ic in q space} is equivalent to
\begin{align*}
&(\theta_1-1)q(\theta_1)+u_{11}+\theta_2\cdot q(\theta_2)+u_{2m}\geq \theta_1\cdot q(\theta_2)+u_{2m}+(\theta_2-1)\cdot q(\theta_1)+u_{11}\\
\iff &(\theta_1-\theta_2)(q(\theta_1)-q(\theta_2))\geq 0,
\end{align*}
which contradicts with $\theta_1<\theta_2$. Hence, $q(\theta_1)\geq q(\theta_2)$ and $q(\theta)$ is non-decreasing in $\theta$. 

For every $\theta\in [0,1]$, let $U^*(\theta)=V^*_{\theta}(E(\theta))-t(\theta)$ be the agent's utility by reporting truthfully and following the recommendation of the experiment. Let $\theta^*=\sup\{\theta:q(\theta)\leq u_{11}-u_{2m}\}$. Assume $q(\theta^*)\leq u_{11}-u_{2m}$.\footnote{The other case follows from a similar argument, where we replace both intervals $[0,\theta^*],(\theta^*,1]$ by $[0,\theta^*),[\theta^*,1]$.}
When $\theta\in [0,\theta^*]$, $U^*(\theta)=\theta\cdot q(\theta)+u_{2m}-t(\theta)$ is a quasilinear function, and by Myerson's theory~\cite{Myerson81}, $\MM$ is {\responsiveIC} implies that $U^*(\theta)=U^*(0)+\int_{0}^\theta q(x)dx$. Thus
\begin{equation}\label{equ:payment-low}
t(\theta)=V^*_\theta(E(\theta))-U^*(\theta)=\theta\cdot q(\theta)+u_{2m}-U^*(0)-\int_{0}^\theta q(x)dx
\end{equation}
Similarly, for all $\theta\in (\theta^*,1]$, we have $U^*(\theta)=U^*(1)-\int_{\theta}^{1} q(x)dx$. Hence
\begin{equation}\label{equ:payment-high}
t(\theta)=(\theta-1)\cdot q(\theta)+u_{11}-U^*(1)+\int_{\theta}^{1} q(x)dx
\end{equation}

According to \Cref{equ:payment-low}, when $\theta=0$, $t(0)=u_{2m}-U^*(0)\geq 0$. Since $\MM$ is IR, $U^*(0)\geq u_{2m}$. Thus $U^*(0)=u_{2m}$. Similarly, by \Cref{equ:payment-high} we have $U^*(1)=u_{11}$.

When the agent has type $\theta^*$ and misreports to $\theta'>\theta^*$, {\responsiveIC} implies that 
\begin{align*}
V^*_{\theta^*}(E(\theta'))-t(\theta')=&(\theta^*-\theta')\cdot q(\theta')+u_{11}-\int_{\theta'}^1q(x)dx\\
\leq & U^*(\theta^*)=u_{2m}+\int_{0}^{\theta^*}q(x)dx
\end{align*} 

Let $\theta'\rightarrow \theta^{*^+}$,~\footnote{$a\rightarrow b^+$ means that $a$ approaches $b$ from the right.} we have $\int_{0}^1q(x)dx\geq u_{11}-u_{2m}$. Similarly, when the agent has type $\theta'>\theta^*$ and misreports to $\theta^*$, {\responsiveIC} implies that
\begin{align*}
V^*_{\theta'}(E(\theta^*))-t(\theta^*)=&(\theta'-\theta^*)\cdot q(\theta^*)+u_{2m}+\int_0^{\theta^*}q(x)dx\\
\leq & U^*(\theta')=u_{11}-\int_{\theta'}^{1}q(x)dx
\end{align*} 

Let $\theta'\rightarrow \theta^{*^+}$, we have $\int_{0}^1q(x)dx\leq u_{11}-u_{2m}$. Thus $\int_{0}^1q(x)dx=u_{11}-u_{2m}$. Now it's easy to verify that $t(\theta)$ satisfies \Cref{equ:payment}. Thus for every $\theta$, $U^*(\theta)=V_\theta^*(E(\theta))-t(\theta)=u_{2m}+\int_{0}^\theta q(x)dx$. The third property directly follows from the fact that $\MM$ is IR.

For sufficiency, suppose $\Bq$ satisfies all of the properties in the statement, construct the payment $\Bt$ using \Cref{equ:payment}. Then by \Cref{obs:value-function}, for every true type $\theta$, misreporting to type $\theta'$ induces utility
$$V^*_{\theta}(E(\theta'))-t(\theta')=(\theta-\theta')\cdot q(\theta')+u_{2m}+\int_{0}^{\theta'}q(x)dx$$

Since $q(\theta')$ is non-decreasing, it is not hard to see that the utility is maximized when $\theta'=\theta$. Thus the mechanism is {\responsiveIC}. Moreover, the utility for reporting truthfully her type $\theta$ is $u_{2m}+\int_{0}^\theta q(x)dx$. Thus by property 3, the mechanism is IR. 


\subsection{Proof of \Cref{lem:modified-program-equivalent}}

By \Cref{lem:ic in q space}, for any $\Bq$ that can be implemented with some {\responsiveIC} and IR mechanism $\MM$, the revenue of $\MM$ can be written as follows:
\begin{align*}
   \rev(\MM)&=\int_0^1 t(\theta){f(\theta)} d\theta=\int_0^1\left(\theta\cdot q(\theta) + \min\{u_{11} - u_{2m} - q(\theta), 0\} - \int_0^\theta q(x) dx\right){f(\theta)}d\theta\\
   &=\int_0^1 [\theta f(\theta) q(\theta) + \min\{(u_{11}-u_{2m} - q(\theta))f(\theta), 0\}]d\theta-\int_0^1\int_0^\theta q(x) d x d F(\theta)\\
    &=\int_0^1 [\theta f(\theta) q(\theta) + \min\{(u_{11}-u_{2m} - q(\theta))f(\theta), 0\}]d\theta-F(\theta)\int_0^\theta q(x) d x~\Big{|}_0^1 +\int_0^1 q(\theta) F(\theta)d \theta\\
   &=\int_0^1 [(\theta f(\theta) + F(\theta))q(\theta) + \min\{(u_{11}-u_{2m} - q(\theta))f(\theta), 0\}]d\theta{-(u_{11}-u_{2m})}
\end{align*}

\begin{figure}[ht!]
\colorbox{MyGray}{
\begin{minipage}{.98\textwidth}
$$\quad\textbf{{sup}  } \int_0^1 [(\theta f(\theta) + F(\theta))q(\theta) + \min\{(u_{11}-u_{2m} - q(\theta))f(\theta), 0\}]d\theta$$
  \begin{align*}
\text{s.t.} \qquad 
 &q(\theta) \text{ is non-decreasing in }\theta\in [0,1]\\
 & q(0)\geq -u_{2m},\quad q(1)\leq u_{11}\\ 
    &\int_{0}^1 q(\theta) d\theta = u_{11}-u_{2m} \\
    & u_{2m}+\int_0^{\theta}q(x)dx\geq u(\theta), &\forall \theta\in [0,1]
\end{align*}
\end{minipage}}
\caption{The Optimization Problem with Explicit IR constraints}\label{fig:continous-program}
\end{figure}

Thus the optimal mechanism $\MM^*\in \MenuWithQ$ is captured by the optimization problem in \Cref{fig:continous-program}. Notice that there is an IR constraint for every type $\theta$ in \Cref{fig:continous-program}. To bound the size of the optimal {\responsiveIC} and IR mechanism, we propose an equivalent optimization problem that only contains $O(m)$ constraints. 
Recall that for every $\theta\in [0,1]$, $u(\theta)=\max_{k\in [m]}\{\theta\cdot u_{1k}+(1-\theta)\cdot u_{2k}\}$. For every $k\in [m]$, let $h_k(\theta)=\theta\cdot u_{1,m+1-k}+(1-\theta)\cdot u_{2,m+1-k}=\theta\cdot \ell_k+u_{2,m+1-k},\forall \theta\in [0,1]$ be the agent's value function when she has no additional information, and follows action $m+1-k$. Recall that $\ell_k=u_{1,m+1-k}-u_{2,m+1-k}$ is also the slope of the $k$-th piece (from the left) of the IR curve. 

Given any feasible solution $\{q(\theta)\}_{\theta\in [0,1]}$ of the program in \Cref{fig:continous-program}, the IR constraint at type $\theta$ is equivalent to: $u_{2m}+\int_{0}^{\theta}q(x)dx\geq h_k(\theta),\forall k\in [m]$. For every $k\in [m]$, let $\theta_k :=\sup\{\theta\in [0,1]:q(\theta)\leq \ell_k\}$. Consider the following function $$u_{2m}+\int_{0}^{\theta}q(x)dx- h_k(\theta)=u_{2m}+\int_{0}^{\theta}q(x)dx-\theta\cdot \ell_k-u_{2,m+1-k}$$

By taking the derivative on $\theta$, we can see that the above function is minimized at $\theta_k$, {as $q(\theta)\leq\ell_k$ for all $\theta<\theta_k$ and $q(\theta)> \ell_k$ for all $\theta>\theta_k$.} Thus the set of constraints $\left\{u_{2m}+\int_{0}^{\theta}q(x)dx\geq h_k(\theta),\forall \theta\in [0,1]\right\}$ is captured by a single constraint $u_{2m}+\int_{0}^{\theta_k}q(x)dx\geq h_k(\theta_k)$. We notice that $\theta_k=\int_{0}^1\ind[q(x)\leq \ell_k]dx$. By the definition of $h_k$,
the constraint can be rewritten as

\begin{equation}\label{equ:IR_consraint}
\int_0^{1}(q(x)-\ell_k)\cdot \ind[q(x)\leq \ell_k]dx\geq u_{2,m+1-k}-u_{2m}    
\end{equation}

A feasible $\Bq$ must satisfy that $q(\theta)\in [-u_{2m},u_{11}]$ and $\int_{0}^1q(x)dx=u_{11}-u_{2m}$. Also $\ell_1=-u_{2m},\ell_m=u_{11}$. Thus, Inequality~\eqref{equ:IR_consraint} is trivial when $k=1$ or $m$. In our modified optimization problem, we only include the constraints for $2\leq k\leq m-1$.

\vspace{.2in}

\begin{prevproof}{Lemma}{lem:modified-program-equivalent}

It suffices to show that: A solution $\Bq=\{q(\theta)\}_{\theta\in [0,1]}$ is feasible in the optimization problem in \Cref{fig:continous-program} if and only if it is feasible in the optimization problem in \Cref{fig:continous-program-modified}.

We denote $Q$ the optimization problem in \Cref{fig:continous-program} and $Q'$ the optimization problem in \Cref{fig:continous-program-modified}. Suppose $\Bq$ is feasible in $Q$. For every $k\in [m]$, let $\theta_k=\sup\{\theta\in [0,1]:q(\theta)\leq \ell_k\}$. Then since $\theta_k=\int_{0}^1\ind[q(\theta)\leq \ell_k]dx$, we have
\begin{align*}
\int_0^1(q(x)-\ell_k)\cdot \ind[q(x)\leq \ell_k]dx&=\int_0^{\theta_k}q(x)dx-\theta_k\ell_k\\
&\geq u(\theta_k)-u_{2m}-\theta_k\ell_k \geq u_{2,m+1-k}-u_{2m}
\end{align*}

Here the first inequality follows from the IR constraint of program $Q$ at type $\theta_k$, and the second inequality follows from the definition of the IR curve $u(\cdot)$. Thus $\Bq$ is feasible in $Q'$.

For the other direction, suppose $\Bq$ is feasible for $Q'$. For every $k\in [m]$, consider the following function $H_k:[0,1]\to \mathbb{R}$, where $H(\theta):=\int_0^\theta q(x)dx-\theta\cdot \ell_k$. Since $q(\theta)$ is non-decreasing in $\theta$, $H_k(\theta)$ is minimized at $\theta=\theta_k=\sup\{\theta\in [0,1]:q(\theta)\leq \ell_k\}$. We notice that the last constraint of $Q'$ implies that $H_k(\theta_k)\geq u_{2,m+1-k}-u_{2m}$ for all  $k\in \{2,3,...,m-1\}$. It is not hard to verify that the inequality also holds for $k=1$ and $k=m$. When $k=1$, the inequality is trivial when $q(0)>-u_{2m}$. If $q(0)=-u_{2m}$, since $\ell_1=-u_{2m}$ and $q(\theta)\in [-u_{2m},u_{11}],\forall \theta\in [0,1]$, thus $q(\theta)=-u_{2m},\forall \theta\in [0,\theta_1]$. Thus $H_1(\theta_1)=0=u_{2m}-u_{2m}$. When $k=m$, since $\theta_m=1$ and $\ell_m=u_{11}$, $H_m(\theta_m)=u_{11}-u_{2m}-u_{11}=u_{21}-u_{2m}$.

Thus for every $k\in [m],\theta\in [0,1]$, we have
$$\int_0^\theta q(x)dx-\theta\cdot \ell_k\geq u_{2,m+1-k}-u_{2m}$$

Since $u(\theta)=\max_{k\in [m]}\{\theta\cdot \ell_k+u_{2,m+1-k}\}$, we have
$\int_0^\theta q(x)dx\geq u(\theta)-u_{2m}$. Thus $\Bq$ is feasible for $Q$. 
\end{prevproof}

\begin{prevproof}{Theorem}{thm:1/k approximation of full information for $k$ actions }
By Lemma~\ref{lem:opt without sigma-IC has k pieces}, let 
$\MM^*$ be the optimal semi-informative, {\responsiveIC} and IR mechanism with option size at most $\optionsize$.
By \Cref{lem:describe the menu with q}, $\rev(\MM^*)=\optstar\geq \opt$. By Lemma~\ref{lem:1/k-approx for menus of size k}, there exists a menu $\MM$ that contains only the fully informative experiment, such that $\rev(\MM)\geq \frac{\rev(\MM^*)}{\optionsize}$. Thus, $\optfi\geq \frac{\opt}{\optionsize}$. 
\end{prevproof}

%% file: proof_option_size.tex
\subsection{Proof of \Cref{lem:opt without sigma-IC has k pieces}}\label{sec:proof of m upper bound}


We first consider the case where $D$ is a discrete distribution. For simplicity we assume that $0\in \supp(D)$. Let $\supp(D)=\{\theta_1,\ldots,\theta_N\}$, where $0=\theta_1<\theta_2<\ldots <\theta_N\leq 1$ and $N$ is the size of the support. For every $i\in N$, denote $f_i$ the density of type $\theta_i$. For ease of notations, we denote $\theta_{N+1}=1$. The optimization problem in \Cref{fig:continous-program-modified} w.r.t. the discrete distribution $D$ is shown in \Cref{fig:discrete-program}.~\footnote{For any solution $\{q_i\}_{i\in [N]}$, the problem in \Cref{fig:discrete-program} is clearly equivalent to the problem in \Cref{fig:continous-program-modified} where $q(\cdot):[0,1]\to [-u_{2m},u_{11}]$ is defined as $q(\theta)=q_i, \forall i\in [N], \theta\in [\theta_i,\theta_{i+1})$.} The set of variables in the program is $\{q_i\}_{i\in [N]}$. Denote $\cP(D)$ (or $\cP$ when $D$ is clear from context) the optimization problem in \Cref{fig:discrete-program}.

\begin{figure}[ht!]
\colorbox{MyGray}{
\begin{minipage}{.98\textwidth}
$$\quad\textbf{max }\sum_{i=1}^{N}\left[q_i\cdot\left(\theta_i\cdot f_i+(\theta_{i+1}-\theta_i)\cdot\sum_{j=1}^i f_j\right)+\min\left\{(u_{11}-u_{2m}-q_i)\cdot f_i,0\right\}\right]$$
\begin{align*}
\text{s.t.} \qquad 
 &{\MonoConstr}\quad q_i\leq q_{i+1}, & \forall i\in [N-1]\\
 & {\MarginConstr}\quad q_1\geq-u_{2m},\quad q_N\leq u_{11}\\
 &{\IntegralConstr}\quad \sum_{i=1}^N q_i\cdot (\theta_{i+1}-\theta_i)= u_{11}-u_{2m} \\
 & {\IRConstr}\quad \sum_{i=1}^{N}\min\{q_i-\ell_k,0\}\cdot (\theta_{j+1}-\theta_j)\geq u_{2,m+1-k}-u_{2m}, &\forall k\in \{2,3,...,m-1\}
\end{align*}
\end{minipage}}
\caption{The Optimal {\ResponsiveIC} and IR Mechanism for Discrete Distribution $D$}~\label{fig:discrete-program}
\end{figure}

{
Our goal is to turn the program $\cP(D)$ into a collection of LPs, so that the highest optimum among the collection of LPs correspond to the optimum of $\cP(D)$. Each LP is parametrized by a vector $(i^*,\Bi=\{i_2,\ldots, i_{m-1}\})$, and finds the optimal solution among all $(q_1,\ldots, q_N)$'s that satisfies the following conditions:
\begin{enumerate}
    \item $q_{i}\leq u_{11}-u_{2m}$ for all $i< i^*$ and $q_{i}\geq u_{11}-u_{2m}$ for all $i\geq i^*$.
    \item $q_{i}\leq \ell_k$ for all $i< i_k$ and $q_{i}\geq \ell_k$ for all $i\geq i_k$ for all $k\in\{2,\ldots,m-1\}$. Recall that $\ell_k=u_{1,m+1-k}-u_{2,m+1-k}$.
\end{enumerate}



The reason that we consider solutions that satisfy these conditions is that now the objective and constraint {\IRConstr} are linear. In particular, the objective becomes 
$$\sum_{i=1}^{N}q_i\cdot\left(\theta_i\cdot f_i+(\theta_{i+1}-\theta_i)\cdot\sum_{j=1}^if_i\right)+\sum_{i=i^*+1}^N(u_{11}-u_{2m}-q_i)\cdot f_i.$$ 

Constraint {\IRConstr} becomes
$$\sum_{j=1}^{i_k}(q_j-\ell_k)\cdot (\theta_{j+1}-\theta_j)\geq u_{2,m+1-k}-u_{2m}.$$

To further simplify the program, we introduce variables $\hat{q}_{i}=q_{i+1}-q_i,\forall i\in [N-1]$, and replace $q_i$ with $q_1+\sum_{j=1}^{i-1}\hat{q}_j,\forall i\in [N]$. The monotonicity constraint of $\cP$ is thus captured by all $\hat{q}_i$'s being non-negative. See \Cref{fig:discrete-program2} for the LP we construct. In particular, constraint {\MarginConstrDis} corresponds to constraint {\MarginConstr} of $\cP$; Constraint {\IntegralConstrDis} corresponds to constraint {\IntegralConstr} of $\cP$. Constraints {\istarConstrLeft} and {\istarConstrRight} follow from the definition of $i^*$. Constraints {\ikConstrLeft} and {\ikConstrRight} follow from the definition of $i_k$, and {\IRConstrDis} corresponds to constraint {\IRConstr} of $\cP$. In the LP in \Cref{fig:discrete-program2}, both $i^*$ and $\mathbf{i}=\{i_k\}_{2\leq k\leq m-1}$ are fixed parameters. We denote $\cP'(D, i^*,\mathbf{i})$ the LP with those parameters.}

\begin{figure}[ht!]
\colorbox{MyGray}{
\begin{minipage}{.98\textwidth}
$$\quad\textbf{max  } \sum_{i=1}^{N}(q_1+\sum_{j=1}^{i-1}\hat{q}_j)\cdot(\theta_i f_i+(\theta_{i+1}-\theta_i)\cdot \sum_{r=1}^if_r)+\sum_{i=i^*+1}^N(u_{11}-u_{2m}-q_1-\sum_{j=1}^{i-1}\hat{q}_j)\cdot f_i$$
\vspace{-.3in}
  \begin{align*}
\text{s.t.} \qquad
   & {\MarginConstrDis}\quad q_1\geq -u_{2m},
   \quad q_1+\sum_{j=1}^{N-1}\hat{q}_j\leq u_{11}\\
    &{\IntegralConstrDis}\quad \sum_{i=1}^N(q_1+\sum_{j=1}^{i-1}\hat{q}_j)\cdot (\theta_{i+1}-\theta_i) = u_{11} - u_{2m}\\
    &{\istarConstrLeft}\quad q_1+\sum_{j=1}^{i^*-1}\hat{q}_j\leq u_{11}-u_{2m}\\
    &{\istarConstrRight}\quad q_1+\sum_{j=1}^{i^*}\hat{q}_j\geq u_{11}-u_{2m}\\ 
    & {\ikConstrLeft}\quad q_1+\sum_{j=1}^{i_k-1}\hat{q}_j\leq u_{1,m+1-k}-u_{2,m+1-k}, &\forall k\in \{2,3,\ldots,m-1\}\\
    & {\ikConstrRight}\quad q_1+\sum_{j=1}^{i_k}\hat{q}_j\geq u_{1,m+1-k}-u_{2,m+1-k}, &\forall k\in \{2,3,\ldots,m-1\}\\  
    &{\IRConstrDis}\quad \sum_{j=1}^{i_k}(q_1+\sum_{r=1}^{j-1}\hat{q}_r-\ell_k)\cdot (\theta_{j+1}-\theta_j)\geq u_{2,m+1-k}-u_{2m}, &\forall k\in \{2,3,\ldots,m-1\}\\    
    &\hat{q}_i\geq 0,&\forall i\in [N-1]
\end{align*}
\end{minipage}}
\caption{The Linear Program with Parameters $(i^*,\mathbf{i})$}~\label{fig:discrete-program2}
\end{figure}



We first show in \Cref{lem:discrete-program-most-zeros} that the LP has an optimal solution $\hat{\Bq}=(q_1,\{\hat{q}_i\}_{i\in [N-1]})$ such that all but ($\optionsize$) $\hat{q}_i$'s are zero, 
for any discrete distribution $D$ and any set of parameters $(i^*,\mathbf{i})$. We notice that it implies that the corresponding $\Bq=\{q_i\}_{i\in [N]}$ takes at most $\optionsize$ different values, which then implies that the mechanism contains at most $\optionsize$ different options. 

\begin{lemma}\label{lem:discrete-program-most-zeros}
For any integer $N$ and set of parameters $(i^*,\mathbf{i})$ where each parameter is in $[N]$, if $\cP'(D,i^*,\mathbf{i})$ is feasible, then there exists an optimal solution of $\cP'(D,i^*,\mathbf{i})$, denoted as $(q_1^*,\{\hat{q}_i^*\}_{i\in [N-1]})$, such that $|\{i\in [N-1]:\hat{q}_i^*>0\}|\leq \optionsize$.  
\end{lemma}

\begin{proof}

The LP contains $N$ variables ($q_1$ is unconstrained), 1 equality constraint, and $4+3(m-2)=3m-2$ inequality constraints. We add slack variables to change the LP into the canonical form $\max\{\bm{c}^T\bm{x}:~A\bm{x}=\bm{b},\bm{x}\geq \bm{0}\}$. In particular, we replace variable $q_1$ by two non-negative variables $q_1^+,q_1^-$ such that $q_1=q_1^+-q_1^-$. For each inequality constraints, we add one slack variable and make the constraint an equality. 

Now, we have a canonical form LP, with $N+3m-1$ variables, and $3m-1$ equality constraints. Thus by the Fundamental Theorem of linear programming, if $\cP'(D, i^*,\mathbf{i})$ is feasible, any basic {feasible} solution in the canonical-form LP has at most $\optionsize$ non-zero entries. Since there must be an optimal solution of $\cP'(D, i^*,\mathbf{i})$ that correspond to a basic feasible solution, so it must have at most $\optionsize$ non-zero entries. 
\end{proof}

\begin{prevproof}{Lemma}{lem:opt without sigma-IC has k pieces} 
We first prove the case when $D$ is a discrete distribution. Let $\opt'$ be the maximum objective, over all parameters $(i^*,\mathbf{i})$,
and all feasible solutions of $\cP'(D, i^*,\mathbf{i})$.

\begin{claim}\label{claim:program-equivalent}
For any discrete distribution $D=(\{\theta_i\}_{i\in [N]},\{f_i\}_{i\in [N]})$, a solution $\Bq=\{q_i\}_{i\in [N]}$ is feasible for $\cP(D)$ if and only if there exists parameters $(i^*,\mathbf{i})$ such that the solution $\hat{\Bq}=(q_1,\{\hat{q}_i\}_{i\in [N-1]})$ is feasible for $\cP'(D,i^*,\mathbf{i})$. Here $\hat{q}_i=q_{i+1}-q_i,\forall i\in [N-1]$. Moreover, $\opt'=\opt^*$.  
\end{claim}

\begin{proof}
Suppose $\Bq$ is feasible for $\cP(D)$. Let $i^*=\max\{i\in [N]:q_i\leq u_{11}-u_{2m}\}$. For every $k\in \{2,\ldots,m-1\}$, let $i_k=\max\{i\in [N]: q_i\leq \ell_k\}$. We verify that $\hat{\Bq}$ satisfies all constraints of $\cP'(D,i^*,\mathbf{i})$. Constraints {\MarginConstrDis}, {\IntegralConstrDis} follow from constraints {\MarginConstr}, {\IntegralConstr} of $\cP(D)$ accordingly. Constraints {\istarConstrLeft}-{\ikConstrRight} follow from the definition of $i^*$ and all $i_k$'s. Constraint {\IRConstrDis} follows from the definition of $i_k$ and {\IRConstr} of $\cP(D)$. Moreover, by the definition of $i^*$ and the fact that $q_i=q_1+\sum_{j=1}^{i-1}\hat{q}_i,\forall i\in [N]$, the objective of solution $\Bq$ in $\cP(D)$ is equal to the objective of solution $q_1$ and $\hat{\Bq}$ in $\cP'(D,i^*,\mathbf{i})$. Choosing $\Bq$ as the optimal solution for $\cP(D)$ implies that $\opt^*\leq \opt'$. 

On the other hand, suppose $\hat{\Bq}$ is feasible for $\cP'(D,i^*,\mathbf{i})$ for some parameters $(i^*,\mathbf{i})$. For every $i\in \{2,\ldots,N\}$, define $q_i=q_1+\sum_{j=1}^{i-1}\hat{q}_i$. We verify that $\Bq$ is feasible for $\cP(D)$. Constraint {\MonoConstr} holds as all $\hat{q}_i$'s are non-negative; Constraints {\MarginConstr} and  {\IntegralConstr} follow from constraints {\MarginConstrDis} and {\IntegralConstrDis} of $\cP'(D,i^*,\mathbf{i})$ accordingly; Constraint {\IRConstr} follows from constraints {\ikConstrLeft}-{\IRConstrDis} of $\cP'(D,i^*,\mathbf{i})$. Similarly, according to constraints {\istarConstrLeft} and {\istarConstrRight} of $\cP'(D,i^*,\mathbf{i})$ and the fact that $q_i=q_1+\sum_{j=1}^{i-1}\hat{q}_i,\forall i\in [N]$, the objective of solution $\hat{\Bq}$ in $P'(D,i^*,\mathbf{i})$ is equivalent to the objective of solution $\Bq$ in $\cP(D)$. Thus by choosing $\hat{\Bq}$ as the optimal solution among all parameters $(i^*,\mathbf{i})$ and all feasible solutions of $\cP'(D,i^*,\mathbf{i})$, we have $\opt'\leq \opt^*$. Hence, $\opt'=\opt^*$.
\end{proof}

By \Cref{lem:discrete-program-most-zeros} and the definition of $\opt'$, there exist a set of parameters $(i^*,\mathbf{i})$ as well as a feasible solution $\hat{\Bq}$ of $\cP'(D,i^*,\mathbf{i})$ such that: (i) the solution $\hat{\Bq}$ achieves objective $\opt'$ in the program $\cP'(D,i^*,\mathbf{i})$; (ii) $|\{i\in [N-1]:\hat{q}_i^*>0\}|\leq \optionsize$. By \Cref{claim:program-equivalent} and property (i), the corresponding $\Bq$ is the optimal solution of $\cP(D)$. By property (ii), there are at most $\optionsize$ different values among $\{q_i\}_{i\in [N]}$. Let $\MM=\{(q_i,t_i)\}_{i\in [N]}\in \MenuWithQ$ be the {\responsiveIC} and IR mechanism that implements $\Bq$. According to \Cref{equ:payment} and the fact that $q_i\leq q_{i+1},\forall i\in [N-1]$, $t_i=t_j$ for all $i,j$ such that $q_i=q_j$. Thus $\MM$ has option size at most $\optionsize$. 

Next we prove the case for continuous distributions. 

We consider a discretization of the continuous interval $[0,1]$. Let $N\geq 2$ be any integer and $\eps=\frac{1}{N}$. We consider the discretized space $\{0, \eps,2\eps,...,(N-1)\eps\}$. 
We denote by $D^{(N)}$ the discretized distribution of $D$, and $\{f_i^{(N)}\}_{i\in [N]}$, $\{F_i^{(N)}\}_{i\in [N]}$ the pdf and cdf of $D^{(N)}$. Formally, $F_i^{(N)}=F(i\eps)$, and $f_i^{(N)}=F(i\eps)-F((i-1)\eps)$, $\forall i\in [N]$. 

Denote $\cS=\{\Bq: [0,1]\rightarrow [-u_{2m},u_{11}] : q(\cdot) \text{ is non-decreasing}\}$. For every $N$ and every $\Bq\in \cS$, denote $h_N(\Bq)$ the objective of the program in \Cref{fig:discrete-program} under solution $\{q((i-1)\eps)\}_{i\in [N]}$, with respect to the distribution $D^{(N)}$. Formally, 
$$h_N(\Bq)=\sum_{i=1}^{N}\left[q((i-1)\eps)\cdot((i-1)\eps\cdot f_i^{(N)}+\eps\cdot F_i^{(N)})+\min\{(u_{11}-u_{2m}-q((i-1)\eps))\cdot f_i^{(N)},0\}\right]$$

For every $\Bq\in \cS$, denote $h(\Bq)$ the objective of the program in \Cref{fig:continous-program-modified} under solution $\Bq$, with respect to the continuous distribution $D$, i.e., 
$$h(\Bq)=\int_0^1 [(\theta f(\theta) + F(\theta))q(\theta) + \min\{(u_{11}-u_{2m} - q(\theta))f(\theta), 0\}]d\theta$$

\begin{claim}\label{claim:uniform-converge}
$h_N(\cdot)$ uniformly converges to $h(\cdot)$ on $\cS$, i.e., for every $\delta>0$ there exists $N_0\in \mathcal{N}$ such that $|h_N(\Bq)-h(\Bq)|<\delta,\forall N\geq N_0,\Bq\in \cS$.
\end{claim}
\begin{proof}
Fix any $\delta>0$. Fix any $\Bq\in \cS$. Consider the term
\begin{align*}
A_1^{(N)}(\Bq)&=\sum_{i=1}^{N}\left[q((i-1)\eps)\cdot((i-1)\eps\cdot f_i^{(N)}+\eps\cdot F_i^{(N)})\right]-\int_0^1 (\theta f(\theta) + F(\theta))q(\theta)d\theta\\
&=\sum_{i=1}^{N}\left[q((i-1)\eps)\cdot(i\eps\cdot F(i\eps)-(i-1)\eps\cdot F((i-1)\eps))-\int_{(i-1)\eps}^{i\eps} (\theta f(\theta) + F(\theta))q(\theta)d\theta\right]
\end{align*}

Here the second equality follows from the definition of $f_i^{(N)}, F_i^{(N)}$. For every $i\in [N]$, since $q(\cdot)$ is non-decreasing, 
\begin{align*}
\int_{(i-1)\eps}^{i\eps} (\theta f(\theta) + F(\theta))q(\theta)d\theta &\geq q((i-1)\eps)\cdot (\theta F(\theta))\bigg\vert_{\theta=(i-1)\eps}^{i\eps}\\
&=q((i-1)\eps)\cdot(i\eps\cdot F(i\eps)-(i-1)\eps\cdot F((i-1)\eps))
\end{align*}

Thus $A_1^{(N)}(\Bq)\leq 0$. On the other hand, for every $i\in [N]$, similarly we have
\begin{align*}
\int_{(i-1)\eps}^{i\eps} (\theta f(\theta) + F(\theta))q(\theta)d\theta\leq q(i\eps)\cdot(i\eps\cdot F(i\eps)-(i-1)\eps\cdot F((i-1)\eps))
\end{align*}

Thus
\begin{align*}
A_1^{(N)}(\Bq)&\geq \sum_{i=1}^N (q((i-1)\eps)-q(i\eps))\cdot(i\eps\cdot F(i\eps)-(i-1)\eps\cdot F((i-1)\eps))\\
&\geq \max_{i\in [N]} \{i\eps\cdot F(i\eps)-(i-1)\eps\cdot F((i-1)\eps)\}\cdot (q(0)-q(1))\\
&\geq -2u_{11}\cdot \max_{i\in [N]} \{i\eps\cdot F(i\eps)-(i-1)\eps\cdot F((i-1)\eps)\}
\end{align*}
Here the last inequality follows from $q(0)\geq -u_{2m}\geq -u_{11}$ and $q(1)\leq u_{11}$. Since $F$ is continuous on the closed interval $[0,1]$, the function $\theta\cdot F(\theta)$ is continuous on $[0,1]$ and thus it is uniformly continuous. Hence there exists $N_1\in \mathcal{N}$ such that $\max_{i\in [N]} \{i\eps\cdot F(i\eps)-(i-1)\eps\cdot F((i-1)\eps)\}<\frac{\delta}{4u_{11}},\forall N\geq N_1$. Thus $A_1^{(N)}(\Bq)>-\frac{\delta}{2}$. 

Now consider the other term
\begin{align*}
A_2^{(N)}(\Bq)&=\sum_{i=1}^N\min\{(u_{11}-u_{2m}-q((i-1)\eps))\cdot f_i^{(N)},0\}-\int_{0}^1 \min\{(u_{11}-u_{2m} - q(\theta))f(\theta), 0\}d\theta\\
&=\sum_{i=1}^N\left[ f_i^{(N)}\cdot\min\{\bar{q}((i-1)\eps),0\}-\int_{(i-1)\eps}^{i\eps} \min\{\bar{q}(\theta), 0\}f(\theta)d\theta\right],
\end{align*}
where $\bar{q}(\theta)=u_{11}-u_{2m}-q(\theta)\in [-u_{2m},u_{11}]$ is non-increasing on $\theta$. For every $i\in [N]$, we denote $T_i$ the $i$-th term in the sum. Let $\theta^*=\sup\{\theta\in [0,1]:\bar{q}(\theta)\leq 0\}$ and $i^*=\max\{i\in [N]: (i-1)\eps\leq \theta^*<i\eps\}$. 

When $i<i^*$, since $\bar{q}(\cdot)$ is non-increasing and $f_i^{(N)}=F(i\eps)-F((i-1)\eps)$,
\begin{align*}
&T_i\geq f_i^{(N)}\cdot \bar{q}((i-1)\eps)-\bar{q}((i-1)\eps)\cdot \int_{(i-1)\eps}^{i\eps}f(\theta)d\theta=0\\
&T_i\leq f_i^{(N)}\cdot [\bar{q}((i-1)\eps)-\bar{q}(i\eps)]
\end{align*}

When $i>i^*$. Clearly $T_i=0$. When $i=i^*$, 
$$|T_{i^*}|=\left|f_{i^*}^{(N)}\cdot \bar{q}((i^*-1)\eps)-\int_{(i^*-1)\eps}^{\theta^*} \bar{q}(\theta)f(\theta)d\theta\right|\leq 2u_{11}\cdot f_{i^*}^{(N)},$$
where the inequality follows from the fact that $|\bar{q}(\theta)|\leq \max\{u_{11},u_{2m}\}=u_{11},\forall \theta\in [0,1]$ and that $F(\theta^*)-F((i^*-1)\eps)\leq f_{i^*}^{(N)}$. Thus
\begin{align*}
|A_2^{(N)}(\Bq)|\leq\left|\sum_{i=1}^{i^*-1}T_i\right|+|T_{i^*}|&\leq \left|\sum_{i=1}^{i^*-1}f_i^{(N)}\cdot [\bar{q}((i-1)\eps)-\bar{q}(i\eps)]\right|+2u_{11}\cdot f_{i^*}^{(N)}\\
&\leq (|\bar{q}((i-1)\eps)-\bar{q}(0)|+2u_{11})\cdot \max_{i\in [N]}f_{i}^{(N)}\\
&\leq 4u_{11}\cdot \max_{i\in [N]}f_{i}^{(N)}
\end{align*}
Since $F(\cdot)$ is uniformly continuous on $[0,1]$, there exists $N_2\in\mathcal{N}$ such that $\max_{i}f_{i}^{(N)}=\max_{i}(F(i\eps)-F((i-1)\eps))<\frac{\delta}{8u_{11}},\forall N\geq N_2$. Thus $|A_2^{(N)}(\Bq)|<\frac{\delta}{2}$.

Combining everything together, for every $N\geq N_0=\max\{N_1,N_2\}$, for every $\Bq\in \cS$, we have
$$|h_N(\Bq)-h(\Bq)|= |A_1^{(N)}(\Bq)+A_2^{(N)}(\Bq)|<\frac{\delta}{2}+\frac{\delta}{2}=\delta$$

Thus $h_N(\cdot)$ uniformly converges to $h(\cdot)$ on $\cS$.
\end{proof}

Back to the proof of \Cref{lem:opt without sigma-IC has k pieces}. Denote $\pcont$ the program in \Cref{fig:continous-program-modified} and $\opt^*$ the supremum of the objective over all feasible solutions of $\pcont$. We first argue that there exists a feasible solution $\Bq^*$ whose objective $h(\Bq^*)$ equals to $\opt^*$. 
\begin{claim}\label{sup_equal_max}
There exists a feasible solution $\Bq^*$ to $\pcont$, whose objective $h(\Bq^*)$ equals to $\opt^*$.
\end{claim}
\begin{proof}
Let $\{\Bq_\ell^*\}_{\ell\in \mathcal{N}_+}$ be a sequence of feasible solutions such that $\lim_{\ell\to\infty}h(\Bq_\ell^*)=\opt^*$. We notice that for every $\ell$, the feasible solution $\Bq_{\ell}^*(\cdot)$ is a non-decreasing function mapping $[0,1]$ to a bounded interval $[-u_{2m},u_{11}]$. Thus $\{\Bq_\ell^*\}_{\ell\in \mathcal{N}_+}$ is a sequence of non-decreasing functions and it's uniformly-bounded. By Helly's selection theorem (see for instance \cite{ambrosio2000functions,barbu2012convexity}), there exists a subsequence $\{\Bq^*_{\ell_i}\}_{i\in \mathcal{N}_+}$ and a function $\Bq^*:[0,1]\to [-u_{2m},u_{11}]$ such that $\{\Bq^*_{\ell_i}\}_i$ pointwisely converges to $\Bq^*$.   

For every $\theta\in [0,1]$ and $q\in [-u_{2m},u_{11}]$, define
\begin{equation}\label{equ:objective-G}
G(\theta,q)=(\theta f(\theta) + F(\theta))q + \min\{(u_{11}-u_{2m} - q)f(\theta), 0\}    
\end{equation} 
For every $\theta\in [0,1]$, $i\in \mathcal{N}_+$, define $g_i(\theta)=G(\theta,q^*_{\ell_i}(\theta))$ and $g(\theta)=G(\theta,q^*(\theta))$. Then since $G(\theta,q)$ is continuous on $q$ for every $\theta$, $\{g_i\}_{i\in \mathcal{N}_+}$ pointwisely converges to $g$. Moreover, define $\bar{g}(\theta)=u_{11}\cdot (\theta f(\theta) + F(\theta) + f(\theta))\geq 0,\forall \theta\in [0,1]$. Then $\bar{g}(\cdot)$ is integrable in $[0,1]$. And for every $i$ and $\theta$, since $q^*_{\ell_i}(\theta)\in [-u_{2m},u_{11}]$ and $u_{11}\geq u_{2m}$, we have $|g_i(\theta)|\leq \bar{g}(\theta)$. Thus by the dominated convergence theorem,
$$\opt^*=\lim_{i\to\infty}h(\Bq^*_{\ell_i})=\lim_{i\to\infty}\int_0^1 g_i(\theta)d\theta=\int_0^1(\lim_{i\to\infty}g_i(\theta))d\theta=\int_0^1g(\theta)d\theta=h(\Bq^*)$$

It remains to verify that $\Bq^*$ is a feasible solution to $\pcont$. Since $\{\Bq^*_{\ell_i}\}_i$ is a sequence of feasible solutions, all constraints in $\pcont$ hold for $\Bq^*_{\ell_i}$. For every inequality, we take the limit $i\to\infty$ on both sides. Constraints {\MonoConstr}, {\MarginConstr} hold for $\Bq^*$ since $\{\Bq^*_{\ell_i}\}_i$ pointwisely converges to $\Bq^*$. In order for constraints {\IntegralConstr} and {\IRConstr} to hold for $\Bq^*$, it's sufficient to argue that we can swap the limit and integral for both inequalities. This is because for every $i$ and $\theta$, we have $|q^*_{\ell_i}(\theta)|\leq u_{11}$ and $|(q^*_{\ell_i}(\theta)-\ell_k)\cdot \ind[q^*_{\ell_i}(\theta)\leq \ell_k]|\leq |q^*_{\ell_i}(\theta)|+|\ell_k|\leq 2u_{11}$ (recall that $|\ell_k|=|u_{1,m+1-k}-u_{2,m+1-k}|\leq \max\{u_{11},u_{2m}\}=u_{11}$). Thus by the dominated convergence theorem, constraints {\IntegralConstr} and {\IRConstr} both hold for $\Bq^*$. Thus $\Bq^*$ is feasible for $\pcont$. 
\end{proof}

For every $N\geq 2$, denote $\opt_N$ the optimum of $\cP(D^{(N)})$, the program in \Cref{fig:discrete-program} with respect to the distribution $D^{(N)}$. By \Cref{lem:discrete-program-most-zeros} and \Cref{claim:program-equivalent}, there exists a set of parameters $(i^*,\mathbf{i})$ and a feasible solution $(q_1^{(N)},\{\hat{q}_i^{(N)}\}_{i\in [N-1]})$ to $\cP'(D^{(N)},i^*,\mathbf{i})$ such that: (i) the solution has objective exactly $\opt_N$, and (ii) $|\{i\in [N-1]:\hat{q}_i^*>0\}|\leq \optionsize$. 
For every $i\in \{2,\ldots,N\}$, let $q_i^{(N)}=q_1^{(N)}+\sum_{j=1}^{i-1}\hat{q}_j^{(N)}$. Define $\Bq^{(N)}:[0,1]\to [-u_{2m},u_{11}]$ as follows: $q^{(N)}(\theta)=q_i^{(N)}, \forall i\in [N], \theta\in [(i-1)\eps,i\eps)$, and $q^{(N)}(1)=q_N^{(N)}$. By \Cref{claim:program-equivalent}, $\{q^{(N)}((i-1)\eps)\}_{i\in [N]}$ is a feasible solution to program $\cP(D^{(N)})$. Moreover, $h_N(\Bq^{(N)})=\opt_N$.


For every $i\in [N]$, by \Cref{claim:program-equivalent}, $\{q_i^{(N)}\}_{i\in [N]}$ is a feasible solution to $\cP(D^{(N)})$. By the definition of $\Bq^{(N)}$, it's not hard to verify that $\Bq^{(N)}$ is a feasible solution to $\pcont$. 
By property (ii), $\text{Im}(\Bq)=\{q(\theta):\theta\in [0,1]\}$, the image of $\Bq(\cdot)$, has size at most $\optionsize$. Let $n=\optionsize$ and $\cS'\subseteq \cS$ be the set of all feasible solutions $\Bq$ to $\pcont$ such that $|\text{Im}(\Bq)|\leq n$. 
We notice that there is a mapping $\Phi$ from every $\Bq\in \cS'$ to a set of $(2n-1)$ real numbers $\xi=(\{\theta_i\}_{i\in [n-1]},\{q_i\}_{i\in [n]})$ such that $q(\theta)=q_i, \forall i\in [n], \theta\in [\theta_{i-1},\theta_i)$ and $q(1)=q_{n}$. Here $\theta_0=0$, $\theta_{n}=1$. 

Let $\cT$ be the set of $\xi$'s that satisfy all of the following properties:
\begin{enumerate}
    \item $0\leq \theta_1\leq \ldots\leq \theta_{n-1}\leq 1$.
    \item $-u_{2m}\leq q_1\leq \ldots\leq q_{n}\leq u_{11}$.
    \item $\sum_{i=1}^{n} (\theta_{i}-\theta_{i-1})\cdot q_i=u_{11}-u_{2m}$.
    \item $\sum_{i=1}^{n} (\theta_{i}-\theta_{i-1})\cdot (q_i-\ell_k)\cdot \ind[q_i\leq \ell_k]\geq u_{2,m+1-k}-u_{2m},\forall k\in \{2,3,\ldots,m-1\}$.
\end{enumerate}

Then it's not hard to verify that $\Phi$ is an 1-1 mapping from $\cS'$ to $\cT$. Denote $\Phi^{-1}$ the inverse of $\Phi$. Moreover, since every constraint above is either equality or non-strict inequality among the sum and product of numbers $(\{\theta_i\}_{i\in [n-1]},\{q_i\}_{i\in [n]})$, thus $\cT$ is a closed subset of $\mathbb{R}^{2n-1}$. Thus $\cT$ is a compact space under the $L_{\infty}$-norm. Therefore, the sequence $\{\Phi(\Bq^{(N)})\}_{N}$ contains a subsequence $\{\Phi(\Bq^{(N_\ell)})\}_{\ell\in \mathbb{N}_+}$ that converges to some $\xi'\in \cT$. Let $\Bq'=\Phi^{-1}(\xi')\in \cS'$. 

\begin{claim}\label{claim:limit-subsequence}
$\lim_{\ell\to\infty}h_{N_\ell}(\Bq^{(N_{\ell})})=h(\Bq')$.
\end{claim}
\begin{proof}
Fix any $\delta>0$. Since both functions $\theta F(\theta)$ and $F(\theta)$ are continuous, there exists some $\eta(\delta)\in (0,\frac{\delta}{8n})$ such that $|\theta F(\theta)-\theta' F(\theta')|<\frac{\delta}{16n\cdot u_{11}}$ and $|F(\theta)-F(\theta')|<\frac{\delta}{16n\cdot u_{11}}$ for any $\theta$ and $\theta'$ with $|\theta-\theta'|<\eta(\delta)$. 

Since $\{\Phi(\Bq^{(N_\ell)})\}_{\ell\in \mathbb{N}_+}$ converges to $\Phi(\Bq')$, then there exists $\ell_1$ such that for every $\ell\geq \ell_1$, $||\Phi(\Bq^{(N_\ell)})-\Phi(\Bq')||_{\infty}<\eta(\delta)$. Denote $\Phi(\Bq^{(N_\ell)})=(\{\theta_i^{\ell}\}_{i\in [n-1]},\{q_i^{\ell}\}_{i\in [n]})$ and $\Phi(\Bq')=(\{\theta_i'\}_{i\in [n-1]},\{q_i'\}_{i\in [n]})$. Let $\theta_0^{\ell}=\theta_0'=0$ and $\theta_n^{\ell}=\theta_n'=1$. We are going to bound $|h(\Bq^{(N_{\ell})})-h(\Bq')|$. We notice that
\begin{align*}
h(\Bq^{(N_{\ell})})&=
\sum_{i=1}^{n}\left[q_i^\ell\cdot \int_{\theta_{i-1}^{\ell}}^{\theta_i^{\ell}}(\theta f(\theta)+F(\theta))d\theta+(F(\theta_{i}^\ell)-F(\theta_{i-1}^{\ell}))\cdot \min\{u_{11}-u_{2m}-q_i^{\ell},0\}\right]\\
&=\sum_{i=1}^{n}\left[q_i^\ell\cdot (\theta_{i}^\ell F(\theta_{i}^\ell)-\theta_{i-1}^\ell F(\theta_{i-1}^\ell))+(F(\theta_{i}^\ell)-F(\theta_{i-1}^{\ell}))\cdot \min\{u_{11}-u_{2m}-q_i^{\ell},0\}\right]    
\end{align*}
Similarly,
$$h(\Bq')=\sum_{i=1}^{n}\left[q_i'\cdot (\theta_{i}' F(\theta_{i}')-\theta_{i-1}' F(\theta_{i-1}'))+(F(\theta_{i}')-F(\theta_{i-1}'))\cdot \min\{u_{11}-u_{2m}-q_i',0\}\right]
$$

Since $|\theta_i^{\ell}-\theta_i'|<\eta(\delta),\forall i\in [n-1]$, then for every $i\in [n]$, $|(\theta_{i}^\ell F(\theta_{i}^\ell)-\theta_{i-1}^\ell F(\theta_{i-1}^\ell))-(\theta_{i}' F(\theta_{i}')-\theta_{i-1}' F(\theta_{i-1}'))|<2\cdot \frac{\delta}{16n\cdot u_{11}}$, and $|(F(\theta_{i}^\ell)-F(\theta_{i-1}^{\ell}))-(F(\theta_{i}')-F(\theta_{i-1}'))|<2\cdot \frac{\delta}{16n\cdot u_{11}}$. 

Since $||\Phi(\Bq^{(N_\ell)})-\Phi(\Bq')||_{\infty}<\eta(\delta)<\frac{\delta}{8n}$, we have $$|\min\{u_{11}-u_{2m}-q_i^{\ell},0\}-\min\{u_{11}-u_{2m}-q_i',0\}|\leq |q_i^{\ell}-q_i'|<\frac{\delta}{8n}$$ 

Moreover, $|q_i^{\ell}|\leq u_{11}$, $|\min\{u_{11}-u_{2m}-q_i^\ell,0\}|\leq u_{11}$ and $\theta_{i}' F(\theta_{i}')-\theta_{i-1}' F(\theta_{i-1}')\in [0,1], F(\theta_{i}')-F(\theta_{i-1}')\in [0,1]$. Using the inequality $|ab-cd|\leq |a|\cdot |b-d|+|d|\cdot |a-c|$ for every $a,b,c,d\in \mathbb{R}$, we have for every $\ell\geq \ell_1$, 
$$|h(\Bq^{(N_{\ell})})-h(\Bq')|\leq \sum_{i=1}^{n}2\cdot(u_{11}\cdot \frac{\delta}{8n\cdot u_{11}}+1\cdot \frac{\delta}{8n})<\frac{\delta}{2}$$

Moreover, by \Cref{claim:uniform-converge}, there exists $\ell_2$ such that for every $\ell\geq \ell_2$, $|h_{N_\ell}(\Bq^{(N_{\ell})})-h(\Bq^{(N_{\ell})})|<\frac{\delta}{2}$. Thus when $\ell\geq \max\{\ell_1,\ell_2\}$, we have
$|h_{N_\ell}(\Bq^{(N_{\ell})})-h(\Bq')|\leq |h(\Bq^{(N_{\ell})})-h(\Bq')|+|h_{N_\ell}(\Bq^{(N_{\ell})})-h(\Bq^{(N_{\ell})})|<\delta$.
\end{proof}

We prove $h(\Bq')\geq \opt^*$. For every $N$, we construct a feasible solution of $\cP(D^{(N)})$. For every $i\in [N]$, let $p_i^{(N)}=\frac{1}{\eps}\cdot \int_{(i-1)\eps}^{i\eps}q^*(\theta)d\theta$ (see \Cref{sup_equal_max} for the definition of $\Bq^*$). We verify that $\{p_i^{(N)}\}_{i\in [N]}$ is a feasible solution to $\cP(D^{(N)})$. Constraint {\MonoConstr} follows from $q^*(\cdot)$ being non-decreasing. For constraint {\MarginConstr}, since $q^*(\cdot)$ is non-decreasing, $p_1^{(N)}\geq\frac{1}{\eps}\cdot \int_{0}^{\eps}q^*(0)d\theta=q^*(0)\geq -u_{2m}$. Similarly $p_N^{(N)}\leq q^*(1)\leq u_{11}$. For constraint {\IntegralConstr}, we have
$$\eps\cdot \sum_{i=1}^Np_i^{(N)}=\sum_{i=1}^N\int_{(i-1)\eps}^{i\eps}q^*(\theta)d\theta=\int_0^1q^*(\theta)d\theta=u_{11}-u_{2m}$$

For constraint {\IRConstr}, for every $k\in \{2,\ldots,m-1\}$. Let $\theta_k=\sup\{\theta\in [0,1]:q^*(\theta)\leq \ell_k\}$ and let $i^*$ be the unique number such that $(i^*-1)\eps\leq \theta_k<i^*\eps$. Then $p_i^{(N)}\leq \ell_k,\forall i<i^*$ and $p_i^{(N)}>\ell_k,\forall i>i^*$. We have
\begin{equation}\label{qN-feasible}
\begin{aligned}
\eps\cdot\sum_{i=1}^N\min\{p_i^{(N)}-\ell_k,0\}&=\sum_{i=1}^{i^*-1}\left(\int_{(i-1)\eps}^{i\eps}q^*(\theta)d\theta-\eps\ell_k\right)+\eps\cdot\min\{p_{i^*}^{(N)}-\ell_k,0\}\\
&=\int_{0}^{(i^*-1)\eps}(q^*(\theta)-\ell_k)d\theta+\eps\cdot\min\{p_{i^*}^{(N)}-\ell_k,0\}\\
&\geq \int_{0}^{(i^*-1)\eps}(q^*(\theta)-\ell_k)d\theta+\int_{(i^*-1)\eps}^{\theta_k}(q^*(\theta)-\ell_k)d\theta\\
&=\int_0^1(q^*(\theta)-\ell_k)\ind[q^*(\theta)\leq\ell_k]d\theta\geq u_{2,m+1-k}-u_{2m}
\end{aligned}
\end{equation}

The first inequality of \Cref{qN-feasible} is because: By the definition of $\theta_k$, $\int_{\theta_k}^{i^*\eps}(q^*(\theta)-\ell_k)d\theta\geq 0$. Thus 
$$\eps\cdot\min\{p_{i^*}^{(N)}-\ell_k,0\}\geq \eps\cdot (p_{i^*}^{(N)}-\ell_k) =\int_{(i^*-1)\eps}^{i^*\eps}(q^*(\theta)-\ell_k)d\theta\geq \int_{(i^*-1)\eps}^{\theta_k}(q^*(\theta)-\ell_k)d\theta$$
The third equality of \Cref{qN-feasible} follows from the definition of $\theta_k$. The last inequality follows from that $\Bq^*$ is a feasible solution to $\pcont$. Thus $\{p_i^{(N)}\}_{i\in [N]}$ is a feasible solution to $\cP(D^{(N)})$. 
Define $\Bp^{(N)}:[0,1]\to [-u_{2m},u_{11}]$ as follows: $p^{(N)}(\theta)=p_i^{(N)}, \forall i\in [N], \theta\in [(i-1)\eps,i\eps)$, and $p^{(N)}(1)=p_N^{(N)}$. Then $h_N(\Bp^{(N)})$ is exactly the objective of $\{p_i^{(N)}\}_{i\in [N]}$ with respect to the program $\cP(D^{(N)})$. By the optimality of $\Bq^{(N)}$, $h_N(\Bp^{(N)})\leq h_N(\Bq^{(N)})$. Thus $\forall \ell\in \mathcal{N}_+$, $h_{N_\ell}(\Bp^{(N_{\ell})})\leq h_{N_{\ell}}(\Bq^{(N_{\ell})})$. 

We will argue that $h_{N_\ell}(\Bp^{(N_{\ell})})$ converges to $\opt^*$ as $\ell\to\infty$. If this is true, then by taking the limit on both sides of the above inequality, we have that $\opt^*\leq \lim_{\ell\to\infty}h_{N_\ell}(\Bq^{(N_{\ell})})=h(\Bq')$, where the equality follows from \Cref{claim:limit-subsequence}.

Since $q^*(\cdot)$ is non-decreasing, for every $i\in [N]$ and $\theta\in [(i-1)\eps,i\eps)$, we have $$|p^{(N)}(\theta)-q^*(\theta)|=|p_i^{(N)}-q^*(\theta)|\leq q^*(i\eps)-q^*((i-1)\eps).$$ 

Thus if $q^*(\cdot)$ is continuous at $\theta$, then $\lim_{N\to\infty} |p^{(N)}(\theta)-q^*(\theta)|=0$. Moreover, since $q^*(\cdot)$ is a monotone function in a closed interval $[0,1]$, the set of non-continuous points has zero measure. Thus $\lim_{N\to\infty} p^{(N)}(\theta)=q^*(\theta)$ almost everywhere. 

Similar to the proof of \Cref{sup_equal_max}, we have 
$$\lim_{\ell\to\infty}h({\Bp^{(N_{\ell})}})=\lim_{\ell\to\infty}\int_0^1G(\theta,{p^{(N_{\ell})}}(\theta))d\theta=\int_0^1G(\theta,q^*(\theta))d\theta=h(\Bq^*),$$
where the definition of $G$ is shown in \Cref{equ:objective-G}. Thus, by \Cref{claim:uniform-converge}, we have $\lim_{\ell\to\infty}h_{N_\ell}({\Bp^{(N_{\ell})}})=\lim_{\ell\to\infty}h({\Bp^{(N_{\ell})}})=h(\Bq^*)=\opt^*$. 
We have proved that $h(\Bq')\geq \opt^*$. On the other hand, since $\Bq'\in \cS'$, $\Bq'$ is feasible solution to $\pcont$. Thus $h(\Bq')\leq \opt^*$. Thus we have $h(\Bq')=\opt^*$, which implies that $\Bq'$ is indeed an optimal solution to $\pcont$. Let $\Phi(\Bq')=(\{\theta_i'\}_{i\in [n-1]},\{q_i'\}_{i\in [n]})$. Then $q'(\theta)=q_i', \forall i\in [n], \theta\in [\theta_{i-1}',\theta_i')$ and $q'(1)=q_{n}'$. Since $\Bq'$ is a feasible solution to $\pcont$, let $\MM'=\{q(\theta),t(\theta)\}_{\theta\in [0,1]}\in \MenuWithQ$ be the {\responsiveIC} and IR mechanism that implements $\Bq'$ (\Cref{lem:ic in q space}). To show that $\MM'$ contains at most $n=\optionsize$ different options, it suffices to prove that for every $i\in [n]$, $t(\theta)$ is a constant in $[\theta_{i-1}',\theta_i)$. {Suppose there exist $\theta,\theta'\in [\theta_{i-1}',\theta_i)$, $\theta\not=\theta'$ such that $t(\theta)\not=t(\theta')$. Without loss of generality, assume $t(\theta)<t(\theta')$. Then when the buyer has type $\theta'$, misreporting $\theta$ induces the same experiment, but a lower price. It contradicts with the fact that $\MM'$ is {\responsiveIC}. Thus $\MM'$ is an optimal {\responsiveIC}, IR mechanism that contains at most $\optionsize$ different options. We finish the proof.
}

\end{prevproof}

%% file: proof_characterization.tex
\section{Missing Details from Section~\ref{sec:characterization}}\label{sec:proof_characterization}




In this section we will provide a proof of \Cref{cor:special_distributions}. We first provide a sufficient and necessary condition under which selling complete information achieves the highest revenue among all {\responsiveIC} and IR mechanisms (\Cref{thm:characterization-full-info}). Note that the same condition also guarantees that selling complete information is the optimal mechanism among all IC and IR mechanisms. 

Since $f$ is strictly positive on $[0,1]$, $F(\cdot)$ is strictly increasing. Denote $F^{-1}(\cdot)$ the inverse of $F$. Recall that for every $\theta\in [0,1]$, $\varphi^-(\theta)=\theta f(\theta)+F(\theta)$ and $\varphi^+(\theta)=(\theta-1)f(\theta)+F(\theta)$. Intuitively, they can be viewed as the agent's virtual values when $q(\theta)\leq u_{11}-u_{2m}$ and when $q(\theta)>u_{11}-u_{2m}$ respectively. If either $\varphi^-$ or $\varphi^+$ is not monotonically non-decreasing, we iron the functions using the generalized ironing procedure by Toikka~\cite{toikka2011ironing}, which is a generalization of the ironing procedure by Myerson~\cite{Myerson81}. Denote $\tilde{\varphi}^-$ and  $\tilde{\varphi}^+$ the ironed virtual value functions of $\varphi^-$ and  $\varphi^+$ respectively. See \Cref{subsec:characterization-responsiveIC} for formal definitions of the ironing procedure and the ironed interval.   

\begin{theorem}\label{thm:characterization-full-info}
Let $\MM^*$ be the menu that contains only the fully informative experiment with price $p>0$. Then $\MM^*$ achieves the maximum revenue among all {\responsiveIC} and IR mechanisms if and only if, there exist two multipliers $\eta^*$ and $\lambda^*\geq 0$ such that: 
\begin{enumerate}
    \item $p\leq \widehat{p}$, where $\widehat{p}=\min\left\{\frac{(u_{11}-u_{12})u_{2m}}{u_{11}-u_{12}+u_{22}},\frac{(u_{2m}-u_{2,m-1})u_{11}}{u_{2m}+u_{1,m-1}-u_{2,m-1}}\right\}.$ 
    Moreover, $\lambda^*>0$ only when $m=3$, $u_{12}-u_{22}=u_{11}-1$, and $p=\widehat{p}=1-u_{22}$.
    \item Let $\theta_p^L=\frac{p}{u_{11}}$ and $\theta_p^H=1-\frac{p}{u_{2m}}$ be the two points where the utility function of buying the fully informative experiment, i.e., a linear function, intersects with the IR curve. For every $\theta\in [0,\theta_p^{L})$, $\tilde{\varphi}^{-}(\theta)-\eta^*+\lambda^*\leq 0$; For every $\theta\in [\theta_p^{L},\theta_p^{H})$, $\tilde{\varphi}^{-}(\theta)-\eta^*+\lambda^*\geq 0$ and $\tilde{\varphi}^{+}(\theta)-\eta^*\leq 0$; For every $\theta\in [\theta_p^{H},1]$, $\tilde{\varphi}^{+}(\theta)-\eta^*\geq 0$. 
    \item $\theta_p^{L}$ is not in the interior of any ironed interval of $\tilde{\varphi}^-(\cdot)$. And $\theta_p^{H}$ is not in the interior of any ironed interval of $\tilde{\varphi}^+(\cdot)$.
\end{enumerate}
\end{theorem}

Here is a proof sketch of \Cref{thm:characterization-full-info}. We first prove an exact characterization of the optimal {\responsiveIC} and IR mechanism in $\MenuWithQ$, i.e., the optimal solution $\Bq^*=\{q^*(\theta)\}_{\theta\in [0,1]}$ of the program in \Cref{fig:continous-program-modified}, by Lagrangian duality (\Cref{thm:characterization}). It is a generalization of the characterization by Bergemann et al.~\cite{BergemannBS2018design} to $m\geq 3$ actions. We Lagrangify constraints {\IntegralConstr} and {\IRConstr} in the program. $\Bq=\{q(\theta)\}_{\theta\in [0,1]}$ is an optimal solution iff there exist Lagrangian multipliers that satisfy the KKT conditions. 
As a second step, we apply this characterization to $\Bq$ that corresponds to selling the complete information. To simplify our characterization, we show that at most two of the Lagrangian multipliers can be non-zero (in contrast to $\Theta(m)$ non-zero multipliers as in \Cref{thm:characterization}), when the solution $\Bq$ corresponds to a mechanism that only sells complete information. This is enabled by showing that in order to be the optimal menu, the price of the complete information can not be too high (\Cref{lem:endpoints of the full info}).

We notice that if a fully informative menu achieves the maximum revenue among all {\responsiveIC} and IR mechanisms, then it must also achieve the maximum revenue among all IC and IR mechanisms.  
In the rest of this section, we focus on {\responsiveIC} and IR mechanisms. By \Cref{lem:describe the menu with q} and \Cref{obs:q to exp}, we only need to consider mechanisms where every experiment is {\semiinformative} (\Cref{table:two-patterns}). Recall that $\MenuWithQ$ is the set of {\responsiveIC}, IR mechanisms that have this format. Throughout this section, we use $\Bq=\{q(\theta)\}_{\theta\in[0,1]}$ to represent the experiments of the mechanism. Recall that the optimal mechanism in $\MenuWithQ$ is captured in the program in \Cref{fig:continous-program-modified}.

In \Cref{subsec:characterization-responsiveIC}, we come up with an exact characterization of the optimal {\responsiveIC} and IR mechanism $\MM^*\in \MenuWithQ$, in any environment with 2 states and $m$ actions. In \Cref{subsec:characterization-full-info}, we apply this result to achieve a sufficient condition under which the fully informative menu is optimal. 

\subsection{Characterization of the Optimal {\ResponsiveIC} and IR Mechanism}\label{subsec:characterization-responsiveIC}

\notshow{
\mingfeinote{Mingfei: I moved the following part to \Cref{sec:upper_bound_full_info}. Please make changes on \Cref{sec:upper_bound_full_info} and remove things here.}

We notice that in \Cref{fig:continous-program}, there is an IR constraint for every type $\theta$. In order to state the characterization, we propose an equivalent program that only contains $O(m)$ constraints. 
Recall that for every $\theta\in [0,1]$, $u(\theta)=\max_{k\in [m]}\{\theta\cdot u_{1k}+(1-\theta)\cdot u_{2k}\}$. For every $k\in [m]$, let $h_k(\theta)=\theta\cdot u_{1,m+1-k}+(1-\theta)\cdot u_{2,m+1-k}=\theta\cdot \ell_k+u_{2,m+1-k},\forall \theta\in [0,1]$ be the agent's value function when she has no information, and follows action $m+1-k$. Here $\ell_k=u_{1,m+1-k}-u_{2,m+1-k}$. Recall that $\ell_k$ is also the slope for the $i$-th piece of the IR curve. 

Given any feasible solution $\{q(\theta)\}_{\theta\in [0,1]}$ of the program in \Cref{fig:continous-program}, the IR constraint at type $\theta$ is equivalent to: $u_{2m}+\int_{0}^{\theta}q(x)dx\geq h_k(\theta),\forall k\in [m]$. For every $k\in [m]$, let $\theta_k=\sup\{\theta\in [0,1]:q(\theta)\leq \ell_k\}$. Consider the following function $$u_{2m}+\int_{0}^{\theta}q(x)dx- h_k(\theta)=u_{2m}+\int_{0}^{\theta}q(x)dx-\theta\cdot \ell_k-u_{2,m+1-k}$$

By taking the derivative on $\theta$, we can see that the above function is minimized at $\theta_k$. Thus the set of constraints $\left\{u_{2m}+\int_{0}^{\theta}q(x)dx\geq h_k(\theta),\forall \theta\in [0,1]\right\}$ is captured by a single constraint $u_{2m}+\int_{0}^{\theta_k}q(x)dx\geq h_k(\theta_k)$. We notice that $\theta_k=\int_{0}^1\ind[q(x)\leq \ell_k]dx$. By the definition of $h_k$,
the constraint can be rewritten as

\begin{equation}\label{equ:IR_consraint-2}
\int_0^{1}(q(x)-\ell_k)\cdot \ind[q(x)\leq \ell_k]dx\geq u_{2,m+1-k}-u_{2m}    
\end{equation}

A feasible $\Bq$ must satisfy that $q(\theta)\in [-u_{2m},u_{11}]$ and $\int_{0}^1q(x)dx=u_{11}-u_{2m}$. Also $\ell_1=-u_{2m},\ell_m=u_{11}$. Thus Inequality~\eqref{equ:IR_consraint-2} is trivial when $k=1,m$. In our modified program we only include the constraint for $2\leq k\leq m-1$. The modified program is shown in \Cref{fig:continous-program-modified}. 
In \Cref{lem:modified-program-equivalent} we prove that the programs in \Cref{fig:continous-program} and \Cref{fig:continous-program-modified} are equivalent.


\begin{figure}[ht!]
\colorbox{MyGray}{
\begin{minipage}{.98\textwidth}
$$\quad\textbf{max  } \int_0^1 [(\theta f(\theta) + F(\theta))q(\theta) + \min\{(u_{11}-u_{2m} - q(\theta))f(\theta), 0\}]d\theta$$
\vspace{-.3in}
  \begin{align*}
\text{s.t.} \qquad 
 &{\MonoConstr}\quad q(\theta) \text{ is non-decreasing in }\theta\in [0,1]\\
 & {\MarginConstr}\quad q(0)\geq-u_{2m},\quad q(1)\leq u_{11}\\
 &{\IntegralConstr}\quad \int_{0}^1 q(\theta) d\theta = u_{11}-u_{2m} \\
 & {\IRConstr}\quad \int_0^{1}(q(x)-\ell_k)\cdot \ind[q(x)\leq \ell_k]dx\geq u_{2,m+1-k}-u_{2m}, &\forall k\in \{2,3,...,m-1\}
\end{align*}
\end{minipage}}
\caption{The Equivalent Program for the Optimal {\ResponsiveIC}, IR Mechanism}~\label{fig:continous-program-modified-2}
\end{figure}

\begin{lemma}\label{lem:modified-program-equivalent}
A solution $\Bq=\{q(\theta)\}_{\theta\in [0,1]}$ is feasible in the program in \Cref{fig:continous-program} if and only if it's feasible in the program in \Cref{fig:continous-program-modified}.
\end{lemma}

\begin{proof}
We denote $Q$ the program in \Cref{fig:continous-program} and $Q'$ the program in \Cref{fig:continous-program-modified}. Suppose $\Bq$ is feasible in $Q$. For every $k\in [m]$, let $\theta_k=\sup\{\theta\in [0,1]:q(\theta)\leq \ell_k\}$. Then since $\theta_k=\int_{0}^1\ind[q(\theta)\leq \ell_k]dx$, we have
\begin{align*}
\int_0^1(q(x)-\ell_k)\cdot \ind[q(x)\leq \ell_k]dx&=\int_0^{\theta_k}q(x)dx-\theta_k\ell_k\\
&\geq u(\theta_k)-u_{2m}-\theta_k\ell_k \geq u_{2,m+1-k}-u_{2m}
\end{align*}

Here the first inequality follows from the IR constraint of program $Q$ at type $\theta_k$; the second inequality follows from the definition of the IR curve $u(\cdot)$. Thus $\Bq$ is feasible in $Q$.

On the other hand, suppose $\Bq$ is feasible in the program $Q'$. For every $k\in [m]$, consider the following function $H_k:[0,1]\to \mathbb{R}$: $H(\theta)=\int_0^\theta q(x)dx-\theta\cdot \ell_k$. Since $q(\theta)$ is non-decreasing in $\theta$, $H_k(\theta)$ is minimized at $\theta=\theta_k=\sup\{\theta\in [0,1]:q(\theta)\leq \ell_k\}$. We notice that the last constraint of $Q'$ implies that $H_k(\theta_k)\geq u_{2,m+1-k}-u_{2m}$ for $k\in \{2,3,...,m-1\}$. We verify that the inequality is also true for $k=1$ and $k=m$. When $k=1$, since $\ell_1=-u_{2m}$ and $q(\theta)\in [-u_{2m},u_{11}],\forall \theta\in [0,1]$, thus $q(\theta)=-u_{2m},\forall \theta\in [0,\theta_1]$. Thus $H_1(\theta_1)=0=u_{2m}-u_{2m}$. When $k=m$, since $\theta_m=1$ and $\ell_m=u_{11}$, $H_m(\theta_m)=u_{11}-u_{2m}-u_{11}=u_{21}-u_{2m}$.

Thus for every $k\in [m],\theta\in [0,1]$, we have
$$\int_0^\theta q(x)dx-\theta\cdot \ell_k\geq u_{2,m+1-k}-u_{2m}$$

Since $u(\theta)=\max_{k\in [m]}\{\theta\cdot \ell_k+u_{2,m+1-k}\}$, we have
$\int_0^\theta q(x)dx\geq u(\theta)-u_{2m}$. Thus $\Bq$ is feasible in program $Q$. We finish our proof.
\end{proof}

\mingfeinote{Mingfei: This is the end of the part moving to \Cref{sec:upper_bound_full_info}.} \yangnote{Yang: Let's remove Lemma 17 here.}
}

In this section, we are going to present our exact characterization of the optimal mechanism in $\MenuWithQ$, i.e., the optimal solution $\Bq^*=\{q^*(\theta)\}_{\theta\in [0,1]}$ of the program in \Cref{fig:continous-program-modified}. We Lagrangify Constraints {\IntegralConstr} and {\IRConstr} in the program. Denote $\eta$ and $\lambda=\{\lambda_k\}_{2\leq k\leq m-1}$ the Lagrangian multipliers for both sets of constraints accordingly. 

\begin{definition}\label{def:lagrangian}
Given Lagrangian multipliers $\eta$ and $\lambda=\{\lambda_k\}_{k\in\{2,\ldots, m-1\}}\geq \mathbf{0}$, define the virtual surplus function $J^{(\eta,\lambda)}:[-u_{2m},u_{11}]\times [0,1]\to \mathbb{R}$ as follows:
\begin{align*}
J^{(\eta,\lambda)}(q,\theta)=&(\theta f(\theta)+F(\theta))\cdot q+f(\theta)\cdot \min\{u_{11}-u_{2m}-q,0\}\\
&+\eta\cdot (u_{11}-u_{2m}-q)+\sum_{k=2}^{m-1}\lambda_k\cdot \left((q-\ell_k)\cdot \ind[q\leq \ell_k]-u_{2,m+1-k}+u_{2m}\right)    
\end{align*}
Then the Lagrangian $L(\Bq,\eta,\lambda)$ of the program in \Cref{fig:continous-program-modified} can be written as: $L(\Bq,\eta,\lambda)=\int_{0}^1J^{(\eta,\lambda)}(q(\theta),\theta)d\theta$.
\end{definition}


\begin{observation}\label{obs:program-convex}
For any fixed $\theta\in [0,1]$, all of the following functions are continuous, differentiable in $[-u_{2m},u_{11}]$ except at finitely many points, and weakly concave in $q$ (note that $g_2(q)$ is linear)
\begin{enumerate}
    \item $g_1(q,\theta)=(\theta f(\theta)+F(\theta))\cdot q+f(\theta)\cdot \min\{u_{11}-u_{2m}-q,0\}$.
    \item $g_2(q)=u_{11}-u_{2m}-q$. \item $g_3^{(k)}(q)=(q-\ell_k)\cdot \ind[q\leq \ell_k]-u_{2,m+1-k}+u_{2m},\forall k\in \{2,3,...,m-1\}$.
\end{enumerate} 
Thus for any Lagrangian multipliers $(\eta,\lambda\geq \mathbf{0})$, the function $J^{(\eta,\lambda)}(q,\theta)$ is continuous, differentiable in $q\in[-u_{2m},u_{11}]$ except at finitely many points, and weakly concave in $q$, for any fixed $\theta\in [0,1]$. Besides,
$$\frac{\partial g_1}{\partial q}(q,\theta)=
\begin{cases}
\theta f(\theta)+F(\theta), &\forall q\in (-u_{2m},u_{11}-u_{2m})\\
(\theta-1) f(\theta)+F(\theta), &\forall q\in (u_{11}-u_{2m},u_{11})
\end{cases}
$$
are both continuous in $\theta$ on $[0,1]$. Thus except at $q=-u_{2m}, u_{11}-u_{2m}, u_{11}$, $\frac{\partial J^{(\eta,\lambda)}}{\partial q}(q,\theta)$ exists and is continuous in $\theta$ on $[0,1]$.
\end{observation}

Denote $\cS=\{\Bq: [0,1]\rightarrow [-u_{2m},u_{11}] : q(\cdot) \text{ is non-decreasing}\}$. We notice that $\cS$ is convex, thus by \Cref{obs:program-convex}, the program in \Cref{fig:continous-program-modified} is a general convex programming problem. Strong duality holds according to Theorem 8.3.1 and Theorem 8.4.1 of \cite{luenberger1997optimization}.

\begin{lemma}\cite{luenberger1997optimization}\label{lem:strong_duality}
Let $(\eta^*,\lambda^*)\in\argmin_{\lambda\geq \mathbf{0},\eta}\max_{\Bq\in \cS}L(\Bq,\eta,\lambda)$ be the optimal Lagrangian multiplier. Then a solution $\Bq^*$ is optimal in the program in \Cref{fig:continous-program-modified} if and only if $\Bq^*\in \argmax_{\Bq\in \cS}L(\Bq,\eta^*,\lambda^*)$.
\end{lemma}

Now we are ready to state the characterization of the optimal solution.

\begin{definition}[Ironing~\cite{Myerson81}]\label{def:myerson_ironing}
Given any differentiable function $\varphi:[0,1]\to \mathbb{R}$. Let $F$ be any continuous distribution on $[0,1]$, with density $f(\cdot)$. Define the ironed function $\tilde{\varphi}(\cdot)$ as follows: For every $r\in [0,1]$, let $H(r):=\int_{0}^{r}\varphi(F^{-1}(x))dx$ and $G(\cdot)$ be the convex hull of $H(\cdot)$.~\footnote{$G(\cdot)$ is the convex hull of $H(\cdot)$ if $G$ is the highest convex function such that $G(\theta)\leq H(\theta),\forall \theta\in [0,1]$.} Let $g(\cdot)$ be the derivative of $G(\cdot)$,~\footnote{\label{footnote:extended_deriv}Since $H(\cdot)$ is differentiable, $G(\cdot)$, as the convex hull of $H(\cdot)$, is continuously differentiable on $(0,1)$. We extend $g$ to $[0,1]$ by continuity.} and $\tilde{\varphi}(\theta):=g(F(\theta)),\forall \theta\in [0,1]$. For any open interval $I\subseteq (0,1)$, we call $I$ an \textbf{ironed interval} of $\tilde{\varphi}(\cdot)$ if $G(F(\theta))<H(F(\theta)),\forall \theta\in I$. 
\end{definition}

\begin{definition}\label{def:ironed-virtual-value}
 Let $F(\cdot)$ be the agent's type distribution on $[0,1]$, with density $f(\cdot)$. For every $\theta\in [0,1]$, let $\varphi^-(\theta)=\theta f(\theta)+F(\theta)$ and $\varphi^+(\theta)=(\theta-1)f(\theta)+F(\theta)$ be the agent's two virtual values at type $\theta$.\footnote{$\varphi^-(\theta)$ is the virtual value when $q(\theta)\leq u_{11}-u_{2m}$, and $\varphi^+(\theta)$ is the virtual value when $q(\theta)>u_{11}-u_{2m}$.} Denote $\tilde{\varphi}^-(\cdot)$ and  $\tilde{\varphi}^+(\cdot)$ the ironed virtual value functions of $\varphi^-(\cdot)$ and  $\varphi^+(\cdot)$ respectively. {We say an open interval $I\subseteq (0,1)$ an \emph{ironed interval} of $\tilde{\varphi}^-$ (or $\tilde{\varphi}^+$) if $\int_0^{F(\theta)} \tilde{\varphi}^-\left(F^{-1}(x)\right) dx<\int_0^{F(\theta)} {\varphi}^-\left(F^{-1}(x)\right) dx$ (or $\int_0^{F(\theta)} \tilde{\varphi}^+\left(F^{-1}(x)\right) dx<\int_0^{F(\theta)} {\varphi}^+\left(F^{-1}(x)\right) dx$) for all $\theta\in I$.}
\end{definition}

\begin{theorem}\label{thm:characterization}

Suppose $\Bq^*=\{q^*(\theta)\}_{\theta\in [0,1]}$ is a feasible solution of the program in \Cref{fig:continous-program-modified}. Then $\Bq^*$ is optimal if and only if, there exist multipliers $\eta^*$ and $\lambda^* =\{\lambda^*_k\}_{k\in\{2,\ldots,m-1\}}\geq \mathbf{0}$, such that all of the following properties are satisfied:
\begin{enumerate}
    \item $q^*(\cdot)$ is non-decreasing, $$q^*(\theta)\in \argmax_{q\in [-u_{2m},u_{11}]}\left\{\tilde{\varphi}^-(\theta)\cdot q^- +\tilde{\varphi}^+(\theta)\cdot q^+ -\eta^*\cdot q+\sum_{k=2}^{m-1}\lambda_k^*(q-\ell_k)\cdot \ind[q\leq \ell_k]\right\},$$
    where $q^-=\min\{q,u_{11}-u_{2m}\}, q^+=\max\{q-u_{11}+u_{2m},0\}$, {almost everywhere}.
    \item $\Bq^*$ and $(\eta^*,\lambda^*)$ satisfy the complementary slackness, i.e., for every $k\in \{2,3,...,m-1\}$, $$\lambda_k^*\cdot \left[\int_0^{1}(q^*(x)-\ell_k)\cdot \ind[q^*(x)\leq \ell_k]dx- (u_{2,m+1-k}-u_{2m})\right]=0$$
    \item $\Bq^*$ satisfies the  \emph{generalized pooling property} (\Cref{def:generalized_pooling_property}):
    For any 
$x\in [-u_{2m},u_{11}]$, define $q^{*^{-1}}(x)=\inf\{\theta\in [0,1]:q^*(\theta)\geq x\}$. Then $\Bq^*$ satisfies the \emph{generalized pooling property} if 
    
    \begin{enumerate}
        \item for every 
       $x\in [-u_{2m},u_{11}-u_{2m}]$ and open interval $I\subseteq (0,1)$:
$$q^{*^{-1}}(x)\in I\text{ and $I$ is an ironed interval for $\tilde{\varphi}^-(\cdot)$}\implies q^*(\cdot) \text{ is constant on } I.$$
     \item for every 
   $x\in (u_{11}-u_{2m},u_{11}]$
     and open interval $I\subseteq (0,1)$:
$$q^{*^{-1}}(x)\in I\text{ and $I$ is an ironed interval for $\tilde{\varphi}^+(\cdot)$}\implies q^*(\cdot) \text{ is constant on } I.$$
    \end{enumerate}
\end{enumerate}
\end{theorem}

\Cref{thm:characterization} follows from the generalized ironing procedure by Toikka~\cite{toikka2011ironing}. We first introduce the necessary background.
\paragraph{General Virtual Surplus Function:} Given any function $J(\cdot,\cdot):[a,b]\times [0,1]\to \mathbb{R}$. We assume that $J$ satisfies: (i) For every fixed $q\in [a,b]$, $J(q,\cdot)$ is continuously differentiable on $(0,1)$, 
and (ii) For every fixed $\theta\in [0,1]$, $J(\cdot,\theta)$ is continuous on $[a,b]$ and weakly concave. (iii) Except at finitely many points $q\in [a,b]$, $\frac{\partial J}{\partial q}(q,\theta)$ exists and is continuous in $\theta$ on $[0,1]$.
We define the ironed virtual surplus in \Cref{def:ironed_virtual_surplus}.

\begin{definition}[Generalized Ironing Procedure~\cite{toikka2011ironing}]
\label{def:ironed_virtual_surplus}
Given any virtual surplus function $J$, for every $r\in [0,1]$, let $h(q,r):=\frac{\partial J}{\partial q}\left(q,F^{-1}(r)\right)$. By assumption (iii), it is well-defined on $[a,b]$ except at finitely many points. We extend $h$ to the whole interval $[a,b]$ by {left-continuity}. For every $q\in [a,b],r\in [0,1]$, let $H(q,r):=\int_{0}^{r}h(q,x)dx$. \footnote{By assumption (iii) of $J(\cdot,\cdot)$ and the continuity of $F^{-1}(\cdot)$, $h(q,x)$ continuous in $x$ and hence is integrable in $x$.} For every $q\in [a,b]$, define $G(q,\cdot):=\conv H(q,\cdot)$ as the convex hull of function $H(q,\cdot)$. Let $g(q,r):=\frac{\partial G}{\partial r}(q,r)$, which is non-decreasing in $r$ for every $q\in[a,b]$.~\footnote{
{Since $H(q,\cdot)$ is differentiable on $(0,1)$ for every $q\in [a,b]$, $G(q,\cdot)$, as the convex hull of $H(q,\cdot)$, is continuously differentiable on $(0,1)$ for every $q\in[a,b]$. We extend $g(q,\cdot)$ to $[0,1]$ by continuity.}
} Define the \emph{ironed virtual surplus} function $\tilde{J}:[a,b]\times [0,1]\to \mathbb{R}$ as follows:
$$\tilde{J}(q,\theta)=J(0,\theta)+\int_0^qg(s,F(\theta))ds$$
\end{definition}

\begin{definition}\label{def:generalized_pooling_property}\cite{toikka2011ironing}
Given $\Bq=\{q(\theta)\}_{\theta\in [0,1]}$. For any $x\in [a,b]$, define $q^{-1}(x)=\inf\{\theta\in [0,1]:q(\theta)\geq x\}$. Then $\Bq$ satisfies the \emph{generalized pooling property} if for every $x\in [a,b]$ and open interval $I\subset [0,1]$:
$$q^{-1}(x)\in I \text{ and }G(x,\theta)<H(x,\theta),\forall \theta\in I\implies q(\cdot) \text{ is constant on } I.$$
\end{definition}

\begin{theorem}\cite{toikka2011ironing}\label{thm:toikka-ironing}
Let $\cS=\{\Bq: [0,1]\rightarrow [a,b] :q(\cdot) \text{ is non-decreasing}\}$ and $\cS'=\{\Bq: [0,1]\rightarrow [a,b]\}$. Then $$\sup_{\Bq\in \cS}\left\{\int_0^1 J(q(\theta),\theta)d\theta\right\}=\sup_{\Bq\in \cS'}\left\{\int_0^1 \tilde{J}(q(\theta),\theta)d\theta\right\}.$$
Moreover, $\Bq\in \cS$ achieves the supremum of $\int_0^1 J(q(\theta),\theta)d\theta$ if and only if: $\tilde{J}(q(\theta),\theta)=\sup_{q\in [a,b]}\tilde{J}(q,\theta)$ {almost everywhere}, and $\Bq$ satisfies the generalized pooling property.
\end{theorem}

\begin{prevproof}{Theorem}{thm:characterization}
By \Cref{obs:program-convex}, for any collection of Lagrangian 
multipliers $\eta,\lambda\geq \mathbf{0}$, the virtual surplus function $J^{(\eta,\lambda)}(\cdot,\cdot)$ satisfies our assumptions. We apply the ironing procedure in \Cref{def:ironed_virtual_surplus} to $J^{(\eta,\lambda)}(\cdot,\cdot)$:

$$h(q,r)=\frac{\partial J^{(\eta,\lambda)}}{\partial q}\left(q,F^{-1}(r)\right)=
\begin{cases}
\varphi^{-}(F^{-1}(r))-\eta+\sum_{j=k(q)+1}^{m-1} \lambda_j, & q\leq u_{11}-u_{2m}\\
\varphi^{+}(F^{-1}(r))-\eta+\sum_{j=k(q)+1}^{m-1} \lambda_j, & q>u_{11}-u_{2m}\\
\end{cases}
$$
Recall that $\ell_k=u_{1,m+1-k}-u_{2,m+1-k},\forall k\in [m]$. Here $k(q)$ is the unique number $j\in \{1,...,m-1\}$ such that $\ell_{j}<q\leq \ell_{j+1}$ (let $k(q)=1$ when $q=\ell_1$). 
For every $q\in [-u_{2m},u_{11}]$ and $r\in [0,1]$, 
$$H(q,r)=\int_{0}^{r}h(q,x)dx=
\begin{cases}
\int_0^r\varphi^{-}(F^{-1}(x))dx+r\cdot \left(\sum_{j=k(q)+1}^{m-1} \lambda_j-\eta\right), & q\leq u_{11}-u_{2m}\\
\int_0^r\varphi^{+}(F^{-1}(x))dx+r\cdot \left(\sum_{j=k(q)+1}^{m-1} \lambda_j-\eta\right), & q>u_{11}-u_{2m}\\
\end{cases}
$$

For every fixed $q\in [-u_{2m},u_{11}]$, to obtain $G(q,\cdot)$, we take the convex hull of $H(q,\cdot)$. Note that $r\cdot \left(\sum_{j=k(q)+1}^{m-1} \lambda_j-\eta\right)$ is a linear function on $r$, thus we are effectively only taking the convex hull of the functions $\widehat{H}^-(r):=\int_0^r\varphi^{-}(F^{-1}(x))dx$ and $\widehat{H}^+(r):=\int_0^r\varphi^{+}(F^{-1}(x))dx$. Denote $\widehat{G}^-(\cdot)$ and $\widehat{G}^+(\cdot)$ the convex hull of $\widehat{H}^-(\cdot)$ and $\widehat{H}^+(\cdot)$ accordingly. Then

$$G(q,r)=
\begin{cases}
\widehat{G}^-(r)+r\cdot \left(\sum_{j=k(q)+1}^{m-1} \lambda_j-\eta\right), & q\leq u_{11}-u_{2m}\\
\widehat{G}^+(r)+r\cdot \left(\sum_{j=k(q)+1}^{m-1} \lambda_j-\eta\right), & q>u_{11}-u_{2m}\\
\end{cases}
$$


By definition of $\tilde{\varphi}^-$, $\tilde{\varphi}^+$ (\Cref{def:ironed-virtual-value}), we have  

$$g(q,r)=
\begin{cases}
\tilde{\varphi}^{-}(F^{-1}(r))+\sum_{j=k(q)+1}^{m-1} \lambda_j-\eta, & q\leq u_{11}-u_{2m}\\
\tilde{\varphi}^{+}(F^{-1}(r))+\sum_{j=k(q)+1}^{m-1} \lambda_j-\eta, & q>u_{11}-u_{2m}\\
\end{cases}
$$

The ironed virtual surplus $\tilde{J}^{(\eta,\lambda)}(q,\theta)= J^{(\eta,\lambda)}(0,\theta)+\int_0^qg(s,F(\theta))ds$. We notice that $J^{(\eta,\lambda)}(0,\theta)=-\sum_{k=2}^{m-1}\lambda_k\ell_k$, and by definition of $k(q)$,
$$\int_{0}^q\sum_{j=k(s)+1}^{m-1}\lambda_jds=\sum_{j=2}^{m-1}\lambda_j\cdot \int_{0}^{q}\ind[s\leq \ell_j]ds=\sum_{j=2}^{m-1}\lambda_j\cdot \min\{\ell_j,q\}$$
Thus we have
\begin{align*}
&\tilde{J}^{(\eta,\lambda)}(q,\theta)= J^{(\eta,\lambda)}(0,\theta)+\int_0^q g(s,F(\theta))ds\\
=&\begin{cases}
\tilde{\varphi}^{-}(\theta)\cdot q-\eta\cdot q+\sum_{k=2}^{m-1}\lambda_k(q-\ell_k)\cdot \ind[q\leq \ell_k], & q\leq u_{11}-u_{2m}\\
\tilde{\varphi}^{-}(\theta)\cdot (u_{11}-u_{2m})+\tilde{\varphi}^{+}(\theta)(q-u_{11}+u_{2m})-\eta\cdot q+\sum_{k=2}^{m-1}\lambda_k(q-\ell_k)\cdot \ind[q\leq \ell_k], & q>u_{11}-u_{2m}
\end{cases}\\
=&\tilde{\varphi}^-(\theta)\cdot q^- +\tilde{\varphi}^+(\theta)\cdot q^+ -\eta\cdot q+\sum_{k=2}^{m-1}\lambda_k(q-\ell_k)\cdot \ind[q\leq \ell_k],
\end{align*}
where $q^-=\min\{q,u_{11}-u_{2m}\}, q^+=\max\{q-u_{11}+u_{2m},0\}$.

We first prove that the properties are necessary for $\Bq^*$ to be optimal. Suppose $\Bq^*$ is an optimal solution. Let $(\eta^*,\lambda^*)$ be the optimal Lagrangian multipliers. By \Cref{lem:strong_duality}, $\Bq^*$ maximizes $L(\Bq,\eta^*,\lambda^*)=\int_0^1J^{(\lambda^*,\eta^*)}(q(\theta),\theta)d\theta$ over $\Bq\in \cS$. By \Cref{thm:toikka-ironing}, $\tilde{J}^{(\eta^*,\lambda^*)}(q^*(\theta),\theta)=\sup_{q\in [-u_{2m},u_{11}]}\tilde{J}^{(\eta^*,\lambda^*)}(q,\theta)$ almost everywhere, which is exactly the first property in the statement of~\Cref{thm:characterization}, and $\Bq^*$ satisfies the generalized pooling property.  
Since $\Bq^*$ is the optimal solution and $(\eta^*,\lambda^*)$ are the optimal Lagrangian multipliers, the second property in the statement of \Cref{thm:characterization} directly follows from the KKT condition. Finally, we prove the last property. {An important property of $H(\cdot,\cdot)$ and $G(\cdot,\cdot)$ is that for any $q\leq u_{11}-u_{2m}$, $H(q,r)>G(q,r)$ iff $\widehat{H}^-(r)>\widehat{G}^-(r)$, and for any $q>u_{11}-u_{2m}$, iff $\widehat{H}^+(r)>\widehat{G}^+(r)$. Therefore, $\Bq^*$ satisfies the generalized pooling property is equivalent to the last property in the statement.}

Now we prove sufficiency. For any feasible solution $\Bq^*$, suppose there exist multipliers $(\eta^*,\lambda^*\geq \textbf{0})$ that satisfy all three properties in the statement. We argue that $\Bq^*$ is the optimal solution (and $(\eta^*,\lambda^*)$ are the optimal Lagrangian multipliers) by verifying the KKT conditions. Firstly, both $\Bq^*$ and $(\eta^*,\lambda^*)$ are feasible primal and dual solutions. Secondly, by the first property of the statement, $\tilde{J}^{(\eta^*,\lambda^*)}(q^*(\theta),\theta)=\sup_{q\in [-u_{2m},u_{11}]}\tilde{J}^{(\eta^*,\lambda^*)}(q,\theta)$ almost everywhere. As argued above, the fact that $\Bq^*$ satisfies the third property implies that $\Bq^*$ satisfies the generalized pooling property. Thus by \Cref{thm:toikka-ironing}, $\Bq^*$ maximizes $L(\Bq,\eta^*,\lambda^*)=\int_0^1J^{(\lambda^*,\eta^*)}(q(\theta),\theta)d\theta$ over $\Bq\in \cS$. The stationary condition is satisfied. Thirdly, the second property of the statement implies that the complementary slackness is satisfied. Thus by the KKT conditions, $\Bq^*$ is the optimal solution.  
\end{prevproof}

\subsection{Optimality Conditions for Selling Complete Information}\label{subsec:characterization-full-info}

In this section, we characterize the conditions under which selling only the fully informative experiment is optimal using \Cref{thm:characterization}. Clearly, any {\responsiveIC} and IR mechanism that only contains the fully informative experiment $E^*$ is determined by the price $p>0$ of $E^*$. And for every type $\theta\in [0,1]$, the agent selects $E^*$ if and only if $V_{\theta}(E^*)\geq u(\theta)$. Since $V_{\theta}(E^*)=\theta u_{11}+ (1-\theta)u_{2m}-p$ 
is linear in $\theta$ and the IR curve is convex in $\theta$, the agent selects $E^*$ when her type $\theta$ is in some closed interval $[\theta_p^{L},\theta_p^{H}]$ (see \Cref{fig:full_info_1} or \Cref{fig:full_info_2} for an illustration of $\theta_p^{L},\theta_p^{H}$).
\footnote{We assume that $\theta_p^{L}$ and $\theta_p^{H}$ exist, otherwise the price is too high, and the mechanism has revenue $0$.} Clearly, the mechanism is also IC, since the agent will always follow the recommendation when receiving either no information or full information.   

We first prove the following corollary of \Cref{lem:ic in q space}. It states that any $k$-piecewise-linear ``utility curve'', which always stays above the IR curve, can be implemented by a {\responsiveIC} and IR mechanism with $k$ experiments. 

\begin{corollary}\label{cor:convex-utility}

Given any finite integer $k\geq 2$ and any continuous, piecewise linear function $h:[0,1]\to \mathbb{R}_+$ with $k$ pieces. Suppose $h(\cdot)$ satisfies all of the following: (i) $h(0)=u_{2m},h(1)=u_{11}$; (ii) $-u_{2m}\leq\ell_1< ...<\ell_k\leq u_{11}$, where for each $i\in [k]$, $\ell_i$ is the slope of the $i$-th piece of $h$; (iii) $h(\theta)\geq u(\theta),\forall \theta\in [0,1].$ 
Then there exists a {\responsiveIC} and IR mechanism $\MM\in \MenuWithQ$ with $k$ experiments, such that for every $\theta\in [0,1]$, the agent's utility at type $\theta$, when she follows the recommendation, is $h(\theta)$. 
\end{corollary}

\begin{proof}

We define the mechanism $\MM=(\Bq,\Bt)$ as follows. Let $0<\theta_1<\theta_2<...<\theta_{k-1}<1$ be the $k-1$ kinks of $h$. Denote $\theta_0=0,\theta_k=1$. For every $\theta\in [0,1)$, define $q(\theta)=\ell_i$, where $i\in [k]$ is the unique number such that $\theta\in [\theta_{i-1},\theta_i)$. Also define $q(1)=\ell_k$.

We will verify that $\Bq$ satisfies all the requirements of \Cref{lem:ic in q space}. By property (ii), we have that $q(\theta)\in [-u_{2m},u_{11}]$ is non-decreasing in $\theta$. Moreover, we notice that $q(\theta)$ is the right derivative of $h(\theta)$ at every $\theta\in [0,1)$. Since $h(0)=u_{2m}$ and $h$ is continuous at $\theta=1$, we have $h(\theta)=u_{2m}+\int_{0}^\theta q(x)dx,\forall \theta\in [0,1]$. Taking $\theta=1$, we have $\int_{0}^1q(x)dx=h(1)-u_{2m}=u_{11}-u_{2m}$. At last, by property (iii) we have $h(\theta)=u_{2m}+\int_{0}^\theta q(x)dx\geq u(\theta),\forall \theta\in [0,1]$.   

Thus by \Cref{lem:ic in q space}, there exists a payment rule $\Bt$ such that $\MM$ is a {\responsiveIC} and IR mechanism, and $\Bt$ satisfies \Cref{equ:payment}. By \Cref{obs:value-function}, with such payment, the agent's utility at every type $\theta$, when he follows the recommendation, is $u_{2m}+\int_{0}^\theta q(x)dx=h(\theta)$. Since $q(\theta)$ takes $k$ different values, the mechanism contains $k$ experiments. 
\end{proof}

To obtain a simplified characterization, an important step is the following observation: If some fully informative menu with price $p$ is the optimal menu, then both $\theta_p^L$ and $\theta_p^H$ have to stay in the first and last piece of the IR curve accordingly (see \Cref{fig:full_info_1}). To see the reason, consider a curve $h(\cdot)$ (a piecewise linear function) that is the maximum of the utility function of buying the fully informative experiment at price $p$ and the IR curve. In other words, $h(\cdot)$ coincides with the utility function of buying the fully informative experiment on interval $[\theta_p^L,\theta_p^H]$, and coincides with the IR curve everywhere else. If $\theta_p^L$ and $\theta_p^H$ do not lie in the first and last piece of the IR curve respectively, there are pieces of the curve $h(\cdot)$ that belong to the $i$-th piece of the IR curve for some $2\leq i\leq m-1$ (the blue line in \Cref{fig:full_info_2}). We argue that we can change the mechanism to offer another experiment (based on this piece of the IR curve), so that (i) the mechanism is still \responsiveIC~and IR, and (ii) the revenue of the mechanism strictly increases.
This contradicts with the optimality of selling only the fully informative experiment.
See \Cref{fig:full_info_2} for an illustration and \Cref{lem:endpoints of the full info} for a formal statement.

\begin{lemma}\label{lem:endpoints of the full info}
Consider the fully informative experiment $E^*$ with any price $p>0$. Assume there is some $\theta$ such that $V_{\theta}(E^*)>u(\theta)$. Let $\theta_p^L,\theta_p^H$ ($0<\theta_p^L<\theta_p^H<1$) be the two $\theta$'s such that $V_{\theta}(E^*)=u(\theta)$. Let $\theta_1$ and $\theta_2$ be the first and last kink of the IR curve.\footnote{\label{footnote:kink-formal}Formally, $\theta_1$ satisfies $u_{2m}(1-\theta_1)=u_{1,m-1}\theta_1+(1-\theta_1)u_{2,m-1}$, and $\theta_2$ satisfies $u_{11}\theta_2=u_{12}\theta_2+(1-\theta_2)u_{22}$.} Suppose selling the full information at price $p$ achieves the optimal revenue among all {\responsiveIC} and IR mechanisms. Then $\theta_p^L\leq \theta_1$ and $\theta_p^H\geq \theta_2$ (see \Cref{fig:full_info_1}). 
\end{lemma}

\begin{proof}
Let $\MM$ be the menu that only contains the fully informative experiment $E^*$ with price $p$. We prove by contradiction. Without loss of generality, assume that $\theta_p^L>\theta_1$. Let $L:[0,1]\to \mathbb{R}$ be the agent's utility function for the experiment $E^*$, i.e., $L(\theta)=V_{\theta}(E^*)-p=\theta\cdot u_{11}+(1-\theta)u_{2m}-p$. Consider the function $h:[0,1]\to \mathbb{R}$ such that $h(\theta)=\max\{u(\theta),L(\theta)\},\forall \theta\in[0,1]$. One can easily verify that $h$ is a piecewise linear function that satisfies the requirement of \Cref{cor:convex-utility}. By \Cref{cor:convex-utility}, there exists a {\responsiveIC} and IR mechanism $\MM'=(\Bq',\Bt')\in \MenuWithQ$, such that for every $\theta\in [0,1]$, the agent's utility at type $\theta$, when she follows the recommendation, is $h(\theta)$. When the agent's type $\theta\in[\theta_p^L,\theta_p^H]$, she purchases experiment $E^*$ at price $p$ in $\MM'$, and thus the mechanism collects total revenue $\rev(\MM)$ at interval $[\theta_p^L,\theta_p^H]$. It suffices to show that $\MM'$ collects strictly positive revenue from types in $[0,\theta_p^L)\cup (\theta_p^H,1]$, which contradicts with the optimality of $\MM$. 

Since $\theta_1<\theta_p^L$, there is at least one (linear) piece of the piecewise linear function $h(\cdot)$ is in $[\theta_1,\theta_p^L]$. Consider any such piece. The whole piece is also on the IR curve. Let $(\theta',u(\theta'))$ and $(\theta'',u(\theta''))$ be the two endpoints of this piece ($\theta'<\theta''$) and $\ell$ be the slope. Then $-u_{2m}<\ell<u_{11}-u_{2m}$. We design an experiment $E$ offered by $\MM'$, so that all types  on interval $(\theta',\theta'')$ purchase this experiment. Denote $q$ the variable corresponds to the experiment $E$. Then by \Cref{obs:value-function}, $q=\ell\leq u_{11}-u_{2m}$. Thus when $\theta\in (\theta',\theta'')$, $V_{\theta}(E)=q\cdot \theta+u_{2m}$. Moreover, notice that (i) $u(\theta')\geq u_{2m}-u_{2m}\theta'$, the RHS is agent's payoff for always choosing action $m$ without receiving any experiment, and (ii) $h(\theta)=\ell\cdot (\theta-\theta')+u(\theta'),\forall \theta\in [\theta',\theta'']$. Hence
$$t'(\theta)=V_{\theta}(E)-h(\theta)=u_{2m}+\ell\cdot \theta'-u(\theta')=(u_{2m}+\ell)\theta'>0,\forall \theta\in (\theta',\theta'')$$

Therefore, $\rev(\MM')\geq \rev(\MM)+{\left(F(\theta'')-F(\theta')\right)}\cdot (u_{2m}+\ell)\theta'>\rev(\MM)$, as $F(\cdot)$ is strictly increasing, contradicting with the optimality of $\MM$. Hence, we must have $\theta_p^L\leq \theta_1$ and $\theta_p^H\geq \theta_2$. 
\end{proof}

\begin{figure}[ht]
	\centering
	\begin{subfigure}[b]{0.32\textwidth}
	    \includegraphics[width=\textwidth]{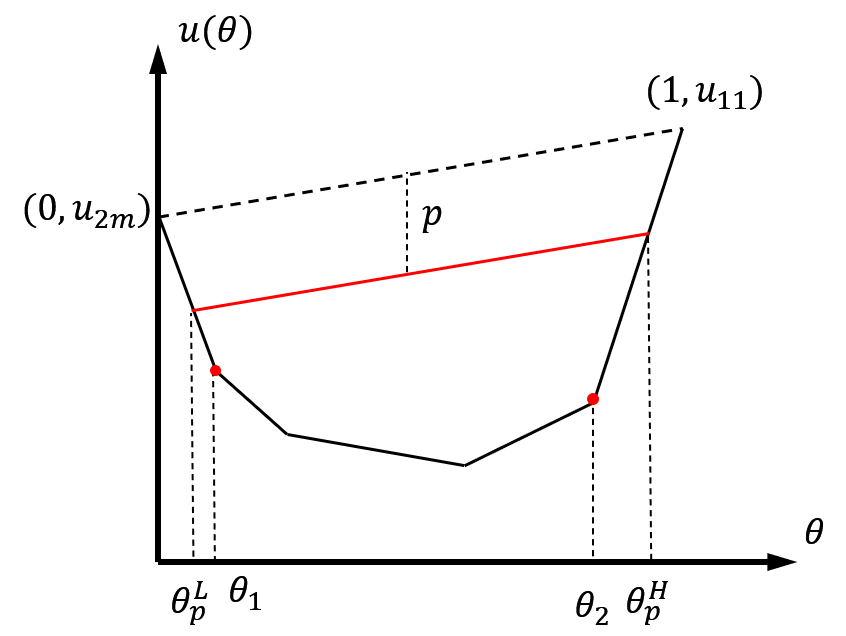}
	    \caption{The fully informative experiment where $\theta_p^L\leq \theta_1$ and $\theta_p^H\geq \theta_2$}
	    \label{fig:full_info_1}
	\end{subfigure}
	\hfill
	\begin{subfigure}[b]{0.32\textwidth}
	    \includegraphics[width=\textwidth]{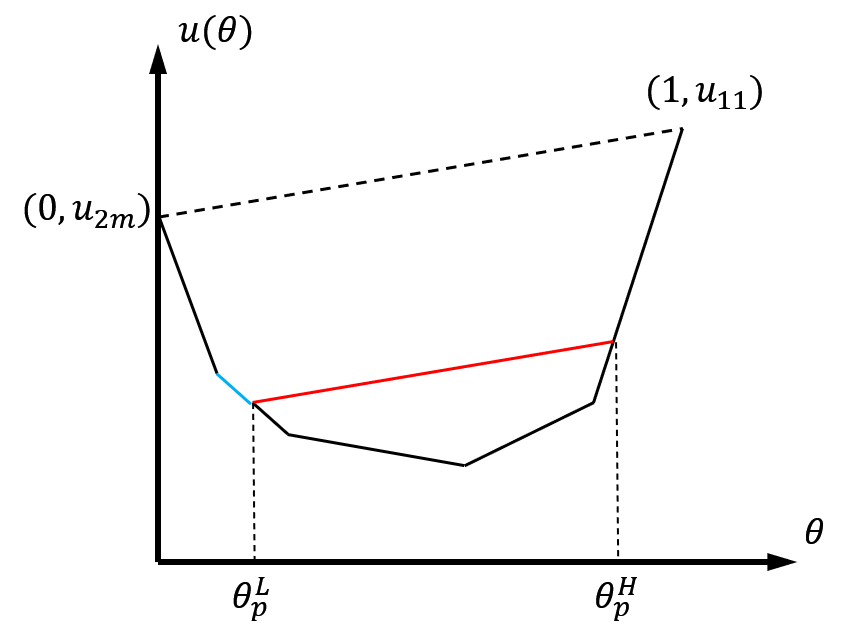}
	    \caption{The scenario when the fully informative experiment is not optimal}
	    \label{fig:full_info_2}
	\end{subfigure}
	\hfill
	\begin{subfigure}[b]{0.32\textwidth}
	    \includegraphics[width=\textwidth]{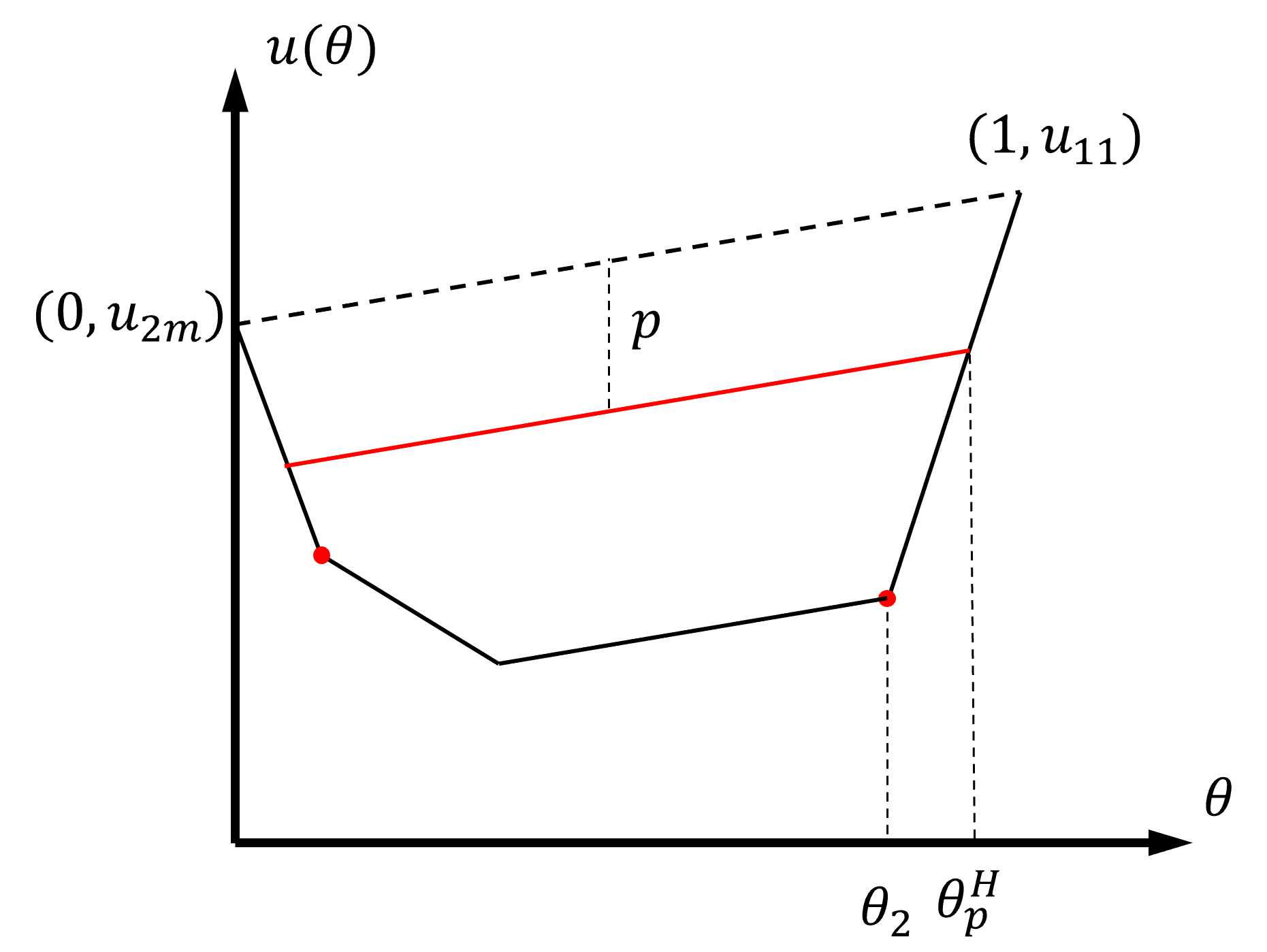}
	    \caption{In the proof of \Cref{thm:characterization-full-info}: $\theta_p^H>\theta_2$ when $\ell_{m-1}=u_{11}-u_{2m}$}
	    \label{fig:full_info_3}
	\end{subfigure}
\caption{Illustrations of notations and the utility curve of full information}\label{fig:full_info}
\end{figure}


\notshow{
Equipped with \Cref{lem:endpoints of the full info}, it is sufficient to consider the following restrictive IR curve $u'(\cdot)$: Let $L_1:[0,1]\to \mathbb{R}$ be the affine function that goes through the first kink $(\theta_1,u_{2m}-u_{2m}\theta_1)$ of the original IR curve, and is parallel to the line that connects $(0,u_{2m})$ and $(1,u_{11})$. It has slope $u_{11}-u_{2m}$. Similarly, let $L_2$ be the affline function that goes through the last kink $(\theta_2,u_{11}\theta_2)$ and has slope $u_{11}-u_{2m}$. Then for every $\theta\in [0,1]$, define $u'(\theta)=\max\{u(\theta),L_1(\theta),L_2(\theta)\}$. \yangnote{Yang: Shouldn't $u'(\cdot)=$ either $L_1(\cdot)$ or $L_2(\cdot)$?} See \Cref{fig:full_info_3} for an illustration.  

Since the utility function for any fully informative experiment is linear and has slope $u_{11}-u_{2m}$, thus a fully informative experiment is optimal w.r.t the IR curve $u(\cdot)$ if and only if it's optimal w.r.t. the restrictive IR curve $u'(\cdot)$. The following observation follows from a simple calculation of the payoff matrix $\cU'$ (with 2 states and 3 actions) that implements $u'(\cdot)$ as the IR curve.   

\begin{observation}\label{obs:relaxed_IR_curve}
The restrictive IR curve $u'(\cdot)$ corresponds to an environment with 2 states and 3 actions with the payoff matrix:
$$\cU'=
\begin{pmatrix}
u_{11}' & u_{12}' & u_{13}'\\
u_{21}' & u_{22}' & u_{23}'
\end{pmatrix}
=\begin{pmatrix}
u_{11} & u^*+(u_{11}-u_{2m}) & 0\\
0 & u^* & u_{2m}
\end{pmatrix}
$$
where 
$$u^*=\max\left\{\frac{u_{22}u_{2m}}{u_{11}-u_{12}+u_{22}},u_{2m}-\frac{(u_{2m}-u_{2,m-1})u_{11}}{u_{2m}+u_{1,m-1}-u_{2,m-1}}\right\}.$$ Moreover, a fully informative experiment with price $p$ is optimal w.r.t. the original payoff matrix $\cU$ if and only if it's optimal w.r.t. the payoff matrix $\cU'$.  
\end{observation}

}


To prove \Cref{thm:characterization-full-info}, we also need the following lemma for the ironed virtual value, which may be of independent interest.

\begin{lemma}\label{lem:compare_ironed_myerson}
(Adapted from Lemma 4.11 in \cite{toikka2011ironing}) Let $\varphi_1$, $\varphi_2$ be any pair of differentiable real-valued functions on $[0,1]$. Let $F$ be any continuous distribution on $[0,1]$. Denote $\tilde{\varphi}_1$, $\tilde{\varphi}_2$ the ironed functions of $\varphi_1$, $\varphi_2$ respectively, with respect to the distribution $F$ (\Cref{def:myerson_ironing}). Suppose $\varphi_1(\theta)\geq\varphi_2(\theta),\forall\theta\in [0,1]$. Then $\tilde{\varphi}_1(\theta)\geq\tilde{\varphi_2}(\theta),\forall\theta\in [0,1]$.  
\end{lemma}

\begin{proof}
Let $f$ be the pdf of $F$. As shown in \Cref{def:myerson_ironing}, define $H_1(r)=\int_{0}^r\varphi_1(F^{-1}(x))dx,\forall r\in [0,1]$, $G_1(\cdot)$ as the convex hull of $H_1(\cdot)$, and $g_1(\cdot)$ as the (extended) derivative of $G_1(\cdot)$ (see \Cref{footnote:extended_deriv}). Define $H_2,G_2,g_2$ similarly for $\varphi_2$. 

For every $\theta\in [0,1]$, $\tilde{\varphi}_1(\theta)=g_1(F(\theta)), \tilde{\varphi}_2(\theta)=g_2(F(\theta))$. Thus to show $\tilde{\varphi}_1(\theta)\geq\tilde{\varphi}_2(\theta)$, it suffices to prove that $g_1(r)\geq g_2(r),\forall r\in [0,1]$. We prove the claim by contradiction.  Suppose there exists $r_0\in [0,1]$ such that $g_1(r_0)< g_2(r_0)$. We notice that
$$\cH(r)\triangleq H_1(r)-H_2(r)=\int_{0}^r \left[\varphi_1(F^{-1}(x))-\varphi_2(F^{-1}(x))\right]dx,$$
which is non-decreasing in $r$, since $\varphi_1(F^{-1}(x))\geq\varphi_2(F^{-1}(x)),\forall x\in [0,1]$. Let $\rho=G_1(r_0)-G_2(r_0)$. To reach a contradiction, we show that there exists $r_1\leq r_0$ such that $\cH(r_1)\geq \rho$ and there exists $r_2>r_0$ such that $\cH(r_2)<\rho$, which contradicts with the fact that $\cH$ is non-decreasing.

Let $L_1:[0,1]\to \mathbb{R}$ (or $L_2$) be the unique linear function tangent
to $G_1$ (or $G_2$) at $r_0$. Then since $G_1$ (or $G_2$) is the convex hull of $H_1$ (or $H_2$), we have for every $r\in [0,1]$, $L_1(r)\leq G_1(r)\leq H_1(r)$ and $L_2(r)\leq G_2(r)\leq H_2(r)$. 

\vspace{-.1in}
\paragraph{The existence of $r_1$:} If $G_2(r_0)=H_2(r_0)$, then $H_1(r_0)-H_2(r_0)\geq G_1(r_0)-G_2(r_0)=\rho$. By choosing $r_1=r_0$ we have $\cH(r_1)\geq \rho$. Now assume $G_2(r_0)<H_2(r_0)$. Then $r_0$ is in the interior of some ironed interval $I$ of $\varphi_2$. Let $r_1<r_0$ be the left endpoint of the interval $I$. Then $H_2(r_1)=L_2(r_1)$. 
We have $$\cH(r_1)=H_1(r_1)-H_2(r_1)\geq L_1(r_1)-L_2(r_1)> L_1(r_0)-L_2(r_0)=G_1(r_0)-G_2(r_0)=\rho$$
Here the first inequality follows from $H_1(r_1)\geq L_1(r_1)$ and $H_2(r_1)=L_2(r_1)$. The second inequality follows from the fact that $d(L_1(r)-L_2(r))/dr=G_1'(r_0)-G_2'(r_0)=g_1(r_0)-g_2(r_0)<0$. The second-last equality follows from the fact that $L_1$ (or $L_2$) meets $G_1$ (or $G_2$) at $r_0$. 
\vspace{-.1in}
\paragraph{The existence of $r_2$:} The proof follows from a similar argument as $r_1$. If $G_1(r_0)=H_1(r_0)$, then $G_1'(r_0)=H_1'(r_0)=g_1(r_0)$. We notice that 

$$\left.\frac{d(H_1(r)-L_2(r))}{dr}\right\vert_{r=r_0}=g_1(r_0)-g_2(r_0)<0.$$

Thus there exists some $r_2>r_0$ such that $H_1'(r)-g_2(r_0)<0$ for all $r\in [r_0,r_2]$ due to the continuity of $H_1'(\cdot)$, which implies that
$$\rho=G_1(r_0)-G_2(r_0)=H_1(r_0)-L_2(r_0)>H_1(r_2)-L_2(r_2)\geq H_1(r_2)-H_2(r_2)=\cH(r_2)$$

Now assume $G_1(r_0)<H_1(r_0)$. Then $r_0$ is in the interior of some ironed interval $I$ of $\varphi_1$. Let $r_2>r_0$ be the right endpoint of the interval $I$. Then $H_1(r_2)=L_1(r_2)$. We have $$\cH(r_2)=H_1(r_2)-H_2(r_2)\leq L_1(r_2)-L_2(r_2)<L_1(r_0)-L_2(r_0)=G_1(r_0)-G_2(r_0)=\rho$$
Here the first inequality follows from $H_1(r_2)=L_1(r_2)$ and $H_2(r_2)\geq L_2(r_2)$. The second inequality follows from the fact that $d(L_1(r)-L_2(r))/dr<0$ and $r_2>r_0$. The second-last equality follows from the fact that $L_1$ (or $L_2$) meets $G_1$ (or $G_2$) at $r_0$. 

The existence of $r_1$ and $r_2$ contradicts with the fact that $\cH$ is non-decreasing. Thus $\tilde{\varphi}_1(\theta)\geq\tilde{\varphi}_2(\theta),\forall \theta\in[0,1]$.
\end{proof}


\notshow{
\begin{lemma}\label{lem:compare_ironed_myerson}
(Adapted from Lemma 4.11 in \cite{toikka2011ironing}) Let $\varphi_1$, $\varphi_2$ be any pair of differentiable, real-valued functions on $[0,1]$. Let $D$ be any continuous distribution on $[0,1]$. Denote $\tilde{\varphi}_1$, $\tilde{\varphi}_2$ the ironed functions of $\varphi_1$, $\varphi_2$ respectively, with respect to the distribution $D$ (\Cref{def:myerson_ironing}). Suppose $\varphi_1(\theta)\geq\varphi_2(\theta),\forall\theta\in [0,1]$. Then $\tilde{\varphi}_1(\theta)\geq\tilde{\varphi_2}(\theta),\forall\theta\in [0,1]$.  
\end{lemma}

\begin{proof}
Let $f$ and $F$ be the pdf and cdf of $D$ respectively. As shown in \Cref{def:myerson_ironing}, define $H_1(r)=\int_{0}^r\varphi_1(F^{-1}(x))dx,\forall r\in [0,1]$, $G_1(\cdot)$ as the convex hull of $H_1(\cdot)$, and $g_1(\cdot)$ as the (extended) derivative of $G_1(\cdot)$ (see \Cref{footnote:extended_deriv}). Define $H_2,G_2,g_2$ similarly for $\varphi_2$. 

For every $\theta\in [0,1]$, $\tilde{\varphi}_1(\theta)=g_1(F(\theta)), \tilde{\varphi}_2(\theta)=g_2(F(\theta))$. Thus to show $\tilde{\varphi}_1(\theta)\geq \tilde{\varphi}_2(\theta)$, it suffice to prove that $g_1(r)\geq g_2(r),\forall r\in [0,1]$. Suppose, for the sake of contradiction, that there exists $r_0\in [0,1]$ such that $g_1(r_0)<g_2(r_0)$. We notice that
$$\cH(r)\triangleq H_1(r)-H_2(r)=\int_{0}^r \left[\varphi_1(F^{-1}(x))-\varphi_2(F^{-1}(x))\right]dx,$$
which is non-decreasing in $r$, since $\varphi_1(F^{-1}(x))\geq \varphi_2(F^{-1}(x)),\forall x\in [0,1]$. Let $\rho=G_1(r_0)-G_2(r_0)$. To reach a contradiction, we will show that there exists $r_1\leq r_0$ such that $\cH(r_1)\geq \rho$ and there exists $r_2>r_0$ such that $\cH(r_2)<\rho$. This will contradict with the fact that $\cH$ is non-decreasing.

Without loss of generality, assume that both $G_1'(r_0)$ and $G_2'(r_0)$ exist. Let $L_1:[0,1]\to \mathbb{R}$ (or $L_2$) be the unique linear function tangent
to $G_1$ (or $G_2$) at $r_0$. Then since $G_1$ (or $G_2$) is the convex hull of $H_1$ (or $H_2$), we have for every $r\in [0,1]$, $L_1(r)\leq G_1(r)\leq H_1(r)$ and $L_2(r)\leq G_2(r)\leq H_2(r)$. 
\vspace{-.1in}
\paragraph{The existence of $r_1$:} If $G_2(r_0)=H_2(r_0)$, then $H_1(r_0)-H_2(r_0)\geq G_1(r_0)-G_2(r_0)=\rho$. By choosing $r_1=r_0$ we have $\cH(r_1)\geq \rho$. Now assume $G_2(r_0)<H_2(r_0)$. Then $r_0$ is in the interior of some ironed interval $I$ of $\varphi_2$. Let $r_1<r_0$ be the left endpoint of the interval $I$. Then $H_2(r_1)=L_2(r_1)$. 
We have $$\cH(r_1)=H_1(r_1)-H_2(r_1)\geq L_1(r_1)-L_2(r_1)\geq L_1(r_0)-L_2(r_0)=G_1(r_0)-G_2(r_0)=\rho$$
Here the first inequality follows from $H_1(r_1)\geq L_1(r_1)$ and $H_2(r_1)=L_2(r_1)$. The second inequality follows from the fact that $d(L_1(r)-L_2(r))/dr=G_1'(r_0)-G_2'(r_0)=g_1(r_0)-g_2(r_0)<0$. The second-last equality follows from the fact that $L_1$ (or $L_2$) meets $G_1$ (or $G_2$) at $r_0$. 
\vspace{-.1in}
\paragraph{The existence of $r_2$:} The proof follows from a similar argument as $r_1$. If $G_1(r_0)=H_1(r_0)$, then since $G_1'(r_0)$ exists, $r_0$ is not in any ironed interval of $\varphi_1$. Thus there exists a neighborhood of $r_0$ (an open interval $I$ that contains $r_0$) such that $G_1(r)=H_1(r),\forall r\in I$. $H_1$ is differentiable at $r_0$ with derivative $G_1'(r_0)=g_1(r_0)$. We notice that
$$\left.\frac{d(H_1(r)-L_2(r))}{dr}\right\vert_{r=r_0}=g_1(r_0)-g_2(r_0)<0.$$
Thus there exists some $r_2>r_0$ such that
$$\rho=G_1(r_0)-G_2(r_0)=H_1(r_0)-L_2(r_0)>H_1(r_2)-L_2(r_2)\geq H_1(r_2)-H_2(r_2)=\cH(r_2)$$

Now assume $G_1(r_0)<H_1(r_0)$. Then $r_0$ is in the interior of some ironed interval $I$ of $\varphi_1$. Let $r_2>r_0$ be the right endpoint of the interval $I$. Then $H_1(r_2)=L_1(r_2)$. We have $$\cH(r_2)=H_1(r_2)-H_2(r_2)\leq L_1(r_2)-L_2(r_2)< L_1(r_0)-L_2(r_0)=G_1(r_0)-G_2(r_0)=\rho$$
Here the first inequality follows from $H_1(r_2)=L_1(r_2)$ and $H_2(r_2)\geq L_2(r_2)$. The second inequality follows from the fact that $d(L_1(r)-L_2(r))/dr<0$ and $r_2>r_0$. The second-last equality follows from the fact that $L_1$ (or $L_2$) meets $G_1$ (or $G_2$) at $r_0$. 

The existence of $r_1$ and $r_2$ contradicts with the fact that $\cH$ is non-decreasing. Thus $\tilde{\varphi}_1(\theta)\geq \tilde{\varphi}_2(\theta),\forall \theta\in[0,1]$.
\end{proof}
}

\begin{prevproof}{Theorem}{thm:characterization-full-info}
Let $L^*$ be the affine utility function for purchasing the fully informative experiment at price $p$, i.e. $L^*(\theta)=u_{11}\theta+u_{2m}(1-\theta)-p=(u_{11}-u_{2m})\theta+u_{2m}-p$. Denote $\theta_p^L$ and $\theta_p^H$ the two points where $L^*$ intersects with the IR curve. For every $k\in [m]$, let $L_k$ be the affine function for the $k$-th piece of the IR curve, i.e., $L_k(\theta)=\theta\cdot u_{1,m+1-k}+(1-\theta)\cdot u_{2,m+1-k}=\theta\cdot \ell_k+u_{2,m+1-k},\forall \theta\in [0,1]$.

Let $\Bq^*$ be the experiments purchased by the agent in $\MM^*$, i.e., 
\begin{equation}\label{equ:q^*}
q^*(\theta)=
\begin{cases}
-u_{2m}, & \theta\in [0,\theta_p^L)\\
u_{11}-u_{2m}, & \theta\in [\theta_p^L,\theta_p^H)\\
u_{11}, & \theta\in [\theta_p^H,1]
\end{cases}
\end{equation}


\paragraph{Necessary condition.} We first prove that the properties in the statement are necessary. Suppose $\MM^*$ is the optimal {\responsiveIC} and IR mechanism, then $\Bq^*$ is the optimal solution to the program in \Cref{fig:continous-program-modified}. We are going to verify each property of the statement by applying \Cref{thm:characterization} to the solution $\Bq^*$.

\paragraph{Property 1.} Recall that $\theta_1$, $\theta_2$ are the first and last kink of the IR curve. By \Cref{footnote:kink-formal}, $$\theta_1=\frac{u_{2m}-u_{2,m-1}}{u_{1,m-1}+u_{2m}-u_{2,m-1}},\quad \theta_2=\frac{u_{22}}{u_{11}+u_{22}-u_{12}}$$

By \Cref{lem:endpoints of the full info}, $\theta_p^L\leq \theta_1$ and $\theta_p^H\geq \theta_2$. Thus $\theta_p^L$, $\theta_p^H$ are the points where $L^*$ intersects with the first and last piece of the IR curve, respectively. We have $L^*(\theta_p^L)=u_{1m}\theta_p^L+u_{2m}(1-\theta_p^L)\Leftrightarrow \theta_p^L=\frac{p}{u_{11}}$. Similarly, $L^*(\theta_p^H)=u_{11}\theta_p^H+u_{21}(1-\theta_p^H)\Leftrightarrow \theta_p^H=1-\frac{p}{u_{2m}}$. Thus $\theta_p^L\leq \theta_1$ and $\theta_p^H\geq \theta_2$ are equivalent to
$$p\leq\widehat{p}=\min\left\{\frac{(u_{11}-u_{12})u_{2m}}{u_{11}-u_{12}+u_{22}},\frac{(u_{2m}-u_{2,m-1})u_{11}}{u_{2m}+u_{1,m-1}-u_{2,m-1}}\right\}.$$ 

To prove the second half of property 1, by \Cref{thm:characterization}, there exist multipliers $\lambda =\{\lambda_k\}_{k\in\{2,\ldots,m-1\}}\geq \mathbf{0}$ and $\eta$ that satisfy all properties of \Cref{thm:characterization}. We choose $\eta^*=\eta$ and $\lambda^*=\lambda_2$. For every $k\in \{2,3,\ldots,m-1\}$, let $\theta_k^*=\sup\{\theta\in [0,1]:q^*(\theta)\leq \ell_k\}$. Then since $-u_{2m}<\ell_k=u_{1,m+1-k}-u_{2,m+1-k}<u_{11}$, we have $\theta_k^*=\theta_p^L$ iff $\ell_k<u_{11}-u_{2m}$, and $\theta_k^*=\theta_p^H$ otherwise. Notice that
\begin{equation}\label{equ:delta_k}
\begin{aligned}
\Delta_k&:=\int_0^{1}(q^*(x)-\ell_k)\cdot \ind[q^*(x)\leq \ell_k]dx- (u_{2,m+1-k}-u_{2m})\\
&=\int_0^{\theta_k^*}(q^*(x)-\ell_k)dx-u_{2,m+1-k}+u_{2m}
=L^*(\theta_k^*)-L_k(\theta_k^*)
\end{aligned}
\end{equation}

Here the last inequality is because: $u_{2m}+\int_0^{\theta_p^L}q^*(x)dx=u_{2m}(1-\theta_p^L)=u_{2m}(1-\frac{p}{u_{11}})=L^*(\theta_p^L)$ and $u_{2m}+\int_0^{\theta_p^H}q^*(x)dx=L^*(\theta_p^L)+(\theta_p^H-\theta_p^L)\cdot (u_{11}-u_{2m})=L^*(\theta_p^H)$. 

First consider the case when $m=3$. We argue that if $\ell_2\not=u_{11}-u_{2m}$, then $\lambda_2=0$. We only prove the case when $\ell_2<u_{11}-u_{2m}$. The other case is similar. Suppose $\lambda_2>0$, the ironed virtual surplus function $\tilde{J}^{(\eta,\lambda)}$ can be written as follows (denote $\alpha=u_{11}-u_{2m}$):
$$\tilde{J}^{(\eta,\lambda)}(q,\theta)=\tilde{\varphi}^-(\theta)\cdot \min\{q,\alpha\} +\tilde{\varphi}^+(\theta)\cdot \max\{q-\alpha,0\} -\eta q+\lambda_2\cdot\min\{q-\ell_2,0\}$$

For every fixed $\theta$, the function $\tilde{J}^{(\eta,\lambda)}(\cdot,\theta)$ is continuous and 3-piecewise-linear. Thus at least one of the four values $-u_{2m}$, $\ell_2$, $u_{11}-u_{2m}$, and $u_{11}$ must be in $\argmax_{q\in [-u_{2m},u_{11}]}\tilde{J}^{(\eta,\lambda)}(q,\theta)$. We also have
$$\frac{\partial \tilde{J}^{(\eta,\lambda)}}{\partial q}(q,\theta)=
\begin{cases}
\tilde{\varphi}^-(\theta)-\eta+\lambda_2, &q<\ell_2\\
\tilde{\varphi}^-(\theta)-\eta, &q\in (\ell_2,\alpha)\\
\tilde{\varphi}^+(\theta)-\eta, &q>\alpha
\end{cases}
$$

Define $\phi_1(\theta)=\tilde{\varphi}^-(\theta)-\eta+\lambda_2$, $\phi_2(\theta)=\tilde{\varphi}^-(\theta)-\eta$ and $\phi_3(\theta)=\tilde{\varphi}^+(\theta)-\eta$. By \Cref{def:ironed-virtual-value}, $\varphi^-(\theta)-\varphi^{+}(\theta)=f(\theta)>0,\forall \theta\in [0,1]$. Thus by \Cref{lem:compare_ironed_myerson}, $\tilde{\varphi}^-(\theta)\geq\tilde{\varphi}^{+}(\theta),\forall \theta\in [0,1]$. Thus $\phi_1(\theta)>\phi_2(\theta)\geq\phi_3(\theta),\forall \theta\in [0,1]$. We notice that $-u_{2m}\in\argmax_{q}\tilde{J}^{(\eta,\lambda)}(q,\theta)$ iff $\phi_1(\theta)\leq 0$. Since $q^*(\theta)=-u_{2m},\forall \theta\in [0,\theta_p^L)$ and $\Bq^*$ is a pointwise maximizer of $\tilde{J}^{(\eta,\lambda)}(q,\theta)$ almost everywhere (the first property of \Cref{thm:characterization}), we have $\phi_1(0)\leq 0$. Since $\lambda_2>0$, $\phi_2(0)<0$. Similarly, since $u_{11}\in\argmax_{q}\tilde{J}^{(\eta,\lambda)}(q,\theta)$ iff $\phi_3(\theta)\geq 0$, we have $\phi_3(1)\geq 0$. Thus $\phi_2(1)\geq 0$. 
By \Cref{footnote:extended_deriv}, $\phi_2$ is a continuous function on $[0,1]$. Let $\eps=\lambda_2/2>0$, then there exists an opened interval $I$ such that $\phi_2(\theta)\in (-\eps,0),\forall \theta\in I$. At interval $I$, $\phi_1(\theta)=\phi_2(\theta)+\lambda_2>0$, and $\phi_3(\theta)\leq\phi_2(\theta)<0$. This implies that for every $\theta\in I$, $\argmax_{q}\tilde{J}^{(\eta,\lambda)}(q,\theta)$ contains a single value $q=\ell_2$, which contradicts with the fact that $\Bq^*$ is a pointwise maximizer of $\tilde{J}^{(\eta,\lambda)}(q,\theta)$ almost everywhere.  

Thus, if $\lambda^*>0$, then $\ell_2=u_{11}-u_{2m}\implies u_{12}-u_{22}=u_{11}-1$. Moreover, by the second property of \Cref{thm:characterization}, we have $$0=\Delta_2=L^*(\theta_p^H)-L_2(\theta_p^H)=u_{2m}-p-u_{2,m-1}\implies p=1-u_{22}$$

Now consider the case when $m\geq 4$. We prove that $\lambda=\textbf{0}$ 
and thus $\lambda^*=\lambda_2=0$. For every $k\in\{3,\ldots,m-2\}$, since $\theta_p^L\leq \theta_1$ and $\theta_p^H\geq \theta_2$, we have $L^*(\theta_p^L)\geq L_2(\theta_p^L)>L_k(\theta_p^L)$ and $L^*(\theta_p^H)\geq L_{m-1}(\theta_p^H)>L_k(\theta_p^H)$. For every $k\in \{3,\ldots,m-2\}$, $\Delta_k>0$. By the second property of \Cref{thm:characterization}, we have $\lambda_k=0$. For $k=2$, if $\ell_2\geq u_{11}-u_{2m}$, then $\theta_2^*=\theta_p^H$. We thus have $L^*(\theta_p^H)\geq L_{m-1}(\theta_p^H)>L_2(\theta_p^H)$, since $m\geq 4$. By the second property of \Cref{thm:characterization}, $\lambda_2=0$. Similarly, if $\ell_{m-1}<u_{11}-u_{2m}$, $\lambda_{m-1}=0$. If $\ell_{m-1}=u_{11}-u_{2m}$, we have $\theta_p^H>\theta_2$ (see \Cref{fig:full_info_3} for a proof by graph). We have $L^*(\theta_p^H)= L_m(\theta_p^H)>L_{m-1}(\theta_p^H)$, which implies that $\lambda_{m-1}=0$. 
Since $\lambda_k=0,\forall k\in \{3,\ldots,m-2\}$, then $\tilde{J}^{(\eta,\lambda)}(q,\theta)$ can be written as
$$\tilde{J}^{(\eta,\lambda)}(q,\theta)=\tilde{\varphi}^-(\theta)\cdot \min\{q,\alpha\} +\tilde{\varphi}^+(\theta)\cdot \max\{q-\alpha,0\} -\eta q+\lambda_2\cdot\min\{q-\ell_2,0\}+\lambda_{m-1}\cdot\min\{q-\ell_{m-1},0\}$$

Now suppose $\lambda_2>0$ and $\lambda_{k-1}>0$. The case when $\lambda_2>0=\lambda_{k-1}$ and $\lambda_{k-1}>0=\lambda_2$ follows from a similar argument where either the term $\lambda_2\cdot\min\{q-\ell_2,0\}$ or $\lambda_{m-1}\cdot\min\{q-\ell_{m-1},0\}$ is redundant. Since $\lambda_2>0$ and $\lambda_{k-1}>0$, we have $\ell_2<u_{11}-u_{2m}<\ell_{m-1}$. Thus    
$$\frac{\partial \tilde{J}^{(\eta,\lambda)}}{\partial q}(q,\theta)=
\begin{cases}
\tilde{\varphi}^-(\theta)-\eta+\lambda_2+\lambda_{m-1}, &q<\ell_2\\
\tilde{\varphi}^-(\theta)-\eta+\lambda_{m-1}, &q\in (\ell_2,\alpha)\\
\tilde{\varphi}^+(\theta)-\eta+\lambda_{m-1}, &q\in (\alpha,\ell_{m-1})\\
\tilde{\varphi}^+(\theta)-\eta, &q>\ell_{m-1}
\end{cases}
$$

Similar to the proof for the $m=3$ case, since $\lambda_2>0$, we can find an opened interval $I$ such that for every $\theta\in I$, $\argmax_{q}\tilde{J}^{(\eta,\lambda)}(q,\theta)$ contains a single value $\theta=\ell_2$, which contradicts with the fact that $\Bq^*$ is a pointwise maximizer of $\tilde{J}^{(\eta,\lambda)}(q,\theta)$ almost everywhere.\footnote{If $\lambda_2=0$ and $\lambda_{m-1}>0$, we reach a contradiction by finding an opened interval $I$ such that for every $\theta\in I$, $\argmax_{q}\tilde{J}^{(\eta,\lambda)}(q,\theta)$ contains a single value $\theta=\ell_{m-1}$.} Thus we have $\lambda_k=0,\forall k\in \{2,\ldots,m-1\}$.    





\paragraph{Property 2.} 
In property 1, we show that when $m=3$, either $\lambda_2=0$ or $\ell_2=u_{11}-u_{2m}$. When $m\geq 4$, $\lambda=\textbf{0}$. Recall that $\eta^*=\eta$ and $\lambda^*=\lambda_2$. In any of the scenarios, the ironed virtual surplus function $\tilde{J}^{(\eta,\lambda)}$ can be expressed as follows (recall that $\alpha=u_{11}-u_{2m}$):
\begin{equation}\label{equ:ironed_virtual_surplus}
\begin{aligned}
\tilde{J}^{(\eta,\lambda)}(q,\theta)&=\tilde{\varphi}^-(\theta)\cdot \min\{q,\alpha\} +\tilde{\varphi}^+(\theta)\cdot \max\{q-\alpha,0\} -\eta^*\cdot q+\lambda^*(q-\alpha)\cdot \ind[q\leq \alpha]\\
&=
\begin{cases}
(\tilde{\varphi}^-(\theta)-\eta^*+\lambda^*)\cdot q-\lambda^*\alpha, &q\leq \alpha\\
(\tilde{\varphi}^+(\theta)-\eta^*)\cdot q+(\tilde{\varphi}^-(\theta)-\tilde{\varphi}^+(\theta))\cdot \alpha, &q>\alpha
\end{cases}
\end{aligned}
\end{equation}
For every fixed $\theta$, the function $\tilde{J}^{(\eta,\lambda)}(\cdot,\theta)$ is continuous and 2-piecewise-linear. Thus at least one of the three values $-u_{2m}$, $u_{11}-u_{2m}$, and $u_{11}$ must be in $\argmax_{q\in [-u_{2m},u_{11}]}\tilde{J}^{(\eta,\lambda)}(q,\theta)$. 
We study the set $\argmax_{q}\tilde{J}^{(\eta,\lambda)}(q,\theta)$ by a case analysis. 

\begin{itemize}
    \item $\tilde{\varphi}^-(\theta)-\eta^*+\lambda^*\leq 0$: Since $\lambda^*\leq 0$, $\tilde{\varphi}^+(\theta)-\eta^*\leq\tilde{\varphi}^-(\theta)-\eta^*\leq 0$. Thus $-u_{2m}\in\argmax_{q}\tilde{J}^{(\eta,\lambda)}(q,\theta)$. On the other hand, $-u_{2m}\in\argmax_{q}\tilde{J}^{(\eta,\lambda)}(q,\theta)$ clearly implies that $\tilde{\varphi}^-(\theta)-\eta^*+\lambda^*\leq 0$.
    \item $\tilde{\varphi}^-(\theta)-\eta^*+\lambda^*\geq 0$ and $\tilde{\varphi}^+(\theta)-\eta^*\leq 0$: Now $q=u_{11}-u_{2m}$ maximizes $\tilde{J}^{(\eta,\lambda)}(q,\theta)$, and vice versa.
    \item $\tilde{\varphi}^+(\theta)-\eta^*\geq 0$: Then $\tilde{\varphi}^-(\theta)-\eta^*+\lambda^*\geq\tilde{\varphi}^+(\theta)-\eta^*\geq 0$. Thus $u_{11}\in\argmax_{q}\tilde{J}^{(\eta,\lambda)}(q,\theta)$. On the other hand, $u_{11}\in\argmax_{q}\tilde{J}^{(\eta,\lambda)}(q,\theta)$ clearly implies that $\tilde{\varphi}^+(\theta)-\eta^*\geq 0$.
\end{itemize}

Thus by the first property of \Cref{thm:characterization}, $\Bq^*$ is a pointwise maximizer of the ironed virtual surplus function almost everywhere. Thus property 2 holds almost everywhere. Moreover, if there exists some $\theta_0\in [0,1]$ such that $q^*(\theta_0)\not\in \argmax_{q}\tilde{J}^{(\eta,\lambda)}(q,\theta)$. Suppose $q^*(\theta_0)=-u_{2m}$. Then according to the above case analysis, $\tilde{\varphi}^-(\theta_0)-\eta^*+\lambda^*>0$. Since $\tilde{\varphi}^-(\cdot)$ is continuous at $\theta_0$, there exists a neighborhood $I$ of $\theta_0$ such that $\tilde{\varphi}^-(\theta_0)-\eta^*+\lambda^*>0,\forall \theta\in I$.\footnote{When $\theta_0=0$ (or $1$), $\tilde{\varphi}^-(\cdot)$ is right-continuous (or left-continuous) at $\theta_0$. We choose $I$ to be $[0,\delta)$ or $(1-\delta,1]$ for some sufficiently small $\delta>0$.} Thus $-u_{2m}\not\in \argmax_{q}\tilde{J}^{(\eta,\lambda)}(q,\theta),\forall \theta\in I$, contradicting with the fact that $\Bq^*$ is a pointwise maximizer of the ironed virtual surplus function almost everywhere. Similarly, we reach a contradition when $q^*(\theta_0)=u_{11}-u_{2m}$ or $u_{11}$. Hence, $\Bq^*$ is a pointwise maximizer of $\tilde{J}^{(\eta,\lambda)}$ everywhere on $[0,1]$, which implies that property 2 holds for every $\theta\in [0,1]$.          

\paragraph{Property 3.} $\Bq^*$ only takes three different values $-u_{2m}, u_{11}-u_{2m}, u_{11}$. Moreover, $q^{*^{-1}}(-u_{2m})=0$, $q^{*^{-1}}(u_{11}-u_{2m})=\theta_p^L$, $q^{*^{-1}}(u_{11})=\theta_p^H$.   
Thus by the third property of \Cref{thm:characterization}, $\Bq^*$ satisfies the generalized pooling property iff $\theta_p^{L}$ is not in the interior of any ironed interval of $\tilde{\varphi}^-(\cdot)$, and $\theta_p^{H}(\cdot)$ is not in the interior of any ironed interval of $\tilde{\varphi}^+$. 

\paragraph{Sufficient Condition.} Next, we prove that the properties in the statement are sufficient. Given a menu $\MM$ that contains only the fully informative experiment with with price $p>0$. Suppose there exist multiplier $\eta^*$ and $\lambda^*\geq 0$ that satisfies all properties in the statement. We prove that $\Bq^*$ (defined in \Cref{equ:q^*}) is an optimal solution to the program in \Cref{fig:continous-program-modified}. 

Firstly, since $p\leq \widehat{p}$, due to the analysis for Property 1, we know that $\theta_p^L\leq \theta_1$ and $\theta_p^H\geq \theta_2$. Thus $\theta_p^L<\theta_p^H$ and $\Bq^*$ is a feasible solution of the program, which induces strictly positive revenue. 

To apply \Cref{thm:characterization}, we choose multipliers $\eta'$ and $\lambda'=\{\lambda_k'\}_{k\in \{2,\ldots,m-1\}}\geq \textbf{0}$, and prove that all three properties in the statement of \Cref{thm:characterization} are satisfied for $\Bq^*$ and $(\eta',\lambda')$, which implies that $\Bq^*$ is an optimal solution. Let $\eta'=\eta^*$. When $m=3$, let $\lambda_2'=\lambda^*$. When $m\geq 4$, let $\lambda'_k={0}$ for all ${k\in \{2,\ldots,m-1\}}$.

For the first property of \Cref{thm:characterization}, we notice that the ironed virtual surplus function $\tilde{J}^{(\eta',\lambda')}(q,\theta)$ satisfies \Cref{equ:ironed_virtual_surplus}. As argued earlier in the case analysis in property 2, when property 2 of~\Cref{thm:characterization-full-info} is satisfied, $\Bq^*$ is a pointwise maximizer of $\tilde{J}^{(\eta',\lambda')}(q,\theta)$. The first property of \Cref{thm:characterization} is satisfied.

For the second property of \Cref{thm:characterization}, according to the choice of $\lambda'$, we only have to argue that the equality holds when $m=3$ and $\lambda^*>0$. For other scenarios, the property trivially holds since the corresponding multiplier $\lambda_k'$ is 0. Thus by property 1 (in the statement of \Cref{thm:characterization-full-info}), $u_{12}-u_{22}=u_{11}-1$ and $p=1-u_{22}$. We follow the notation $\theta_k^*$ and $\Delta_k$ as used for proving property 1 is necessary. By \Cref{equ:delta_k},
$$\Delta_2=L^*(\theta_2^*)-L_2(\theta_2^*)=L^*(\theta_p^H)-L_2(\theta_p^H)=u_{23}-p-u_{22}=0,~~\text{($u_{23}=1$ as $m=3$)}.$$
Thus the second property of \Cref{thm:characterization} is satisfied. 
As argued above, $\Bq^*$ satisfies the generalized pooling property iff $\theta_p^{L}$ is not in the interior of any ironed interval of $\tilde{\varphi}^-$, and $\theta_p^{H}$ is not in the interior of any ironed interval of $\tilde{\varphi}^+$. Thus the third property of \Cref{thm:characterization} is also satisfied. Hence, $\Bq^*$ and $(\eta',\lambda')$ satisfy all properties in \Cref{thm:characterization}, which implies that $\Bq^*$ is an optimal solution.    
\end{prevproof}

\begin{prevproof}{Theorem}{cor:special_distributions}
We are going to show that there exist $p>0$ and multipliers $\eta^*$, $\lambda^*\geq 0$ that satisfy all properties of \Cref{thm:characterization-full-info}. Then by \Cref{thm:characterization-full-info}, selling the complete information at price $p$ is the optimal menu.

We notice that in \Cref{def:myerson_ironing}, if $\varphi(\cdot)$ is non-decreasing, $H(\cdot)$ is a convex function. Thus $G(r)=H(r),\forall r\in [0,1]$ and $\tilde{\varphi}(\theta)=\varphi(\theta),\forall \theta\in [0,1]$. Since $\varphi^-(\cdot)$ and $\varphi^+(\cdot)$ are both monotonically non-decreasing, we have $\varphi^-(\theta)=\tilde{\varphi}^-(\theta)$ and $\varphi^+(\theta)=\tilde{\varphi}^+(\theta),\forall \theta\in [0,1]$.

Define function $W:[0,1]\to \mathbb{R}$ as: $W(p)=\varphi^-(\frac{p}{u_{11}})-\varphi^+(1-\frac{p}{u_{2m}})$. Then by \Cref{def:ironed-virtual-value}, we have $W(0)=\varphi^-(0)-\varphi^+(1)=-1$. If $\varphi^-(\frac{\widehat{p}}{u_{11}})\geq \varphi^+(1-\frac{\widehat{p}}{u_{2m}})$, $W(\widehat{p})\geq 0$. Since both $\varphi^-$ and $\varphi^+$ are continuous, $W$ is also continuous. Thus there exists $p_0\in (0,\widehat{p}]$ such that $W(p_0)=0$. Consider $p=p_0$ and multipliers $\eta^*=\varphi^-(\frac{p_0}{u_{11}})= \varphi^+(1-\frac{p_0}{u_{2m}}), \lambda^*=0$. Property 1 of \Cref{thm:characterization-full-info} holds since $p_0\leq \widehat{p}$ and $\lambda^*=0$. Property 3 of \Cref{thm:characterization-full-info} holds since neither $\varphi^-(\cdot)$ nor $\varphi^+(\cdot)$ requires ironing. For property 2 of \Cref{thm:characterization-full-info}, since $\theta_{p_0}^L=\frac{p_0}{u_{11}}$ and $\theta_{p_0}^H=1-\frac{p_0}{u_{2m}}$, by our choice of $\eta^*$, $\varphi^-(\theta_{p_0}^L)-\eta^*=0$ and    $\varphi^+(\theta_{p_0}^H)-\eta^*=0$. Property 2 then follows from the fact that $\varphi^-(\cdot)$ and $\varphi^+(\cdot)$ are both monotonically non-decreasing. 

\end{prevproof}

For several standard distributions, $\varphi^-(\cdot)$ and $\varphi^+(\cdot)$ are both monotonically non-decreasing. The following examples are some applications of \Cref{cor:special_distributions}.    

\begin{example}[Uniform Distribution]\label{ex:uniform}
Consider the uniform distribution $U[0,1]$. Then $\varphi^-(\theta)=2\theta$ and $\varphi^+(\theta)=2\theta-1$ are both increasing functions. Thus when $m=3$ and $u_{12}-u_{22}=u_{11}-1$, selling the complete information is optimal. When $u_{11}=u_{2m}=1$, $\varphi^-(\frac{\widehat{p}}{u_{11}})\geq \varphi^+(1-\frac{\widehat{p}}{u_{2m}})$ if and only if $\widehat{p}\geq \frac{1}{4}$. In other words, if $u_{22}\leq 3(1-u_{12})$ and $u_{1,m-1}\leq 3(1-u_{2,m-1})$, selling the complete information is optimal. 
\end{example}

\begin{example}[{Exponential Distributions Restricted on $[0,1]$}]\label{ex:exponential}
$f(\theta)=c\cdot \lambda e^{-\lambda \theta}$, where $c=1-e^{-\lambda}$. $f'(\theta)=-c\lambda^2\cdot e^{-\lambda \theta}$. When $\lambda\leq 2$, $(\varphi^-)'(\theta)=2f(\theta)+\theta f'(\theta)=c\lambda(2-\theta\lambda)\cdot e^{-\lambda \theta}\geq 0$. And $(\varphi^+)'(\theta)=2f(\theta)+(\theta-1) f'(\theta)>0$, since $f'(\theta)<0$. Selling the complete information is optimal if either conditions in \Cref{cor:special_distributions} is satisfied.  
\end{example}

\begin{example}[{Normal Distributions Restricted on $[0,1]$}]\label{ex:normal}
Consider the normal distribution $\mathcal{N}(0,\sigma^2)$ restricted on $[0,1]$. $f(\theta)=c\cdot \exp(-\frac{\theta^2}{2\sigma^2})$, where $c=1/\int_0^1\frac{1}{\sigma\sqrt{2\pi}}\cdot  \exp(-\frac{\theta^2}{2\sigma^2}) d\theta$. $f'(\theta)=-c\cdot \frac{\theta}{\sigma^2} \exp(-\frac{\theta^2}{2\sigma^2})<0$. When $\sigma^2\geq \frac{1}{2}$, $(\varphi^-)'(\theta)=2f(\theta)+\theta f'(\theta)=c(2-\frac{\theta}{\sigma^2})\cdot \exp(-\frac{\theta^2}{2\sigma^2})\geq 0$. And $(\varphi^+)'(\theta)=2f(\theta)+(\theta-1) f'(\theta)>0$, since $f'(\theta)<0$. Selling the complete information is optimal if either conditions in \Cref{cor:special_distributions} is satisfied.  
\end{example}

%% file: proof_multi_state.tex
\section{Missing Details from Section~\ref{sec:multi_state}}\label{sec:proof_multi_state}

\subsection{Missing Details from Section~\ref{sec:general_setting}}\label{sec:appx_multi_lower_bound}

\begin{prevproof}{Theorem}{thm:ratio_characterization}
Let $\beta=\int_0^\infty \frac{1}{r(x)}dx$. We first show that for any distribution $D$, $\frac{\rev(\MM,D)}{\optfi(D)}\leq \beta$. In fact, we have
\begin{align*}\rev(\MM,D)=\E_{\theta\sim D}[t(\theta)]&=\int_0^\infty\Pr_{\theta\sim D}[t(\theta)\geq x]dx\leq \int_0^\infty \Pr_{\theta\sim D}[U(\theta)\geq r(x)]dx\\
&\leq \int_0^\infty \frac{\optfi(D)}{r(x)}dx=\beta\cdot \optfi(D)
\end{align*}

Here the first inequality follows from the definition of $r(x)$: $t(\theta)\geq x$ implies that $U(\theta)\geq r(x)$. For the second inequality, we notice that by the definition of $U(\theta)$, the buyer will purchase the fully informative experiment at price $p$ if and only if $U(\theta)\geq p$. Thus $\optfi(D)\geq r(x)\cdot \Pr[U(\theta)\geq r(x)]$. 

Now we prove that $\ratio(\MM)\geq \beta$. In particular, we show that for any $\beta'<\beta$ there exists a distribution $D$ such that $\rev(\MM,D)>\beta'\cdot \optfi(D)$.

Since $1/r(x)$ is weakly decreasing, non-negative, and $\int_0^\infty \frac{1}{r(x)}dx=\beta$, there exist $0=t_0<t_1<\cdots<t_N<\infty$ with $0=r(t_0)<r(t_1)<\cdots<r(t_N)<\infty$ such that~\footnote{We notice that $U(\bm{0})=0$ as the payment $t(\cdot)$ is non-negative. Thus $r(t_0)=r(0)=0$.}
$$\beta''=\sum_{k=1}^N\frac{t_k-t_{k-1}}{r(t_k)}>\beta'.$$

Let $\eps>0$ be small enough so that $\beta''>(1+\eps)\beta'$ and $r(t_{k+1})>(1+\eps)r(t_{k})$ for all $1\leq k<N$. By the definition of $r(\cdot)$, we can choose a sequence of types $\{\theta^{(k)}\}_{k\in [N]}$ in $\Theta$ such that $t(\theta^{(k)})\geq t_k$ and $r(t_k)\leq U(\theta^{(k)})<(1+\eps)r(t_k)$. Then
$$\sum_{k=1}^N\frac{t_k-t_{k-1}}{U(\theta^{(k)})}>\sum_{k=1}^N\frac{t_k-t_{k-1}}{(1+\eps)r(t_k)}=\frac{\beta''}{1+\eps}>\beta'$$

For every $k\in [N]$, let $\xi_k=U(\theta^{(k)})$. We notice that $U(\theta^{(k)})<(1+\eps)r(t_k)<r(t_{k+1})\leq U(\theta^{(k+1)})$ for $0\leq k<N$. Consider the distribution $D$ with support $\{\theta^{(1)},\ldots, \theta^{(N)}\}$ and $\Pr[\theta=\theta^{(k)}]=\xi_1\cdot(\frac{1}{\xi_k}-\frac{1}{\xi_{k+1}})$ for all $k\in [N]$, where $\xi_{N+1}=\infty$.  

Recall that the buyer will purchase the fully informative experiment at price $p$ if and only if $U(\theta)\geq p$. Thus under distribution $D$, it's sufficient to consider the price $p=U(\theta^{(k)})=\xi_k$ for some $k\in [N]$. Since $\{U(\theta^{(k)})\}_{k\in [N]}$ is a strictly increasing sequence, then the buyer will purchase at price $p=\xi_k$ whenever $\theta\in \{\theta^{(k)},\ldots,\theta^{(N)}\}$, which happens with probability $\sum_{j=k}^N\xi_1\cdot(\frac{1}{\xi_k}-\frac{1}{\xi_{k+1}})=\frac{\xi_1}{\xi_k}$. Hence we have $\optfi(D)=\xi_1$. 

\begin{align*}
\rev(\MM,D)=\sum_{k=1}^Nt(\theta^{(k)})\cdot \Pr[\theta=\theta^{(k)}]&=\sum_{k=1}^Nt_k\cdot \xi_1\cdot(\frac{1}{\xi_k}-\frac{1}{\xi_{k+1}})\\
&=\xi_1\cdot \sum_{k=1}^N \frac{t_k-t_{k-1}}{\xi_k}>\beta'\cdot \optfi(D) 
\end{align*}
\end{prevproof}

\begin{prevproof}{Lemma}{lem:ratio-opt and fi-special}
Recall that under this environment $U(\theta)=\sum_{i=1}^3\theta_i\cdot 1-\max\{\theta_1,\theta_2,\theta_3\}=1-\max\{\theta_1,\theta_2,\theta_3\}$ (here $\theta_3=1-\theta_1-\theta_2$). We first construct the distribution $D$ based on the given sequence $\{y_k\}_{k=1}^N$. Let $\{t_k\}_{k=1}^N$ be a sequence of positive numbers that increases fast enough so that: (i) $(t_k/\gap_k)\cdot U(y_k)$ is increasing, and (ii) $t_{k+1}/t_k\geq 1/\eps$ for all $1\leq k<N$. Such a sequence must exist since after $t_1,\ldots,t_{k-1}$ are decided, we can choose a large enough $t_k$ to satisfy both properties. Let $\delta>0$ be any number such that $1/\delta>\max_{k\in [N]}\{t_k/\gap_k\}$. For every $k\in [N]$, define $x_k\in [0,1]^2$ as $x_k=(\delta t_k/\gap_k)\cdot y_k$. For every $k\in [N]$, since $\delta t_k/\gap_k<1$ and $y_k\in \Theta$, we have $x_k\in \Theta$. 

By Property 2 and 3 of the statement, $0.4\geq \sqrt{y_{k,1}^2+y_{k,2}^2}\geq \sqrt{y_{k,1}^2+(0.9y_{k,1})^2}$, which implies that $y_{k,1}\leq \sqrt{0.4^2/1.81}<\frac{1}{3}$. Similarly, $y_{k,2}<\frac{1}{3}$. Let $y_{k,3}=1-y_{k,1}-y_{k,2}$ and $x_{k,3}=1-x_{k,1}-x_{k,2}$. Then $U(y_k)=1-\max\{y_{k,1},y_{k,2}, y_{k,3}\}=1-y_{k,3}=y_{k,1}+y_{k,2}$. Moreover, since $\delta t_k/\gap_k<1$, we have $x_{k,i}<y_{k,i}<\frac{1}{3}$ for $i=1,2$. Thus $U(x_k)=1-x_{k,3}=x_{k,1}+x_{k,2}=(\delta t_k/\gap_k)\cdot U(y_k)$. For every $k\in [N]$, let $\xi_k=U(x_k)$. Then $\{\xi_k\}_{k=1}^N$ is an increasing sequence due to (i). 

Consider the distribution $D$ with support $\{x_1,\ldots, x_N\}$. $\Pr[\theta=x_k]=\xi_1\cdot(\frac{1}{\xi_k}-\frac{1}{\xi_{k+1}})$ for all $k\in [N]$, where $\xi_{N+1}=\infty$. In the next step, we construct a sequence of experiments $\{\Pi^{(k)}\}_{k\in[N]}$ and a mechanism $\MM$. 
For every $k\in [N]$, define experiment $\Pi^{(k)}$ as follows:
\begin{table}[h]
\centering
\begin{tabular}{ c| c c c} \label{tab:payoff matrix}
$\Pi^{(k)}$ & $1$ & $2$ & $3$\\ 
  \hline
$\omega_1$ & $y_{k,1}$ & 0 & $1-y_{k,1}$\\ 
 $\omega_2$ & 0 &$y_{k,2}$ & $1-y_{k,2}$\\
 $\omega_3$ & 0 & 0 & 1
\end{tabular}
\end{table}

Now consider the following mechanism $\MM$. For every buyer's type $\theta$ in the support, the buyer chooses the experiment $\Pi^{(k^*)}$ where $k^*=\argmax_k\left\{ V_{\theta}^*(\Pi^{(k)})-\delta\cdot t_k\right\}$, and pays $\delta\cdot t_{k^*}$. In other words, he chooses the option (with experiment $\Pi^{(k^*)}$ and price $\delta\cdot t_{k^*}$) that obtains the highest utility, when he follows the recommendation. 

\begin{claim}\label{claim:M-IC and IR}
$\MM$ is IC and IR.
\end{claim}

\begin{proof}
We first prove that $\MM$ is IR. For every buyer's type $x_k$, the buyer's utility after purchasing $\Pi^{(k^*)}$ is
\begin{align*}
V_{x_k}(\Pi^{(k^*)})-\delta\cdot t_{k^*}&\geq V^*_{x_k}(\Pi^{(k^*)})-\delta\cdot t_{k^*}\\
&\geq V_{x_k}^*(\Pi^{(k)})-\delta\cdot t_k\\
&=x_{k,1}y_{k,1}+x_{k,2}y_{k,2}+x_{k,3}-\delta\cdot t_k\\
&=\delta t_k\cdot \left(\frac{y_{k,1}^2+y_{k,2}^2}{\gap_k}-1\right)+x_{k,3}\\
&\geq x_{k,3}
\end{align*}
Here the second inequality follows from the definition of $k^*$; The first equality follows from the fact that: under the matching utility environment, $V_{\theta}^*(\Pi^{(k)})=\theta_1 y_{k,1}+\theta_2 y_{k,2}+\theta_3$ for every type $\theta\in \Theta$; The second equality follows from the definition of $x_k$; The last inequality follows from Property 1 and 3 of the statement. We notice that at type $x_k$, the buyer has value $\max_j\{\sum_i x_{k,i}\cdot u_{ij}\}=\max\{x_{k,1},x_{k,2},x_{k,3}\}=x_{k,3}$ before purchasing any experiment. Thus $\MM$ is IR.

By the definition of $\MM$, it's clear that $\MM$ is {\responsiveIC}. To prove that $\MM$ is also IC, it suffices to show that for any type $\theta$ from the distribution $D$ and any experiment $\Pi^{(k)}$, following the recommendation always maximizes the buyer's utility, i.e. $V_{\theta}(\Pi^{(k)})=V_{\theta}^*(\Pi^{(k)})$. In fact, according to the construction of $\Pi^{(k)}$, when action 1 or 2 is recommended, the corresponding state (1 or 2, respectively) is fully revealed and thus following the recommendation clearly maximizes the buyer's expected payoff. When action 3 is recommended, the buyer with type $x_q$ (for some $q\in [N]$) has expected payoff $x_{q,1}(1-y_{k,1})$, $x_{q,2}(1-y_{k,2})$ and $x_{q,3}$, by choosing action 1, 2 and 3 respectively. Since $x_{q,1}<\frac{1}{3}$ and $x_{q,2}<\frac{1}{3}$, we have $x_{q,3}>\frac{1}{3}\geq \max\{x_{q,1}(1-y_{k,1}), x_{q,2}(1-y_{k,2})\}$. Thus the buyer will always follow the recommendation for every experiment $\Pi^{(k)}$. Since $\MM$ is {\responsiveIC}, $\MM$ is also IC.
\end{proof}

Back to the proof of \Cref{lem:ratio-opt and fi-special}. Now we compute $\rev(\MM, D)$. We notice that for every $k\in [N]$ and every type $\theta\in \Theta$, $$V_{\theta}^*(\Pi^{(k)})=\theta_1\cdot y_{k,1}+\theta_2\cdot y_{k,2}+\theta_3.$$
Consider any buyer's type $x_k$. For any $1\leq q<k$, by the definition of $x_k$ and $\gap_k$, we have
\begin{align*}
V_{x_k}^*(\Pi^{(k)})-V_{x_k}^*(\Pi^{(q)})&=x_{k,1}\cdot (y_{k,1}-y_{q,1})+x_{k,2}\cdot (y_{k,2}-y_{q,2})\\
&=\frac{\delta t_k}{\gap_k}\cdot \{y_{k,1}\cdot (y_{k,1}-y_{q,1})+y_{k,2}\cdot (y_{k,2}-y_{q,2})\}\\
&\geq \delta t_k>\delta t_k-\delta t_q    
\end{align*}

Recall that in mechanism $\MM$, the buyer chooses the experiment $\Pi^{(q)}$ that maximizes $\left\{ V_{\theta}^*(\Pi^{(q)})-\delta\cdot t_q\right\}$. Thus the above inequality implies that in $\MM$, the buyer must purchase an experiment $\Pi^{(q)}$ where $q\geq k$. Since $\{t_k\}_{k=1}^N$ is an increasing sequence, the buyer's payment is at least $\delta\cdot t_k$. By \Cref{thm:ratio_characterization},

\begin{align*}
\rev(\MM,D)&\geq \sum_{k=1}^N\delta\cdot t_k\cdot \xi_1\cdot(\frac{1}{\xi_k}-\frac{1}{\xi_{k+1}})=\delta\xi_1\cdot \sum_{k=1}^N \frac{t_k-t_{k-1}}{\xi_k}\\
&\geq (1-\eps)\xi_1\cdot \sum_{k=1}^N \frac{\delta t_k}{\xi_k}=(1-\eps)\cdot \sum_{k=1}^N\frac{\gap_k}{y_{k,1}+y_{k,2}}\cdot \optfi(D)\\
&\geq \frac{3}{2}(1-\eps)\cdot \sum_{k=1}^N\gap_k \cdot \optfi(D)
\end{align*}

Here the second inequality follows from $t_{k+1}/t_k\geq 1/\eps$. The third inequality follows from the fact that both $y_{k,1}$ and $y_{k,2}$ are at most $\frac{1}{3}$, as shown in the beginning of the proof. The second equality follows from $\xi_k=U(x_k)=(\delta t_k/\gap_k)\cdot U(y_k)=(\delta t_k/\gap_k)\cdot (y_{k,1}+y_{k,2})$ and the fact that $\optfi(D)=\xi_1$, which is proved in \Cref{thm:ratio_characterization}. We include the proof here for completeness. Recall that a buyer with type $x$ will purchase the fully informative experiment at price $p$ if and only if  $U(x)\geq p$. Thus under distribution $D$, it's sufficient to consider the price $p=U(x_{k})=\xi_k$ for some $k\in [N]$. Since $\{U(x_{k})\}_{k\in [N]}$ is a strictly increasing sequence, then the buyer will purchase at price $p=\xi_k$ whenever $x\in \{x_{k},\ldots,x_N\}$, which happens with probability $\sum_{j=k}^N\xi_1\cdot(\frac{1}{\xi_k}-\frac{1}{\xi_{k+1}})=\frac{\xi_1}{\xi_k}$. Hence, $\optfi(D)=\xi_1$.

The proof is done by noticing that $\rev(\MM, D)\leq \opt(D)$, since $\MM$ is IC and IR by \Cref{claim:M-IC and IR}.
\end{prevproof}

\begin{prevproof}{Theorem}{thm:general-main}
For any integer $N$, we construct a sequence $\{y_k\}_{k=1}^N$ that satisfies all three properties in the statement of \Cref{lem:ratio-opt and fi-special}, and $\gap_k=\Theta(k^{-6/7})$. Then by \Cref{lem:ratio-opt and fi-special}, there exists a distribution $D$ with support size $N$ such that (by choosing $\eps=\frac{1}{3}$) 
$$\frac{\opt(D)}{\optfi(D)}\geq \sum_{k=1}^N\gap_k=\Omega(\sum_{k=1}^Nk^{-6/7})=\Omega(N^{1/7})$$
All the points $\{y_k\}_{k=1}^N$ are placed in a sequence of ``shells''. As $k$ goes larger, $y_k$ stays in a shell with a larger and larger radius. For every $i=1,\ldots,M$, the $i$-th shell has radius $0.3+\sum_{\ell=1}^i\ell^{-3/2}/\alpha$, where $\alpha=0.1\cdot \sum_{\ell=1}^M\ell^{-3/2}$. For every shell $i$, let $r_i$ be the arc from angle $\arctan(\frac{9}{10})$ to angle $\arctan(\frac{10}{9})$.
 The $i$-th shell contains $i^{3/4}$ points. All points are evenly spread in the arc $r_i$, so that angle between any two of them is at least $c\cdot i^{-3/4}$ for some absolute constant $c$. 

We notice that the radius of every shell is in the range $[0.3, 0.4]$. Thus any point $y_k$ on each shell satisfies $y_{k,1}+y_{k,2}\leq 1$ and $||y_k||_2=\sqrt{y_{k,1}^2+y_{k,2}^2}\in [0.3, 0.4]$. Moreover, according to the definition of $r_i$, all points $y_k$ on this arc must satisfy $\frac{y_{k,1}}{y_{k,2}}\in [\frac{9}{10},\frac{10}{9}]$. It remains to analyze $\gap_k=\min_{0\leq j<k}\{(y_{k,1}-y_{j,1})\cdot y_{k,1}+(y_{k,2}-y_{j,2})\cdot y_{k,2}\}=\min_{0\leq j<k}\{y_k\cdot (y_k-y_j)\}$. We have that $y_j\cdot y_k=||y_j||_2\cdot ||y_k||_2\cdot \cos(\alpha)$ where $\alpha$ is the angle between $y_k$ and $y_j$. Let $i$ and $i'$ be the shell that $y_k$ and $y_j$ are placed in respectively. Then $i'\leq i$. Suppose $i'=i$. Since $\alpha=\Omega(i^{-3/4})$, we have $\cos(\alpha)=1-\Omega(i^{-3/2})$ (because $\cos(\alpha)=1-\alpha^2/2+\alpha^4/24-\ldots$).
$(y_k-y_j)\cdot y_k=||y_k||_2^2\cdot \Omega(i^{-3/2})=\Omega(i^{-3/2})$. 
Now suppose $i'<i$. $||y_k||_2-||y_j||_2=\sum_{\ell=i'+1}^i\ell^{-3/2}/\alpha\geq i^{-3/2}/\alpha$. Since $\sum_{\ell=1}^\infty \ell^{-3/2}$ converges, $\alpha=0.1\cdot \sum_{\ell=1}^M\ell^{-3/2}$ is bounded by some absolute constant. Thus 
$$(y_k-y_j)\cdot y_k\geq ||y_k||_2^2-||y_j||_2 ||y_k||_2 \geq 0.3(||y_k||_2-||y_j||_2)=\Omega(i^{-3/2}),$$
where the second inequality follows from $||y_j||_2\geq 0.3$. Combining both cases, we have proved that $\gap_k=\Theta(i^{-3/2})$. Since the first $i$ shells together contain $\sum_{\ell=1}^i i^{3/4}=\Theta(i^{7/4})$ points, we have $k=\Theta(i^{7/4})$ and $\gap_k=\Theta(k^{-6/7})$.     
\end{prevproof}

\subsection{Missing Details from Section~\ref{sec:matching_utilities}}\label{sec:appx_multi_uniform}

\notshow{
We give a formal definition of the simplex of dimension $n-1$.
\begin{definition}\label{def:simplex}
A $(n-1)$-simplex $C$ determined by $n$ affine independent vectors $\bm{x}_1,\ldots,\bm{x}_n\in \mathbb{R}^{n-1}$ is a $(n-1)$ dimensional polytope defined as $C=\{\sum_{i=1}^{n}a_i\bm{x}_i\mid \sum_{i=1}^na_i=1 \text{ and }a_i\geq 0, \forall i\}$. $C$ is called a unit $(n-1)$-simplex if for every $i\in [n-1]$, $\bm{x}_i=\bm{e}_i$ is the $i$-th standard basis vector, and that $\bm{x}_i=\textbf{0}$. For any $p>0$, $C$ is called a standard $(n-1)$-simplex with side length $p$ if for every $i\in [n-1]$, $\bm{x}_{i+1}=\bm{x}_1+p\cdot \bm{e}_i$. 
\end{definition}

\begin{remark}\label{remark:simplex}
$\Theta$ is a unit $(n-1)$-simplex, which has volume $\vol(\Theta)=\frac{1}{(n-1)!}$. When $n=3$, $\Theta$ is a triangle. When $n=4$, $\Theta$ becomes a tetrahedron. Moreover, any standard $(n-1)$-simplex with side length $p$ has volume $\frac{p^{n-1}}{(n-1)!}$. 
\end{remark}
}

The following observation directly follows from the definition of responsive mechanisms. 

\begin{observation}\label{obs:responsive}
In the matching utilities environment, a mechanism $\MM$ is responsive if and only if $\theta_i\pi_{ii}(\theta)\geq \theta_j\pi_{ji}(\theta),\forall \theta\in T, i,j\in [n]$. Here $\theta_n=1-\sum_{i=1}^{n-1}\theta_i$.
\end{observation}

\begin{prevproof}{Lemma}{lem:one-column-has-1}
Let $\MM$ be any optimal responsive mechanism. Then the buyer's value for experiment $E(\theta)$ is ($\theta_n=1-\sum_{i=1}^{n-1}\theta_i$): 
\begin{equation}\label{equ:V*}
V_{\theta}^*(E(\theta))=\sum_{i=1}^n\theta_i\pi_i(\theta)=\sum_{i=1}^{n-1}\theta_i\cdot(\pi_i(\theta)-\pi_n(\theta))+\pi_n(\theta)    
\end{equation}

Now for every type $\theta$, let $c(\theta)=1-\max_i\{\pi_i(\theta)\}\geq 0$. For every $i\in [n]$, we arbitrarily move a total mass of $c(\theta)$ from $\{\pi_{ij}(\theta)\}_{j\not=i}$ to $\pi_{ii}(\theta)$. We also increase the payment of type $\theta$ by $c(\theta)$. Let $\MM'=\{E'(\theta),t'(\theta)\}_{\theta\in \Theta}$ be the induced mechanism. By \Cref{obs:responsive}, $\MM'$ is responsive since $\MM$ is responsive and $\pi_{ii}(\theta)$ is (weakly) increased while $\pi_{ij}(\theta)$ is (weakly) decreased for $i\not=j$. 

Moreover, for every $\theta\in \Theta$, the buyer's utility after purchasing $E'(\theta)$, $V_\theta^*(E'(\theta))-t'(\theta)$, is the same as $V_\theta^*(E(\theta))-t(\theta)$. Since $\MM$ is IC, IR, $\MM'$ is also IC, IR. And $\rev(\MM')\geq \rev(\MM)$. Thus $\MM'$ is also the optimal responsive mechanism. The proof is finished by noticing that $\MM'$ satisfies the property in the statement due to the choice of $c(\theta)$. 
\end{prevproof}

\begin{prevproof}{Observation}{obs:partial G}
For any $\theta, \theta'\in \Theta$, since the mechanism is responsive, IC and IR, we have
\begin{align*}
&G(\theta')-G(\theta)\geq [V_{\theta'}^*(E(\theta))-t(\theta)]-[V_{\theta}^*(E(\theta))-t(\theta)]\\
=&V_{\theta'}^*(E(\theta))-V_{\theta}^*(E(\theta))=\sum_{i=1}^{n-1}(\pi_i(\theta)-\pi_n(\theta))(\theta_i'-\theta_i)
\end{align*}
Taking derivative on both sides at $\theta'=\theta$ finishes the proof.
\end{prevproof}

\begin{prevproof}{Lemma}{obs:c-function}
For any $\theta$, the buyer's utility function $G(\theta)$ is the maximum of a collection of linear functions over $\theta$. Thus $G(\cdot)$ is convex. Moreover, for any $\theta, \theta'\in \Theta$,
\begin{align*}
G(\theta)-G(\theta')\leq [V_{\theta}(E(\theta))-t(\theta)]-[V_{\theta'}(E(\theta))-t(\theta)]
\leq c(\theta,\theta')  
\end{align*}
Here the first inequality follows from the fact that $G(\theta')=V_{\theta'}(E(\theta'))-t(\theta')\geq V_{\theta'}(E(\theta))-t(\theta)$ since $\MM$ is IC. 
The second inequality follows from the definition of $c(\theta,\theta')$. 

It remains to show that $G(\textbf{0})=1$. Since $\MM$ is IR, $G(\textbf{0})\geq u(\textbf{0})=1$. On the other hand, the buyer's value for any experiment is at most 1. Thus $G(\textbf{0})\leq 1$. 
\end{prevproof}

\begin{prevproof}{Lemma}{lem:divergence}
We use the notations from \Cref{def:measure}. For every $\theta\in \Theta$, by the product rule, we have~\footnote{For any real-valued function $F(\cdot)$ over $\Theta$, $\divg(F)$ is the divergence of $F$, defined as $\divg(F)=\sum_{i=1}^{n-1}\frac{\partial F}{\partial \theta_i}.$}
\begin{align*}
\divg(G(\theta)f(\theta)\cdot \theta)&=\theta\cdot \nabla(G(\theta)\cdot f(\theta))-G(\theta)f(\theta)\cdot \nabla \theta\\
&=\theta\cdot \left[G(\theta)\cdot\nabla f(\theta)+f(\theta)\cdot\nabla G(\theta)\right]+(n-1)\cdot G(\theta)f(\theta)    
\end{align*}

Integrating over $\Theta$ for both sides, by the divergence theorem, we have
\begin{equation}\label{equ:divergence_1}
\int_{\Theta}(\nabla G(\theta)\cdot \theta) f(\theta)d\theta=\int_{\partial \Theta}G(\theta)f(\theta)(\theta\cdot \bm{n})d\theta-\int_{\Theta}G(\theta)\cdot \left[\nabla f(\theta)\cdot \theta+(n-1)f(\theta)\right]d\theta    
\end{equation}

Similarly, for every $i\in [n-1]$,
\begin{align*}
\divg(G(\theta)f(\theta)\cdot \bm{e}_i)
=\bm{e}_i\cdot \left[G(\theta)\cdot\nabla f(\theta)+f(\theta)\cdot\nabla G(\theta)\right]
\end{align*}

Thus
\begin{equation}\label{equ:divergence_2}
\begin{aligned}
\int_{\Theta_i}\frac{\partial G(\theta)}{\partial \theta_i}f(\theta)d\theta&=
\int_{\Theta_i}(\nabla G(\theta)\cdot \bm{e}_i) f(\theta)d\theta\\
&=\int_{\partial \Theta_i}G(\theta)f(\theta)(\bm{e}_i\cdot \bm{n}_i)d\theta-\int_{\Theta_i}G(\theta)\cdot(\bm{e}_i\cdot \nabla f(\theta))d\theta    
\end{aligned}
\end{equation}

Combining \Cref{equ:divergence_1}, \Cref{equ:divergence_2} and the fact that $\int_{\Theta}f(\theta)d\theta=1$ completes the proof.
\end{prevproof}

\begin{prevproof}{Lemma}{lem:weak-duality}
By \Cref{lem:divergence},
\begin{align*}
\int_{\Theta}G(\theta)d\mu^P &= \int_{\Theta}\left[-G(\theta)+\nabla G(\theta)\cdot \theta+G(\textbf{0})\right]f(\theta)d\theta-\sum_{i=1}^{n-1}\int_{\Theta_i} \frac{\partial G(\theta)}{\partial \theta_i}f(\theta)d\theta\\
&\geq \int_\Theta \left[-G(\theta)+\nabla G(\theta)\cdot \theta-\max\left\{\frac{\partial G(\theta)}{\partial \theta_1},\ldots, \frac{\partial G(\theta)}{\partial \theta_{n-1}}, 0\right\}+G(\textbf{0})\right]f(\theta)d\theta 
\end{align*}
The last inequality achieves equality if and only if Condition (1) is satisfied. It suffices to prove that $\int_{\Theta}G(\theta)d\mu^P\leq \int_{\Theta\times \Theta} c(\theta,\theta')d\gamma(\theta,\theta')$. In fact, since $G(\theta)-G(\theta')\leq c(\theta,\theta'),\forall \theta,\theta'\in \Theta$, we have
\begin{align*}
\int_{\Theta\times \Theta} c(\theta,\theta')d\gamma(\theta,\theta')&\geq \int_{\Theta\times \Theta}(G(\theta)-G(\theta'))d\gamma(\theta,\theta')\\
&=\int_{\Theta}G(\theta)d\gamma_1-\int_{\Theta}G(\theta')d\gamma_2\\
&\geq \int_{\Theta}G(\theta)d\mu^P
\end{align*}
The first inequality achieves equality if and only if Condition (2) is satisfied. The last inequality follows from the fact that $\gamma_1-\gamma_2\cuxd \mu^P$ and $G$ is convex. The last inequality achieves equality if and only if Condition (3) is satisfied.
\end{prevproof}

\begin{prevproof}{Observation}{obs:opt-mu-P}
Firstly, $\nabla f(\theta)=\textbf{0}$ for every $\theta$ since $f(\theta)$ is a constant. Thus the last two terms in \Cref{def:measure} is always 0. Moreover, the term $\ind_A(\textbf{0})$ contributes a point mass of $+1$ at $\textbf{0}$. The term $-n\cdot\int_{\Theta}\ind_A(\theta)f(\theta)d\theta$ contributes a total mass of $-3$ uniformly distributed through out $\Theta$. By \Cref{lem:divergence}, the other two terms contribute a total mass of $2$ on the boundary of triangle $\Theta$ and line segments $ad, bd, cd$. We notice that when $\theta$ is on the x-axis (or y-axis), both $\theta\cdot \bm{n}$ and $\bm{e}_1\cdot \bm{n}_1$ (or $\bm{e}_2\cdot \bm{n}_2$) are 0. Moreover, for any point $\theta$ on the line segment $\theta_1+\theta_2=1$, the outward pointing unit vector is $(1,1)$. Thus $\theta\cdot \bm{n}=1=\bm{e}_i\cdot \bm{n}_i$ for any $i\in \{1,2\}$. Thus there is no mass on the boundary of the big triangle $\Theta$ and a total mass of $2$ is distributed among line segments $ad, bd, cd$. The mass on each segment is then obtained by calculating the line integral.
\end{prevproof}

\begin{prevproof}{Lemma}{lem:verify-optimality}
We verify feasibility and all three conditions in the statement of \Cref{lem:weak-duality}:
\paragraph{Primal Solution Feasibility.} Since $G^*$ is induced by a responsive, IC and IR mechanism $\MM^*$, by \Cref{obs:c-function}, $G^*$ is a feasible solution to the primal.

\paragraph{Condition 1.} By \Cref{obs:partial G}, for every $i\in \{1,2\}$, $\frac{\partial G^*(\theta)}{\partial \theta_i}=\pi_i^*(\theta)-\pi_3^*(\theta), \forall \theta\in \Theta$. Now for any $i\in \{1, 2, 3\}$ and $\theta\in \Theta_i$, if $\theta\in \Omega_1$, $\pi_j^*(\theta)=1,\forall j\in \{1, 2, 3\}$. Thus $\max\{\frac{\partial G^*(\theta)}{\partial \theta_1}, \frac{\partial G^*(\theta)}{\partial \theta_2}\}=0$; If $\theta\in \Omega_2$, the buyer does not purchase any experiment. Since $\MM^*$ is responsive, by \Cref{obs:responsive} and the fact that $\theta_i\geq \theta_j, \forall j\not=i$, we have that $\pi_{ji}^*(\theta)=1,\forall j\in \{1,2,3\}$. Thus $\max\{\frac{\partial G^*(\theta)}{\partial \theta_1}, \frac{\partial G^*(\theta)}{\partial \theta_2}\}$ equals to 1 if $i\in \{1,2\}$, and -1 otherwise. Thus Condition (1) is satisfied.

\paragraph{Condition 2.} $\gamma^*(\theta,\theta')>0$ if and only if $\theta=\textbf{0}$ and $\theta'\in \Omega_2$. Thus $G^*(\theta)-G^*(\theta')=1-\max\{\theta_1',\theta_2',1-\theta_1'-\theta_2'\}$, since the buyer with any type $\theta''$ has utility $u(\theta'')=\max\{\theta_1'',\theta_2'',1-\theta_1''-\theta_2''\}$ if she does not purchase any experiment. Consider the experiment $E$ such that $\pi_{j1}(E)=1,\forall j\in \{1,2,3\}$ and 0 elsewhere, which is one of the no information experiments. 
Then $V_{\textbf{0}}(E)=u(\textbf{0})=1$ and $V_{\theta'}(E)=u(\theta')=\max\{\theta_1',\theta_2',1-\theta_1'-\theta_2'\}$. $G^*(\theta)-G^*(\theta')=V_{\textbf{0}}(E)-V_{\theta'}(E)$. Moreover, $G^*(\theta)-G^*(\theta')\leq c(\theta,\theta')$ since $G^*$ is feasible. Thus the inequality achieves equality for any $(\theta,\theta')$ such that $\gamma^*(\theta,\theta')>0$. 

\paragraph{Condition 3 and Dual Solution Feasibility.} Denote $\gamma_1^*, \gamma_2^*$ the marginals of $\gamma$. We prove that we can transform $\mu^P$ to $\gamma_1^*-\gamma_2^*$ through a sequence of mean-preserving spreads in the region where $G^*$ is linear. By the definition of $\gamma^*$, $\gamma_1^*$ has a point mass of $+1$ at $\textbf{0}$. 
$-\gamma_2^*$ has a mass of $-1$ uniformly distributed on $\Omega_2$. Recall that $\Omega_2=\{\theta\in \Theta:\max\{\theta_1,\theta_2\}\geq \frac{2}{3}\}\cup\{\theta\in \Theta: \theta_1+\theta_2\leq \frac{1}{3}\}$, which consists of 3 right triangles with side length $\frac{1}{3}$. Thus $\vol(\Omega_2)=\frac{1}{6}=\frac{1}{3}\cdot \vol(\Theta)$.
Thus by \Cref{obs:opt-mu-P}, $\mu^{P^*}(\Omega_2)=(\gamma_1^*-\gamma_2^*)(\Omega_2)$. Hence $\eta=\mu^{P^*}-(\gamma_1^*-\gamma_2^*)$ can be written as $\eta_1-\eta_2$, where
\begin{enumerate}
    \item $\eta_1$: A mass of $\frac{2}{3}$ uniformly distributed on each line segment $ad, bd, cd$, having a total mass of $2$.
    \item $\eta_2$: A mass of $2$ uniformly distributed through out $\Omega_1$.
\end{enumerate}

We notice that for every $\theta\in \Omega_1$, the buyer purchases the fully informative experiment and thus $G^*(\theta)=\frac{2}{3}$ is constant throughout the region $\Omega_1$. 
Hence, to prove $\gamma_1^*-\gamma_2^*\cuxd \mu^{P^*}$ and that Condition 3 is satisfied, it suffices to show that we can spread the positive mass on 
$ad, bd, cd$ to the whole region $\Omega_1$ via mean-preserving spreads, to ``zero out'' the negative mass (in $-\eta_2$). 

To visualize the proof, we map each point $(\theta_1,\theta_2)\in \Theta$ to the point $(\theta_1,\theta_2,\theta_3=1-\theta_1-\theta_2)\in [0,1]^3$, so that the type space becomes a regular triangle (\Cref{fig:condition3}). Now $\Omega_1$ is separated into three pentagons by $ad, bd$ and $cd$. 

\begin{figure}[ht]
	\centering
	\includegraphics[scale=0.6]{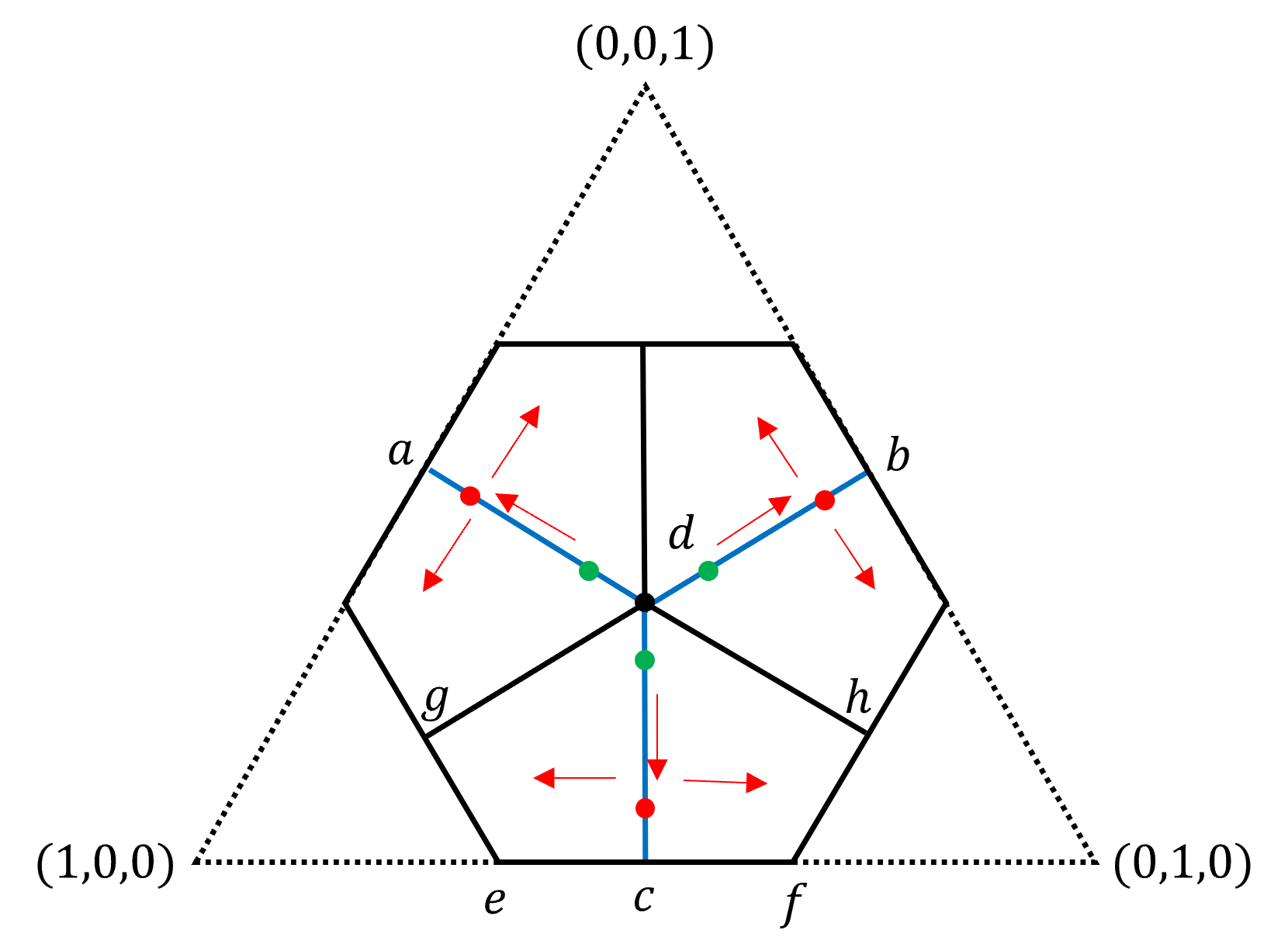}
\caption{Proof of Condition 3}\label{fig:condition3}
\end{figure}

For every $\theta$ on the line segment $cd$, denote $w(\theta)$ the width of the pentagon efhdg at $\theta$ (\Cref{fig:condition3-notation}). Similarly, define $w(\theta)$ for $\theta$ on the line segment $ad, bd$ to be the width of the corresponding pentagon at $\theta$. Let $\eta_2'$ be a measure that has a total mass of 2 on the line segments $ad, bd, cd$, such that the density of each point $\theta$ is proportional to $w(\theta)$. We claim that $\eta_2\cuxd \eta_2'$ by transforming $\eta_2'$ to $\eta_2$ by mean-preserving spreads. For each point $\theta$ on the line segment $cd$, we spread the mass of $\theta$ uniformly to all points $\theta'$ in the pentagon efhdg such that $\theta_3=\theta_3'$. This is a mean-preserving spread since by symmetry, the mean of those points is $\theta$. Moreover, since the density of $\theta$ in $\eta_2'$ is proportional to the width of the pentagon at $\theta$, all the mass are uniformly distributed in the pentagon. We perform similar operations for the other two pentagons. Thus we can transform $\eta_2'$ to $\eta_2$ by mean-preserving spreads.

\begin{figure}[ht]
	\centering
	\includegraphics[scale=0.6]{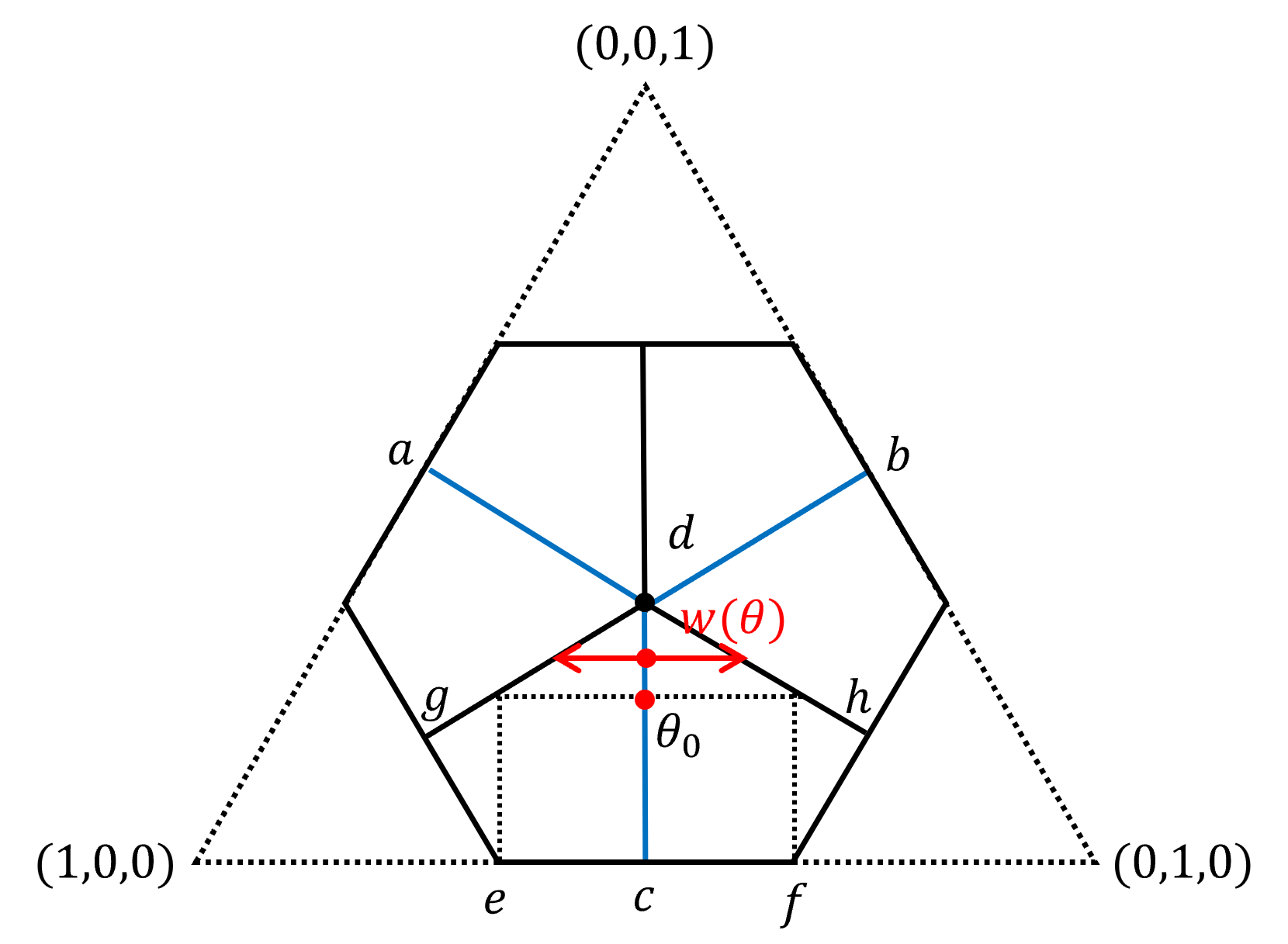}
\caption{Illustration of $w(\theta)$ and $\theta_0$ in the proof}\label{fig:condition3-notation}
\end{figure}

Now it suffices to prove that we can transform $\eta_1$ to $\eta_2'$ by mean-preserving spreads. We prove the following claim.

\begin{claim}\label{claim:uniform-condition-3}
For any $x,y\in \mathbb{R}_+$ such that $x<y$, let $\bm{\theta}_1,\bm{\theta}_2, \bm{\theta}_3$ be the three points on the line segments $ad, bd, cd$ respectively, such that the distance between each $\bm{\theta}_i$ and $d=(\frac{1}{3},\frac{1}{3},\frac{1}{3})$ is $x$ (green points in \Cref{fig:condition3}). Similarly, let $\bm{\theta}_1',\bm{\theta}_2', \bm{\theta}_3'$ be the points that have distance $y$ from point $d$ (red points in \Cref{fig:condition3}). Then any positive mass uniformly distributed in $\bm{\theta}_1,\bm{\theta}_2, \bm{\theta}_3$ can be transformed to the same amount of mass uniformly distributed in $\bm{\theta}_1',\bm{\theta}_2', \bm{\theta}_3'$, via a mean-preserving spread.
\end{claim}
\begin{proof}
For every $i\in \{1,2,3\}$, we notice that $\bm{\theta}_i$ is a convex combination of $d$ and $\bm{\theta}_i'$. Thus we can spread the mass at $\bm{\theta}_i$ to $d$ and $\bm{\theta}_i'$ for every $i\in \{1,2,3\}$. Next, since $d$ is a convex combination of $\bm{\theta}_1',\bm{\theta}_2', \bm{\theta}_3'$, we can then spread of mass at $d$ to $\bm{\theta}_1',\bm{\theta}_2', \bm{\theta}_3'$. By symmetry, the mass is uniformly distributed in those three points. 
\end{proof}

Back to the proof of \Cref{lem:verify-optimality}. The density of each point in $\eta_1$ is $\frac{2/3}{cd}=\frac{2\sqrt{6}}{3}$. For $\eta_2'$, by symmetry, the density of each point $\theta$ is $$\frac{2w(\theta)}{3\int w(\theta')d\theta'}=\frac{2/3\cdot w(\theta)}{\text{Area of pentagon efhdg}}=2\sqrt{3}\cdot w(\theta)$$

Let $\theta_0$ be the (unique) point in the interior of $cd$, such that $w(\theta_0)=\frac{\sqrt{2}}{3}$ (\Cref{fig:condition3-notation}).~\footnote{The point is unique since $\frac{\sqrt{2}}{3}=ef<gh$.} Then for any point $\theta$ in the line segment $d\theta_0$, the density of $\theta$ in $\eta_1$ is at least the density of $\theta$ in $\eta_2'$; For any point $\theta$ in the line segment $\theta_0c$, the density of $\theta$ in $\eta_1$ is at most the density of $\theta$ in $\eta_2'$. By \Cref{claim:uniform-condition-3}, we can transform $\eta_1$ to $\eta_2'$ via mean-preserving spreads, by keeping spreading the mass of points in $d\theta_0$ to points in $\theta_0c$. Thus we have $\eta_2'\cuxd \eta_1$. Combining with the fact that $\eta_2\cuxd \eta_2'$, we have $\eta_2\cuxd \eta_1$, which implies that $\gamma_1^*-\gamma_2^*\cuxd \mu^{P^*}$.

\notshow{
To simplify the proof, we map each point $(\theta_1,\ldots,\theta_{n-1})\in \Theta$ to $\theta\in \Delta_n=\{\theta\in [0,1]^n: \sum_{i=1}^n\theta_i=1\}$ with $\theta_n=1-\sum_{i=1}^{n-1}\theta_i$. Now fix any $i,j\in [n]$, $i<j$. Consider the region $R_{ij}=\Omega_1\cap \{\theta\in \Delta_n \mid \max\{\theta_i,\theta_j\}\geq \theta_k,\forall k\not=i,j\}$. Clearly $\cup_{i<j}R_{ij}=\Omega_1$ and each pair of $R_{ij}$s only overlap at boundaries. Thus $\eta_2$ has a total mass of $\frac{2}{n}$ uniformly distributed on $R_{ij}$. By definition of $R_{ij}$ and $S_{ij}$, $S_{ij}\subseteq R_{ij}$. We notice that there is an 1-1 mapping $h:R_{ij}\to R_{ij}$ by flipping $\theta_i$ and $\theta_j$ for every $\theta\in R_{ij}$. Thus $\frac{\theta+h(\theta)}{2}\in S_{ij}$. 

We perform the following operation on $\eta_2\big|_{R_{ij}}$: For every $\theta\in R_{ij}$, transform the mass from $\theta$ and $h(\theta)$ to their midpoint $\frac{\theta+h(\theta)}{2}$. Let $\eta^{(i,j)}$ be the obtained measure. Then $\eta_2\big|_{R_{ij}}$ is a mean-preserving spread of $\eta^{(i,j)}$, which implies that $\eta^{(i,j)}\cuxd \eta_2\big|_{R_{ij}}$. 

Denote $\theta_{-ij}=\{\theta_k\}_{k\not=i,j}$. Denote $|\theta_{-ij}|_{\infty}$ (or $|\theta_{-ij}|_{1}$) the infinity norm (or 1-norm) of $\theta_{-ij}$, i.e. $|\theta_{-ij}|_{\infty}=\max_{k\not=i,j}\theta_k$ and $|\theta_{-ij}|_1=\sum_{k\in [n], k\not=i,j}\theta_k$. We prove the following claim about measure $\eta^{(i,j)}$.

\begin{claim}
$\eta^{(i,j)}$ is a measure on $S_{ij}$ with $\eta^{(i,j)}(S_{ij})=\eta_2(R_{ij})=\frac{2}{n}$. For any $\theta_{-ij}$ such that $|\theta_{-ij}|_{\infty}\leq \frac{1-|\theta_{-ij}|_1}{2}$.~\footnote{This is exactly the region $\{\theta_{-ij}\in [0,1]^{n-2}\mid \exists \theta_i,\theta_j \text{ s.t. } (\theta_i,\theta_j,\theta_{-ij})\in S_{ij}\}$. }, let $w(\theta_{-ij})=\min\{\frac{1-|\theta_{-ij}|_1}{2}-|\theta_{-ij}|_{\infty}, \frac{1+|\theta_{-ij}|_1}{2}-p\}$. Then for every $\theta\in S_{ij}$, the density of $\eta^{(i,j)}$ at point $\theta$ is proportional to $w(\theta_{-ij})$. Formally, for every measurable set $A\subseteq S_{ij}$,
$$\eta^{(i,j)}(A)=\frac{2}{n}\cdot\frac{\int_{A}w(\theta_{-ij})d\theta}{\int_{S_{ij}}w(\theta_{-ij})d\theta}$$.
\end{claim}
\begin{proof}
Firstly, since $\frac{\theta+h(\theta)}{2}\in S_{ij}$ for every $\theta\in R_{ij}$, $\eta^{(i,j)}$ is a measure on $S_{ij}$. It remains to analyze the density of $\eta^{(i,j)}$ for every point.
Given any $\theta_{-ij}$ such that $|\theta_{-ij}|_{\infty}\leq \frac{1-|\theta_{-ij}|_1}{2}$. It corresponds to a unique point $\theta'=(\frac{1-|\theta_{-ij}|_1}{2},\frac{1-|\theta_{-ij}|_1}{2}, \theta_{-ij})\in S_{ij}$. 
Consider the set of $\theta_i$ such that $\theta_i\geq \theta_j$ and $(\theta_i,\theta_j, \theta_{-ij})\in R_{ij}$ (here $\theta_j=1-\theta_i-|\theta_{-ij}|_1$). This is the set where the mass at point $(\theta_i,\theta_j, \theta_{-ij})$ is transformed to $\theta'$ during the operation. By the definition of $R_{ij}$, we have $|\theta_{-ij}|_{\infty}\leq \theta_j\leq\theta_i\leq 1-p$. Thus
$\theta_i\in [\frac{1-|\theta_{-ij}|_1}{2}, \min\{1-p, 1-|\theta_{-ij}|_1-|\theta_{-ij}|_{\infty}\}]$. By the definition of $\eta^{(i,j)}$, the density at any point $\theta'\in S_{ij}$ is proportional to the length of the above line segment, which is equal to $w(\theta_{-ij})$. 
\end{proof}
}

\end{prevproof}